%% file: main_arxiv.tex
\documentclass[12pt]{article}

\usepackage[thmmarks]{ntheorem}
\usepackage{amsmath,amssymb}
\usepackage{graphicx}
\usepackage{booktabs,array,multirow}
\usepackage{xspace}
\usepackage{natbib}
\usepackage{multibib}
\usepackage[plain,noend]{algorithm2e}
\usepackage{fancyhdr}
\usepackage[colorlinks=true,linkcolor=blue,citecolor=blue,urlcolor=blue,
  pdftitle={Simulation-Based Inference with Convolutional Maximum Mean Discrepancy for Noisy Data},
  pdfauthor={Ritwik Vashistha, Jeff M. Phillips, Abhra Sarkar, Arya Farahi}]{hyperref}

\graphicspath{{./art/}}
\DeclareFontShape{OT1}{cmr}{m}{scit}{<->ssub*cmr/m/sc}{}
\DeclareFontShape{OT1}{cmr}{bx}{sc}{<->ssub*cmr/m/sc}{}

\DeclareTextFontCommand{\texttt}{\ttfamily\hyphenchar\font=`\-\relax}
\renewcommand{\baselinestretch}{1.5}
\newcommand{\MMD}{\ifmmode\text{\textup{\textsc{mmd}}}\else\textup{\textsc{mmd}}\expandafter\xspace\fi}
\newcommand{\convMMD}{\ifmmode\text{\textup{conv\textsc{mmd}}}\else\textup{conv\textsc{mmd}}\expandafter\xspace\fi}

\def\T{{ \mathrm{\scriptscriptstyle T} }}

\def\trans{^{\T}}
\newcommand\wh[1]{{\widehat{#1}}}
\newcommand\wt[1]{{\widetilde{#1}}}
\newcommand\bracket[1]{\left[{#1}\right]}

\newcommand{\E}{E}

\newcommand{\R}{\mathbb{R}}
\newcommand{\N}{\mathcal{N}}

\newcommand{\tX}{\widetilde{X}}

\newcommand{\bfC}{\mathbf{C}}
\newcommand{\bfG}{\mathbf{G}}
\newcommand{\bfH}{\mathbf{H}}
\newcommand{\bfJ}{\mathbf{J}}
\newcommand{\bft}{\mathbf{t}}

\newcommand{\bfU}{\mathbf{U}}
\newcommand{\bfv}{\mathbf{v}}
\newcommand{\bfV}{\mathbf{V}}

\newcommand{\bfY}{\mathbf{Y}}
\newcommand{\bfx}{\mathbf{x}}
\newcommand{\bfX}{\mathbf{X}}
\newcommand{\bfw}{\mathbf{w}}
\newcommand{\bfW}{\mathbf{W}}
\newcommand{\bfz}{\mathbf{z}}
\newcommand{\bfZ}{\mathbf{Z}}

\newcommand{\bfI}{\mathbf{I}}

\newcommand{\bfzero}{\boldsymbol{0}}
\newcommand{\bfmu}{\boldsymbol{\mu}}
\newcommand{\bfsigma}{\boldsymbol{\sigma}}
\newcommand{\bftheta}{\boldsymbol{\theta}}
\newcommand{\bfbeta}{\boldsymbol{\beta}}
\newcommand{\bfdelta}{\boldsymbol{\delta}}

\newcommand{\bfTheta}{\boldsymbol{\Theta}}
\newcommand{\bfphi}{\boldsymbol{\phi}}
\newcommand{\bfpsi}{\boldsymbol{\psi}}
\newcommand{\bfSigma}{\boldsymbol{\Sigma}}

\newcommand{\calF}{\mathcal{F}}
\newcommand{\calH}{\mathcal{H}}

\newcommand{\calN}{\mathcal{N}}

\newcommand{\calS}{\mathcal{S}}

\newcommand{\pr}{\mathrm{pr}}

\newcommand{\alttext}[1]{}
\makeatletter
\newsavebox{\figtempbox}
\newcommand{\figuresize}[1]{\gdef\@figtot{#1}}\figuresize{1}
\def\figurebox#1#2#3{\@ifnextchar[{\@figurebox{#1}{#2}{#3}}{\@figurebox{#1}{#2}{#3}[]}}
\def\@figurebox#1#2#3[#4]{%
  \sbox{\figtempbox}{\includegraphics[scale=\@figtot]{#4}}%
  \centerline{\includegraphics[width=\dimexpr1.1185\wd\figtempbox\relax]{#4}}}
\makeatother
\newsavebox{\tblbox}
\newcommand{\tbl}[2]{%
  \caption{#1}\centering
  \sbox{\tblbox}{\small\renewcommand{\baselinestretch}{1}\selectfont #2}%
  \ifdim\wd\tblbox>\textwidth\resizebox{\textwidth}{!}{\usebox{\tblbox}}\else\usebox{\tblbox}\fi\par}
\newenvironment{tabnote}{\par\vskip 6pt\small\renewcommand{\baselinestretch}{1}\selectfont}{\par}

\theoremheaderfont{\hskip\parindent\normalfont\scshape}
\theorembodyfont{\itshape}
\theoremseparator{.}
\theoremnumbering{arabic}
\newtheorem{theorem}{Theorem}
\newtheorem{lemma}{Lemma}

\newtheorem{proposition}{Proposition}

\theoremheaderfont{\hskip\parindent\normalfont\itshape}
\theorembodyfont{\rm}
\newtheorem{assumption}{Assumption}
\newtheorem{remark}{Remark}

\theoremheaderfont{\hskip\parindent\normalfont\itshape}
\theorembodyfont{\normalfont}
\newtheorem{example}{Example}

\font\QEDlogofont=msam10 at 10pt
\def\QEDlogo{\hbox{\QEDlogofont\char'003}}
\theoremsymbol{\QEDlogo}
\theoremheaderfont{\hskip\parindent\normalfont\itshape}
\theorembodyfont{\normalfont}
\theoremseparator{.}
\theoremstyle{nonumberplain}
\newtheorem{proof}{Proof}
\makeatletter
\renewtheoremstyle{nonumberplain}%
  {\item[\theorem@headerfont\hskip\labelsep ##1\theorem@separator]}%
  {\item[\theorem@headerfont\hskip \labelsep ##1\ ##3\theorem@separator]}
\makeatother

\makeatletter\let\newcites@later\newcites\makeatother

\makeatletter\AtEndDocument{\let\bibcite\@gobbletwo}\makeatother

\makeatletter
\newcommand{\saferef}[1]{\@ifundefined{r@#1}{?}{\expandafter\expandafter\expandafter\@car\csname r@#1\endcsname\@nil}}
\makeatother
\theoremstyle{plain}
\theoremsymbol{}
\theoremheaderfont{\hskip\parindent\normalfont\scshape}
\theorembodyfont{\itshape}
\theoremseparator{.}
\newtheorem{innercustomthm}{Theorem}
\newenvironment{customthm}[1]
  {\renewcommand\theinnercustomthm{#1}\let\ref\saferef\innercustomthm}{\endinnercustomthm}
\newtheorem{innercustomlem}{Lemma}
\newenvironment{customlem}[1]
  {\renewcommand\theinnercustomlem{#1}\let\ref\saferef\innercustomlem}{\endinnercustomlem}
\newtheorem{innercustomprop}{Proposition}
\newenvironment{customprop}[1]
  {\renewcommand\theinnercustomprop{#1}\let\ref\saferef\innercustomprop}{\endinnercustomprop}
\theoremheaderfont{\hskip\parindent\normalfont\itshape}
\theorembodyfont{\normalfont}
\newtheorem{innercustomex}{Example}
\newenvironment{customex}[1]
  {\renewcommand\theinnercustomex{#1}\let\ref\saferef\innercustomex}{\endinnercustomex}

\makeatletter
\newcommand{\suppcitations}{%
  \@for\@suppcite:=cite,citep,citet,citealp,citealt\do{%
    \expandafter\let\csname std@\@suppcite\expandafter\endcsname\csname\@suppcite\endcsname
    \expandafter\edef\csname\@suppcite\endcsname{%
      \noexpand\let\noexpand\@citex\noexpand\mb@@citex
      \noexpand\let\noexpand\@newciteauxhandle\noexpand\@auxoutsupp
      \expandafter\noexpand\csname std@\@suppcite\endcsname}}%
}
\makeatother

\newcommand{\arxivtitle}{Simulation-Based Inference with Convolutional Maximum Mean Discrepancy for Noisy Data}
\newcommand{\arxivauthors}{%
\begin{center}
  Ritwik Vashistha$^{1}$\\
  ritwikvashistha@austin.utexas.edu \\
  \vskip 4mm%
  Jeff M. Phillips$^{2}$\\
  jeffp@cs.utah.edu \\
  \vskip 4mm%
  Abhra Sarkar$^{1}$\\
  abhra.sarkar@utexas.edu \\
  \vskip 4mm%
  Arya Farahi$^{1}$ \\
  arya.farahi@austin.utexas.edu \\
  \vskip 4mm%
  $^1$Department of Statistics and Data Science,\\
  The University of Texas at Austin,\\
  105 E. 24th St., Stop D9800, Austin, Texas 78712, U.S.A.\\
  \vskip 2mm%
  $^2$Kahlert School of Computing,\\
  University of Utah,\\
  50 Central Campus Dr., Salt Lake City, Utah 84112, U.S.A.\\
\end{center}}

\makeatletter
\newcommand{\suppstart}{%
  \clearpage
  \let\bibcite@orig\bibcite
  \def\bibcite##1##2{\global\@namedef{b@##1\@extra@binfo}{##2}}
  \newcites@later{supp}{References}
  \let\bibcite\bibcite@orig
  \pagestyle{fancy}\fancyhf{}%
  \rhead{\bfseries\thepage}\lhead{\bfseries Supplementary Material}%
  \setcounter{page}{1}%
  \renewcommand\thepage{S.\arabic{page}}%
  \setcounter{section}{0}\setcounter{subsection}{0}%
  \setcounter{equation}{0}\setcounter{figure}{0}\setcounter{table}{0}\setcounter{algocf}{0}%
  \setcounter{theorem}{0}%
  \setcounter{lemma}{0}%
  \setcounter{corollary}{0}%
  \setcounter{proposition}{0}%
  \setcounter{definition}{0}%
  \setcounter{assumption}{0}%
  \setcounter{remark}{0}%
  \setcounter{step}{0}%
  \setcounter{condition}{0}%
  \setcounter{property}{0}%
  \setcounter{restrictions}{0}%
  \setcounter{example}{0}%
  \setcounter{algo}{0}%
  \setcounter{section}{0}\setcounter{subsection}{0}%
  \setcounter{equation}{0}\setcounter{figure}{0}\setcounter{table}{0}\setcounter{algocf}{0}%
  \setcounter{theorem}{0}%
  \setcounter{lemma}{0}%
  \setcounter{corollary}{0}%
  \setcounter{proposition}{0}%
  \setcounter{definition}{0}%
  \setcounter{assumption}{0}%
  \setcounter{remark}{0}%
  \setcounter{step}{0}%
  \setcounter{condition}{0}%
  \setcounter{property}{0}%
  \setcounter{restrictions}{0}%
  \setcounter{example}{0}%
  \setcounter{algo}{0}%
  \renewcommand\thesection{S\arabic{section}}%
  \renewcommand\thetable{S\@arabic\c@table}%
  \renewcommand\thefigure{S\@arabic\c@figure}%
  \renewcommand\theequation{S\arabic{equation}}%
  \renewcommand{\thealgocf}{S\arabic{algocf}}%
  \renewcommand\thetheorem{\text{S}\arabic{theorem}}%
  \renewcommand\thelemma{\text{S}\arabic{lemma}}%
  \renewcommand\thecorollary{\text{S}\arabic{corollary}}%
  \renewcommand\theproposition{\text{S}\arabic{proposition}}%
  \renewcommand\thedefinition{\text{S}\arabic{definition}}%
  \renewcommand\theassumption{\text{S}\arabic{assumption}}%
  \renewcommand\theremark{\text{S}\arabic{remark}}%
  \renewcommand\thestep{\text{S}\arabic{step}}%
  \renewcommand\thecondition{\text{S}\arabic{condition}}%
  \renewcommand\theproperty{\text{S}\arabic{property}}%
  \renewcommand\therestrictions{\text{S}\arabic{restrictions}}%
  \renewcommand\theexample{\text{S}\arabic{example}}%
  \renewcommand\thealgo{\text{S}\arabic{algo}}%
  \allowdisplaybreaks
  \suppcitations
  \renewcommand{\baselinestretch}{1.1}\selectfont
  \begin{center}
  {\LARGE Supplementary Material for\\ {\bf \arxivtitle}}
  \end{center}
  \vskip 2mm
  {\baselineskip=14pt\arxivauthors}
  \vskip 5mm
}
\makeatother

\begin{document}
\thispagestyle{empty}
\baselineskip=28pt
\begin{center}
{\LARGE{\bf \arxivtitle}}
\end{center}
\baselineskip=12pt
\vskip 2mm
\arxivauthors

\begin{abstract}
\baselineskip=12pt
Ignoring measurement noise degrades statistical inference. Existing correction techniques are often computationally costly and inefficient. Kernel methods based on maximum mean discrepancy have enabled efficient, likelihood-free inference, but they typically assume clean data. We introduce a framework for inference with samples corrupted by potentially heteroscedastic noise from a known distribution. Central to our method is the convolutional maximum mean discrepancy, which compares distributions after noise convolution and retains metric validity under standard kernel conditions. We establish finite-sample deviation bounds that are unaffected by measurement error and an equivalence between testing under noise and kernel smoothing. Based on this discrepancy, we propose an estimator, establish its consistency and asymptotic normality, and provide an efficient implementation using stochastic gradient descent. We demonstrate the practical effectiveness of our method through simulations and applications in astronomy, anthropometry, and the social sciences. Our method provides a unified, likelihood-free, simulation-based framework for efficient and robust inference for broad classes of problems in noisy data.
\end{abstract}

\vskip 20pt
\baselineskip=12pt
\noindent\underline{\bf Key Words}: Deconvolution; Heteroscedastic noise; Kernel method; Maximum mean discrepancy; Measurement error; Parameter estimation; Regression with errors-in-variables.

\clearpage
\pagenumbering{arabic}
\normalsize

\section{Introduction}

Measurement error is a ubiquitous challenge in modern statistics, complicating data analysis across domains as diverse as biology \citep{leek2010tackling}, environment \citep{gryparis2009measurement}, economics \citep{jerzak2025attenuation},  epidemiology \citep{shaw2020stratos}, 
and astronomy \citep{luri2018gaia}. 
In many cases, practitioners may  possess partial or detailed knowledge of the underlying noise laws, quantified through calibration, replication, or physical modelling.
Neglecting the presence and nature of measurement error during data analysis can lead to   biased estimation, inflated variance, and loss of inferential power \citep{prentice1982covariate,rosner1990correction,li2024sparse}, ultimately undermining the validity of scientific conclusions. However, incorporating the information about measurement error into the inference framework can be surprisingly hard. 

The distribution of the observed data is the convolution of the latent distribution with the noise law, so the likelihood is an integral over the unobserved truth
that is intractable for all but the simplest models.  A vast literature has focused on addressing this challenge \citep{gustafson2003measurement,carroll2006measurement,fuller2009measurement,buonaccorsi2010measurement}, but the main existing methods each give up something to tackle the intractable integral. Fourier deconvolution \citep{carroll1988optimal,stefanski1990deconvolving,fan1991optimal} divides by the noise characteristic function but it is widely  known to be numerically unstable. Methods such as  \textsc{simex} \citep{cook1994simulation,carroll1999nonparametric}, extreme deconvolution \citep{bovy2011extreme} and \texttt{linmix} \citep{kelly2007some} all typically assume Gaussian noise of known variance. Bayesian hierarchical models \citep{berry2002bayesian,sarkar2014bayesian,mantz2016gibbs,sarkar2021bayesian} are flexible but computationally expensive, and they often lack asymptotic guarantees in complex hierarchical settings. 

\textit{Our idea: Do not invert the noise.} We propose to assume a latent model, sample from it and add noise from the known noise law to the ``clean'' samples to obtain the model-generated distribution of the observed data. To estimate the parameters of the latent model, we minimize the distance between the model-generated distribution and the empirical distribution of the observed data. We measure the
distance between the observed and simulated noisy data with the maximum mean discrepancy \citep[\MMD,][]{gretton12a}, which is computed only using samples, and we call the resulting distance the convolutional \MMD (\convMMD). Here, deconvolution happens implicitly, as a consequence of the optimization rather than as an explicit inversion. This simple idea allows us to estimate a wide range of parametric models in the presence of non-Gaussian, heteroscedastic, hierarchical, or discrete noise. Additionally, our method handles deconvolution and regression by treating them as a broad parametric problem; in regression, it accommodates measurement error in both the covariates and the response. To our knowledge, no other existing method in the literature solves all these different problems within a single framework. 

\textit{Theoretical properties.}
We provide the following theoretical results for our proposed method. 
\begin{enumerate}
    \item Noise is equivalent to kernel smoothing. We show that for translation-invariant kernels, \convMMD between the noisy distributions equals an ordinary \MMD between the clean distributions under a noise-smoothed kernel $\wt{k}$ (Theorem~\ref{th: mmd-equivalence}). For example, in the case of Gaussian noise with standard deviation $\tau$ and a Gaussian kernel with bandwidth $l$, we show that $\wt{k}$ is a Gaussian kernel with bandwidth $(l^2+2\tau^2)^{1/2}$.
    \item Identifiability under weak conditions. \convMMD is a valid metric on the latent distributions whenever the zeros of the characteristic function of the noise distribution form a Lebesgue-null set (Theorem~\ref{th: sufficient condition}). This is strictly weaker than requiring the smoothed kernel to be strictly positive definite, and it can admit noise laws, such as uniform rounding error, that a positive-definiteness argument would exclude.
    \item Finite-sample bounds unaffected by noise. The large-deviation bound on the estimation error of \convMMD is the same as in the noise-free case, so it depends on the sample size but not on the noise (Theorem~\ref{th:hypothesis_testing}). The effect of noise appears instead as a variance inflation of the unbiased \convMMD U-statistic when $p=q$, which we bound explicitly in terms of the expected noise variance (Theorem~\ref{th: mmd-var}).
    \item Root-$N$ parametric inference. The minimum-\convMMD estimator is consistent, satisfies a generalization bound, and is asymptotically normal at the $N^{1/2}$ rate with a sandwich covariance (Theorems~\ref{th: bound and consistency} and~\ref{th:clt}). For the special case of a Gaussian location model, we obtain the asymptotic covariance in closed form (Example~\ref{ex: clt-illust}) and show the trade-off between data variance, noise variance, and kernel bandwidth.
\end{enumerate}

\textit{Algorithm and empirical results.}
We propose a practical, efficient algorithm that uses stochastic gradient descent and an unbiased score-function gradient (Proposition~\ref{prop: score-gradient}) to estimate any parametric model with a differentiable density that can be simulated from. 
Across simulations, we show that \convMMD performs comparably to likelihood-based methods in the presence of Gaussian noise. When the Gaussian noise assumption fails, \convMMD generally outperforms the competitor methods. We also illustrate \ convMMD's performance across three real data applications in astronomy, anthropometry, and the social sciences. 

\textit{Contributions.} In summary, we (i) introduce \convMMD, a discrepancy that incorporates a known noise distribution directly into the statistical metric; (ii) prove that noise is equivalent to kernel smoothing and quantify its effect on the concentration and variance of the empirical statistic; (iii) develop a simulation-based minimum-\convMMD estimator with consistency, a generalization bound,
and a central limit theorem, including a closed-form asymptotic covariance in the Gaussian case; (iv) demonstrate that the resulting method is robust to non-Gaussian and heteroscedastic noise and to outliers; and (v) build a flexible unified framework that can solve different problems in the measurement error literature, such as deconvolution and regression with noise in both covariates and responses.

\textit{Related work.}
Measurement error research broadly distinguishes between two types of problems. The first, often termed epsilon-contamination, is the focus of study in robust statistics \citep{huber1964robust}, where a small fraction of the data is assumed to be contaminated. 
Methods in this domain, such as M-estimation, are designed to be insensitive to these outliers. 
The second paradigm, and the focus of our work, is the classical measurement error model, which assumes that all observations are contaminated by an additive noise process arising from imperfect instruments or measurement procedures. In this framework, the observed variable $\wt{W}$ is defined as the sum of an unobservable ``true'' latent variable $W$ and an independent noise component $U$. 

Beyond the literature cited above, Fourier and Bayesian methods have been extended to 
regression problems with error-contaminated covariates \citep[e.g.,][]{fan1993nonparametric,dellaportas1995bayesian,huang2017alternative}, inheriting the same stability limitation.
For noise-free data, \MMD has been used for testing \citep{gretton12a}, model criticism \citep{lloyd2015statistical}, generative modelling \citep{dziugaite2015training}, and minimum-distance estimation, where it is known to be robust to outliers and misspecification \citep{briol2019statistical,cherief2022finite,oates2022minimum,alquier2023estimation,alquier2024universal}. 
More recently, \MMD has also emerged as a robust tool for handling measurement errors. 
\cite{dellaporta2022robust, chen2025total} use it as a loss function within efficient Bayesian nonparametric learning frameworks. Alternatively, \cite{yi2024denoising} frames error correction as a generative task, 
employing \MMD for evaluation rather than as a training objective. However, these approaches do not incorporate the known noise law into the discrepancy itself. Our approach addresses this gap while providing frequentist asymptotic theory for estimation from noisy data, extending the noise-free asymptotics of \citet{briol2019statistical} and \citet{alquier2023estimation} to observations contaminated by known, possibly heteroscedastic, measurement error.

\textit{Scope.}  To fix our main ideas concretely, consistent with precedents in the classical literature, this initial work focuses on the known error distribution case, with a correctly specified latent model. 
Relaxing both assumptions is ongoing work. 
Proofs and implementation details are in the Supplementary Material.

\section{Method} 

\subsection{Background}

We compare $p$ and $q$ only after both have been convolved with the known noise law $m$, so that
the noise enters the distance itself rather than being corrected for separately. We build this up
in stages: first recalling \MMD for the clean variables $\bfW$ and $\bfZ$, together with the kernel
expansion that makes it estimable directly from samples, then setting up the measurement error
model relating $\bfW$ and $\bfZ$ to their observed counterparts, and defining \convMMD and its
estimators before establishing its metric and finite-sample properties. From there, the same
distance serves as a loss function for parametric estimation from noisy data, and we characterize
the asymptotics of the resulting estimator, which retains the $N^{1/2}$ rate despite the noise.

Let $\bfW, \bfZ \in \mathcal{S} \subseteq \mathbb{R}^d$ be random variables with probability distributions $p$ and $q$, respectively, with realizations $\bfw, \bfz \in \mathcal{S}$. 
Let $\calH$ be a reproducing kernel Hilbert space (\textsc{rkhs}) 
defined on $\calS$ associated with a kernel $k(\cdot,\cdot)$ and equipped with a norm $\|\cdot\|_{\calH}$ \citep{berlinet2011reproducing,muandet2017kernel}. 
Also, let $\calF = \{f \in \calH, \|f\|_{\calH} \leq 1\}$. 
\MMD is then defined as 
\begin{equation*}
    \MMD(p,q) = \underset{f \in \calF}{\sup}~|\E_p(f(\bfW)) - \E_q(f(\bfZ))|. 
\end{equation*}
This metric quantifies the dissimilarity between two distributions in terms of the distance between their respective mean embeddings in the \textsc{rkhs}. 
For characteristic kernels~\citep{sriperumbudur2011universality}, it has been shown that $\MMD(p,q) = 0$ if and only if $p = q$, thus giving a consistent test \citep{gretton12a}.
The \MMD statistic can also be expressed using the associated \textsc{rkhs} kernel, providing a rich structure for analysis and computation.  Specifically, a well-known direct computational form for $\MMD^2$ is 
\begin{equation*}
    \MMD^2(p,q) = \E_{{\bfW}, {\bfW}'}\left[k\left({\bfW}, {\bfW}'\right)\right]+\E_{{\bfZ}, {\bfZ}'}\left[k\left({\bfZ}, {\bfZ}'\right)\right] -2 \E_{{\bfW}, {\bfZ}}\left[k\left({\bfW}, {\bfZ}\right)\right],  \nonumber 
\end{equation*}
where $\bfW,\bfW' \overset{iid}{\sim} p$ and $\bfZ,\bfZ' \overset{iid}{\sim} q$. 
This allows for the direct estimation of $\MMD^2$ from empirical samples, bypassing the need for density estimation.


\subsection{MMD with noisy data} 

Let $\{\bfW_{i}\}_{i=1}^{N}$ and $\{\bfZ_{i}\}_{i=1}^{N}$ denote sequences of random variables drawn from two probability distributions $p$ and $q$, respectively (while our framework supports sequences of unequal length, we focus on the balanced case for ease of presentation). 
We wish to compare samples drawn from $p$ and $q$. 
We address the challenge where the corresponding observations $\wt{\bfW}_{i}$ and $\wt{\bfZ}_{i}$ are contaminated
with i.i.d. noise $\bfU_{W,i}$ and $\bfU_{Z,i}$ as  
\begin{equation*}
        \wt{\bfW}_{i} = \bfW_{i} + \bfU_{W,i}, \,\, \wt{\bfZ}_i = \bfZ_{i} + \bfU_{Z,i}, \quad  \text{with } \,\,  \bfU_{W,i}, \bfU_{Z,i} \sim r(\cdot \mid \bfphi_{i}), \text{ } \bfphi_{i} \sim g(\cdot \mid \bfpsi).  
\end{equation*}
Specifically, for our inference problems, $p$ and $q$ are the data-generating and model-specified distributions, respectively.  $\{\wt{\bfW}_{i}\}_{i=1}^{N}$ and $\{\wt{\bfZ}_{i}\}_{i=1}^{N}$ are observed and model-generated samples, respectively. 
For the numerical illustrations and applications presented later in this paper, we define $\wt{\bfW} \equiv \wt{\bfX} \in \mathbb{R}^{d_x}$ for density deconvolution, and $\wt{\bfW} \equiv (\wt{\bfX}\trans, \wt{\bfY}\trans)\trans \in \mathbb{R}^{d_x+d_y}$ for regression; similarly for $\wt{\bfZ}$. 
To prioritize clarity in demonstrating our core ideas, we focus on parsimonious settings: $d_x=1$ for deconvolution, and $d_y=1$ with $d_x \in \{1, 2\}$ for regression.
We let the parameter $\bfphi_{i}$ itself be random, drawn from a distribution $g(\cdot \mid \bfpsi)$, 
where $r(\cdot \mid \bfphi)$ and $g(\cdot \mid \bfpsi)$ are known.  This assumption is motivated by practical applications where noise can be characterized \citep{creevey2013large,leung2019simultaneous,anbajagane2025decade}.  The variation in $\bfphi_{i}$ allows the noise to be heteroscedastic even if the marginal distribution $m(\cdot) = \int r(\cdot \mid \bfphi) g(\bfphi \mid \bfpsi) d\bfphi$ remains identical across all $i$. 


\begin{assumption}[\textsc{iid} samples]\label{assump:iid} The sequences $\{\bfW_{i}\}_{i=1}^{N}$ and $\{\bfZ_{i}\}_{i=1}^{N}$ each consist of $d$-dimensional i.i.d.\ random variables from distributions $p$ and $q$, respectively. 
\end{assumption}

\begin{assumption}[Noise independence]\label{assump:ind} For each $i$, the noise variables $\bfU_{W,i}$, $\bfU_{Z,i}$ are independent of $\bfW_{i}$ and $\bfZ_{i}$ respectively: $\bfW_{i} \perp \!\!\!\perp\, \bfU_{W,i}$ and $\bfZ_{i} \perp \!\!\!\perp\, \bfU_{Z,i}$. 
\end{assumption}

\begin{assumption}[Known noise model]\label{assump:noise-model} The conditional noise distribution $r(\cdot \mid \bfphi)$ and [the distribution] $g(\bfphi \mid \bfpsi)$ are known.  
\end{assumption}

This last assumption is standard in astronomy \citep{kuhn2019kinematics, shah2020determination}, 
where data are typically obtained via physical sensors and the errors associated with those measurements can be carefully quantified and calibrated \citep[e.g.,][]{chen2019calibration,leung2019simultaneous,shah2020determination,marshall2021concordance,anbajagane2025decade}.  Data catalogues are accompanied by object-specific, heteroscedastic uncertainties \citep{sevilla2021dark,gebhardt2021hobby,anbajagane2025decade}.
The assumption is not specific to astronomy. In economics, survey responses can be validated against administrative records, which quantifies the error in the survey variables \citep{bound1991extent,thurow2025characterizing}
See \citet{carroll2006measurement} for additional examples. Differential privacy can be another application, where a random noise following a prescribed distribution \citep{dwork2006calibrating}, such as a Laplace or Gaussian distribution, is added to a dataset to guarantee privacy.

\begin{assumption}\label{assump:finite-moment-noise}
    For each $i$, the noise variables $\bfU_{W,i}$ and $\bfU_{Z,i}$ have finite second moments, conditioned on $\bfphi_{i}$.  
    Specifically, 
        $\E[\|\bfU_{W,i}\|^2\!\mid\!\bfphi_{i}]\!=\!\E[\|\bfU_{Z,i}\|^2 \mid \bfphi_{i}]\!=\!\alpha(\bfphi_{i})$,
        where  $0\!\le\!\alpha(\bfphi_{i})\!<\!\infty$, 
    and $\alpha(\bfphi)$ is  
    (i) measurable and (ii) integrable w.r.t.\ $g(\cdot \mid \bfpsi)$. 
\end{assumption}

\begin{assumption}[Convolution invertibility]\label{assump:convolution-invertibility}
    The marginal noise distribution $m(\cdot)$ is assumed to be \emph{convolution invertible}, meaning that the set of zeros of its characteristic function, $\mathcal{T}_{0} = \{\bft \in \mathbb{R}^{d}:  \varphi_{m}(\bft) = 0\}$, has Lebesgue measure zero. 
\end{assumption}

Assumptions~\ref{assump:finite-moment-noise} and \ref{assump:convolution-invertibility} are fundamental for the identifiability of our method. Additionally, Assumption~\ref{assump:finite-moment-noise} guarantees that the noise in the signal or data is finite, enabling our method to make the distinction between $p$ and $q$; and Assumption~\ref{assump:convolution-invertibility} permits a unique inversion of the convolution, allowing the true distribution to be reconstructed from noisy observations.
Practical applications frequently assume normal, log-normal, or uniform distributions for continuous errors \citep[e.g.,][]{wimmer2000proper,farahi2019detection, mulroy2019locuss,chen2019calibration,glazer2026beyond} 
all of which satisfy our theoretical requirements.  In the following we will use the shorthand notation $ \E_{p * m} [f(\wt{\bfW})] := \E_{\bfW \sim p, \bfU_{W} \sim m} [f(\wt{\bfW})]$ and $ \E_{q * m} [f(\wt{\bfZ})] := \E_{\bfZ \sim q, \bfU_{Z} \sim m} [f(\wt{\bfZ})]$, where $*$ denotes the convolution operator. Here, it is implied that the distribution 
$m(\cdot)$ under the expectation is a mixture over the parameter space $ \bfphi \in \Phi$, that is, $m(\cdot) = \int r(\cdot\mid \bfphi)g(\bfphi \mid \bfpsi)d\bfphi$.

\begin{lemma}[cancellation property] \label{lemma: identifiability}
Let $p$, $q$ be Borel probability measures defined on $\mathbb{R}^{d}$ (with supports possibly restricted to a subset $\calS \subseteq \mathbb{R}^{d}$). Let $m$ be a Borel probability measure defined on $\mathbb{R}^{d}$ where
$m(\cdot)\!=\!\int r(\cdot\!\mid\!\bfphi)g(\bfphi\!\mid\!\bfpsi) d\bfphi$, where $r(\cdot\!\mid\!\bfphi)$  is a family of Borel probability measures parametrized by a random variable  $\bfphi\!\in\!\Phi$ and $g(\bfphi\!\mid\!\bfpsi)$ is a probability measure on the parameter space $\Phi$. 
Then, $p\!=\!q$ if and only if $\E_{p * m} [f(\wt{\bfW})] p\!=\!\E_{q * m} [f(\wt{\bfZ})]$ for all $f \in C(\mathbb{R}^{d})$, where $C(\mathbb{R}^{d})$ is the space of bounded continuous functions.
\end{lemma}

\begin{remark}
If the characteristic function of noise $\varphi_{m}(\bft)$ vanishes only on a set $\mathcal{T}_{0} = \{\bft \in \mathbb{R}^{d}:  \varphi_{m}(\bft) = 0\}$ that has Lebesgue measure zero,  then the cancellation argument holds.
\end{remark}

\begin{example}\label{ex: uniform-ex} Consider a univariate setting where noise is drawn from the uniform distribution on $[-0.5,0.5]$. The characteristic function of the uniform distribution is given by
$\varphi_{\text{Uniform}}(t) = \int_{-0.5}^{0.5} e^{i t u}\, du = 2 \,\sin(t/2)/t$,
with the understanding that at $t=0$ we define $\varphi_{\text{Uniform}}(0)=1$ by continuity. This function vanishes when the numerator $e^{it}-1=0$, which happens when $e^{it}=1$, that is, at $t = 2\pi k$, $k\in\mathbb{Z}\setminus\{0\}$.
Excluding the case $t=0$, the zeros are at $t = 2\pi k$ for nonzero integers $k$. This set of zeros has Lebesgue measure zero in $\mathbb{R}$. Thus, the uniform distribution satisfies the condition that its characteristic function vanishes only on a discrete set, guaranteeing that the cancellation property remains valid.
\end{example}

A classical use is modelling rounding or discretization error as a uniform distribution over half the quantization interval  \citep{wimmer2000proper,glazer2026beyond}. For example, in rounding to the nearest integer, the error may be modelled as 
$U \sim \textrm{Uniform}(-0.5,+0.5)$.

Lemma \ref{lemma: identifiability} allows us to uniquely identify $p$ and $q$ in the presence of noise if and only if $\E_{p * m} [f(\wt{\bfW})] = \E_{q * m} [f(\wt{\bfZ})]$ for all $f \in C(\mathbb{R}^{d})$, where $C(\mathbb{R}^{d})$ is the space of bounded continuous functions in $\mathbb{R}^{d}$. However, it is impractical to work with such an extensive space $C(\mathbb{R}^{d})$ in a finite sample setting. Following \cite{gretton12a}, we consider $C(\mathbb{R}^{d})$ to be $\calF$, a unit ball within an \textsc{rkhs} $\calH$, associated with a positive definite kernel function $k(\cdot, \cdot)$. We then introduce the following (convolution) \MMD statistic for analysis
\begin{equation*}
\convMMD(p,q,m)  \equiv \MMD(p*m,q*m) = \underset{f \in \calF}{\sup}~|\E_{p*m}(f(\bfW)) - \E_{q*m}(f(\bfZ))|.     
\end{equation*}
The selection of $\calF$ allows the application of the kernel mean embedding approach, effectively mapping probability distributions into the elements of $\calH$. Suppose, $k(\cdot, \cdot)$ is measurable and $\E_{p*m} \{k(\wt{\bfW}, \wt{\bfW})\}^{1/2}<\infty$, $\E_{{q*m}} \{k(\wt{\bfZ}, \wt{\bfZ})\}^{1/2}<\infty$, for existence of mean embeddings. Using \cite{gretton12a}, we can rewrite $\MMD(p*m,q*m)$ as 
 $\convMMD(p,q,m)=\left\| \mu_{p*m}-\mu_{q*m}\right\|_{\calH}$. 
Here, $\mu_{p*m}$ and $\mu_{q*m}$ are the mean embedding corresponding to the convolved distributions $p*m$ and $q*m$, respectively. By expanding the squared norm in \textsc{rkhs} and using the kernel trick, we derive the following expression 
\begin{equation*}
 \convMMD^2(p,q,m)  
= \E_{\wt{\bfW},\wt{\bfW}'}\left[k\left(\wt{\bfW}, \wt{\bfW}'\right)\right]+\E_{\wt{\bfZ}, \wt{\bfZ}'}\left[k\left(\wt{\bfZ}, \wt{\bfZ}'\right)\right] -2 \E_{\wt{\bfW}, \wt{\bfZ}}\left[k\left(\wt{\bfW}, \wt{\bfZ}\right)\right].  
\end{equation*}
In practice, a (biased) estimator of this quantity based on the empirical distributions of the noisy samples $(\wt{\bfW},\wt{\bfZ})$, $\wh{(p*m)}_N = N^{-1}\sum_{i=1}^{N}\delta_{\wt{\bfw}_{i}}$ and $\wh{(q*m)}_N = N^{-1}\sum_{i=1}^{N}\delta_{\wt{\bfz}_{i}}$, is

\begin{align*}
  \wh{\convMMD^2_b}(p,q,m)  & \equiv \MMD^2_{b}(\wh{(p*m)}_{N},\wh{(q*m)}_{N}) \\ &= \frac{1}{N^2}\left(\sum_{i=1}^{N} \sum_{j=1}^N k\left(\wt{\bfw}_{i}, \wt{\bfw}_{j}\right)+\sum_{i=1}^{N} \sum_{j=1}^N k\left(\wt{\bfz}_{i}, \wt{\bfz}_{j}\right)-2 \sum_{i=1}^{N} \sum_{j=1}^N k\left(\wt{\bfw}_{i}, \wt{\bfz}_{j}\right)\right). 
\end{align*}
For the (unsquared) biased estimator used in Theorem~\ref{th:hypothesis_testing}, we write $\wh{\convMMD_{b}}(p,q,m) \equiv \{\wh{\convMMD^2_{b}}(p,q,m)\}^{1/2}$.
Additionally, an (unbiased) estimator can be defined as follows 
\begin{equation*}
\wh{\convMMD^2_u}(p,q,m) = \frac{1}{N(N-1)}\sum_{i \neq j}  h(\wt{\bfv}_{i},\wt{\bfv}_{j}),     
\end{equation*}
where $\wt{\bfv}_{i}: (\wt{\bfw}_{i},\wt{\bfz}_{i}) \sim p*m \times q*m$ and the function $h(\wt{\bfv}_{i},\wt{\bfv}_{j})$ is identified as the one-sample U-statistic defined by $h(\wt{\bfv}_{i},\wt{\bfv}_{j}) = k\left(\wt{\bfw}_{i}, \wt{\bfw}_{j}\right) + k\left(\wt{\bfz}_{i}, \wt{\bfz}_{j}\right) - k\left(\wt{\bfw}_{i}, \wt{\bfz}_{j}\right) - k\left(\wt{\bfz}_{i}, \wt{\bfw}_{j}\right)$. Evaluated at the noise-free pairs $\bfv_{i}=(\bfw_{i},\bfz_{i}) \sim p \times q$, the same expression, $\wh{\MMD^2_u}(p,q) = \{N(N-1)\}^{-1}\sum_{i \neq j} h(\bfv_{i},\bfv_{j})$, is the usual unbiased estimator of $\MMD^2(p,q)$.
Both estimators involve $\mathcal{O}(N^2)$ kernel evaluations, which can be computationally expensive, but can be addressed by faster approximations; see \cite{joshi2011comparing,chen2017relative,phillips2020gaussiansketch} for more discussion.

\subsection{Properties}

We first show that, under certain conditions, $\convMMD(p,q,m)$ is a valid metric: $\convMMD(p,q,m) = 0$ if and only if $p = q$. 
Later, we look at the relationship between  $\convMMD(p,q,m)$ and $\MMD(p,q)$. 
Finally, we discuss a bound for the estimation error of $\convMMD(p,q,m)$. 
Throughout, we make the following assumption.
\begin{assumption} \label{assump: mmd-assump}
Let $\calS$, the support of p and q, be a subset of $\mathbb{R}^{d}$. 
Let $k: \mathbb{R}^{d} \times \mathbb{R}^{d} \to \mathbb{R}$ be a bounded, measurable, characteristic kernel defined on the Euclidean space, such that $0 \leq k(\bfw,\bfz) \leq K$ for all $\bfw,\bfz \in \mathbb{R}^{d}$. 
The \textsc{rkhs} $\calH$ is defined by this kernel on $\mathbb{R}^{d}$. We assume $p$, $q$, and $m$ are Borel probability measures on $\mathbb{R}^{d}$.
\end{assumption}
Assumption \ref{assump: mmd-assump} provides the foundation for \MMD to work as a valid metric for probability distributions. 
For instance, a characteristic \textsc{rkhs} ensures that the kernel is rich enough to distinguish between two distinct probability distributions. 
With this assumption, we have the following results. 
\begin{theorem} \label{th: sufficient condition} 
Under Assumptions \ref{assump:iid}--\ref{assump: mmd-assump},  $\convMMD(p,q,m) = 0$ if and only if $p = q$.
\end{theorem} 

Theorem \ref{th: sufficient condition} shows us that we can recover $p$ and $q$ using $\convMMD^2(p,q,m)$. The proof relies mainly on the noise distribution $m$ being independent of $p$ and $q$. 
In kernel hypothesis testing, it is often assumed that the kernel must be strictly positive definite to be characteristic. If one were to simply define an effective ``noisy kernel'' $k_m = k*m$ (as with the uniform noise in Example \ref{ex: uniform-ex}), one might incorrectly discard it because some noise distributions are only positive \textit{semi}-definite (having zeros in their spectrum). 
However, our formulation clarifies that strict positive definiteness of the noise is not required for identifiability. As seen in Example \ref{ex: uniform-ex}, the uniform distribution has spectral zeros, yet they are isolated. Theorem \ref{th: sufficient condition} proves that as long as these zeros do not accumulate, the metric property holds. This relies on Assumption \ref{assump:convolution-invertibility} on ``Convolution invertibility'' not the stricter requirement of a positive definite noise distribution. The following theorem formalizes this, establishing a relationship between $\convMMD^2(p,q,m)$ and $\MMD^2(p,q)$.

\begin{theorem}\label{th: mmd-equivalence}
Let the kernel $k: \mathbb{R}^{d} \times \mathbb{R}^{d} \to \mathbb{R}$ be characteristic and translation-invariant such that $k(\bfw, \bfz) = \kappa(\bfw-\bfz)$ for some function $\kappa$. Under Assumptions \ref{assump:iid}--\ref{assump: mmd-assump}, the \MMD between the noise-convoluted distributions $p*m$ and $q*m$ with respect to $k$ is equal to the \MMD between the true, noiseless distributions $p$ and $q$ with respect to a modified kernel $\wt{k}$: $\convMMD_k(p, q, m) = \MMD_{\wt{k}}(p, q)$, 
where the modified kernel $\wt{k}$ is defined as 
$\wt{k}(\bfw, \bfz) = \E_{\bfU_{W}, \bfU_{Z} \sim m}[k(\bfw+\bfU_{W}, \bfz+\bfU_{Z})]$,
where $\bfU_{W}$ and $\bfU_{Z}$ are i.i.d. random variables drawn from the noise distribution $m$.
\end{theorem}

Theorem \ref{th: mmd-equivalence} establishes that noise contamination is mathematically equivalent to a noiseless \MMD comparison using a smoother, convolved kernel. 
Effectively, the kernel absorbs the noise, widening its bandwidth. 
We illustrate this below with a Gaussian kernel and Gaussian noise.

\begin{example} \label{ex: mmd-equivalence}
Let $k(w,z)\!=\!\exp\left\{-(w-z)^2/(2l^2)\right\}$ and the marginal noise distribution be $m(\cdot)\!=\!\mathcal{N}(0, \tau^{2})$.
The modified kernel $\wt{k}(w, z) = \E_{U, U' \sim \mathcal{N}(0, \tau^{2})}\left(\exp\left[-\{(w-z) + (U-U')\}^2/(2l^2)\right]\right)$ is the expectation over the noise.
The difference of two i.i.d. Gaussians, $\delta = U-U'$, is also Gaussian with zero mean and summed variance: $\delta \sim \mathcal{N}(0, 2\tau^{2})$. 
The expectation is a convolution of two Gaussians (the kernel and noise difference), 
resulting in a new Gaussian, whose variance is the sum of the constituent variances.
Hence, the modified kernel $\wt{k}(w,z) \propto \exp\left[-(w-z)^2/\{2(l^2 + 2\tau^{2})\}\right]$ is also Gaussian with an effective bandwidth of $l^2 + 2\tau^{2}$.
\end{example}

 This shows that calculating \MMD on data with Gaussian noise with a Gaussian kernel is equivalent to analysing the clean data with a smoother kernel of a larger bandwidth. 

\begin{theorem} [Large deviation bound for estimation error]\label{th:hypothesis_testing}
Under Assumptions \ref{assump:iid}--\ref{assump: mmd-assump}, for any $\gamma \in (0, 1)$, we have with probability at least $1 - \gamma$:  
\begin{equation*}
\left|\convMMD(p,q,m) -\wh{\convMMD_b}(p,q,m) \right|
\leq (16 K / N)^{1/2}\left\{1 + ({1}/{4})\log \left({2}/{\gamma}\right)^{1/2}\right\}. \nonumber 
\end{equation*}
\end{theorem}

Theorem \ref{th:hypothesis_testing} gives a  probabilistic bound on the estimation error. The bound remains unchanged after adding noise and is primarily determined by the sample size $N$. We illustrate the result empirically in the right panel of Fig.~\ref{fig:mmd-equivalence}, where $\convMMD^2(p,q,m)$ and $\MMD^2(p,q)$ have roughly the same slope, and the bound is independent of the presence of noise and is dominated by $N$. 

While concentration bounds address the rate, they do not fully characterize the magnitude of the estimator's variability. We derive a novel result that explicitly bounds the variance inflation of the unbiased \convMMD U-statistic caused by the measurement error. This provides a more granular, non-asymptotic analysis. 

\begin{figure}
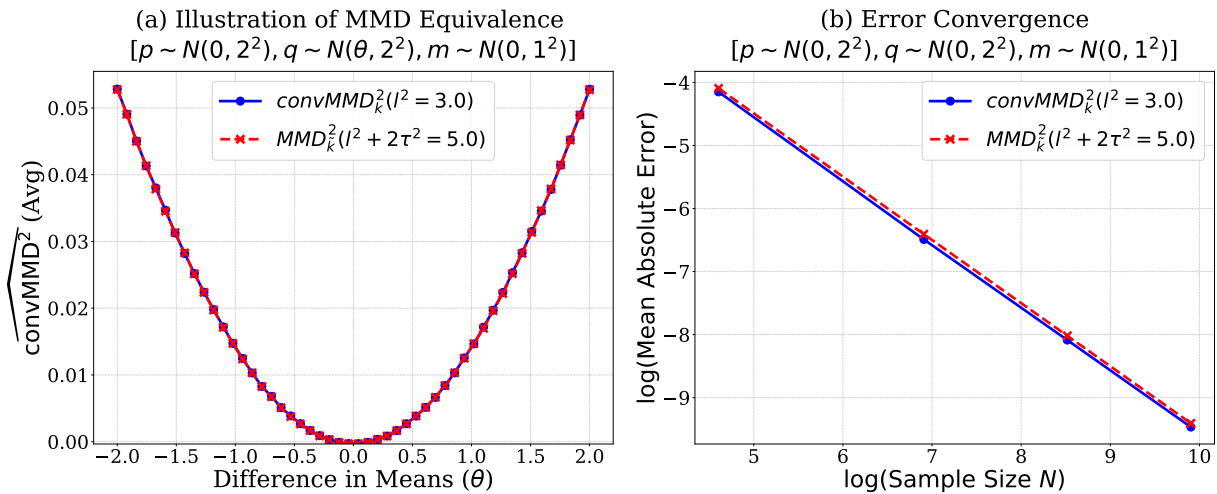

\figuresize{0.2062}
\figurebox{20pc}{25pc}{}[mmd_consistency_result]
\caption{Illustration of Theorem \ref{th: mmd-equivalence} (left) and  Theorem \ref{th:hypothesis_testing} (Right) for the special case of a Gaussian kernel with mean-zero Gaussian noise.}
\label{fig:mmd-equivalence}
\alttext{Two line plots. Left: the squared convolutional discrepancy is a U-shaped function of the difference in means, and the curve for the convolutional discrepancy and the curve for the discrepancy with a wider kernel lie on top of each other. Right: the logarithm of the mean absolute error decreases linearly with the logarithm of the sample size, and the two versions follow almost parallel lines.}
\end{figure}

\begin{theorem} \label{th: mmd-var}
Let $p\!=\!q$, Assumptions \ref{assump:iid}--\ref{assump: mmd-assump} hold, and $k(\wt{\bfW},\wt{\bfZ})\!=\!k_0(\wt{\bfW}\!-\!\wt{\bfZ})$ be a translation invariant kernel, where $k_{0}$ is an $L_{k}$-Lipschitz continuous function, and a finite non-zero variance of $\wh{\MMD_u^2}(p,q)$. Then, for $(\bfW_{i},\bfZ_{i}) \sim p \times q$ and $\delta = [32 L_{k}^{2} \E_{\bfphi \sim g(\bfphi \mid \bfpsi)}\{\alpha(\bfphi)\}]^{1/2}$, we have
\[
\E\left[\left\{\wh{\convMMD_{u}^{2}}(p,q,m)\right\}^2\right] \leq \E\left[\left\{\wh{\MMD_u^2}(p,q)\right\}^2\right] + \frac{2\delta(\delta + 8K)}{N(N-1)}.
\]
\end{theorem}

Theorem \ref{th: mmd-var} provides an insight into the statistical cost of measurement error and quantifies the intuition that noisier measurements lead to less reliable statistics. It formally establishes that the expected squared error of our statistic is bounded by the baseline error in a noise-free setting, plus a ``variance inflation" term that is a direct consequence of the noise. Notably, it is a non-asymptotic bound that holds for any finite sample size $N$. This result precisely identifies the source of the increase in uncertainty: the expected noise variance, $\E_{\bfphi \sim g(\bfphi \mid \bfpsi)}\left\{\alpha(\bfphi)\right\}$.
This links the increase in measurement error in the data with the decrease in statistical precision of the \MMD test statistic. We illustrate this theorem in Section~\ref{sec: sm illust} of the Supplementary Material, see Fig.~\ref{fig: mmd-var}.

\subsection{Parameter estimation} \label{sec: estimation}

We now apply the \convMMD framework to the central problem of parametric estimation from noisy data. Let $\{q_{\bftheta}: \bftheta \in \bfTheta\}$ be a parametric family of distributions defined on a sample space $\calS$, with $\bfTheta \subseteq \mathbb{R}^{d_{\bftheta}}$. We assume that the true latent data generating distribution $p$ is a member of this family, so that $p = q_{\bftheta^{\star}}$ for some unknown true parameter $\bftheta^{\star} \in \bfTheta$. Our goal is to estimate $\bftheta^{\star}$ given a set of $N$ i.i.d observations contaminated with measurement error, drawn from the convolved distribution $p*m = q_{\bftheta^{\star}}*m$.

We define our estimator, $\wh{\bftheta}_N $, as the parameter value that minimizes the discrepancy between the model and the data. This discrepancy is measured by an empirical version of our \convMMD statistic, formed by replacing the true data distribution $p*m$ with its empirical measure based on the noisy samples $\wh{(p*m)}_N$, denoted $\wh{\convMMD}^2(p, q_{\bftheta},m)$. Specifically, the estimator is the minimizer of the semi-empirical \convMMD objective
\begin{equation}
\wh{\bftheta}_{N} = \underset{\bftheta \in \bfTheta}{\arg \min} \ \wh{\convMMD}^2(p, q_{\bftheta},m),
\label{eq:estimator_definition} 
\end{equation}
where the objective function, which we denote by $L_{N}(\bftheta)$, is explicitly defined as 
\begin{equation}
\begin{split}
  L_{N}(\bftheta) &\equiv \wh{\convMMD}^2(p, q_{\bftheta},m) \\
  &=\frac{1}{N^2}\sum_{i,j=1}^{N} k(\wt{\bfw}_{i},\wt{\bfw}_{j})+\E_{\wt{\bfZ}, \wt{\bfZ}'}\left[k\left(\wt{\bfZ}, \wt{\bfZ}'\right)\right] -\frac{2}{N} \sum_{i=1}^{N}\E_{ \wt{\bfZ}}\left[k\left(\wt{\bfw}_{i}, \wt{\bfZ}\right)\right],
\end{split}
\label{eq:empirical_objective}
\end{equation}
where  $\{\wt{\bfw}_{i}\}_{i=1}^{N} \sim p*m$ and $\wt{\bfZ} \sim  q_{\bftheta}*m$. This formulation positions our problem within the well-studied theory of M-estimation \citep{huber1964robust}. The expectations involving the candidate model $q_{\bftheta}*m$ are at the population level, as we can generate an arbitrary number of samples from the model for any given $\bftheta$. Since the objective function in \eqref{eq:empirical_objective} is generally not convex, and its expectations are intractable, we employ a simulation-based optimization approach using stochastic gradient descent (\textsc{sgd}).  The gradient is available through the score function identity, also known as the log-derivative trick, as established by the following proposition.

\begin{proposition} \label{prop: score-gradient}
Suppose that $q_{\bftheta}$ has a valid density with respect to Lebesgue measure and for any $\bfx$, $\bftheta \rightarrow q_{\bftheta}(\bfx)$ is differentiable with respect to $\bftheta$. Then, with ${\bfZ} \sim q_{\bftheta}$  and $\wt{\bfZ} \sim q_{\bftheta}*m$, under regularity conditions that permit the exchange of differentiation and integration, the gradient of the \MMD objective with respect to $\bftheta$ is given by 
\begin{align*}
\nabla_{\bftheta} L_{N}(\bftheta)  &= \nabla_{\bftheta}\left[\E_{\wt{\bfZ}, \wt{\bfZ}'}\left\{k\left(\wt{\bfZ}, \wt{\bfZ}'\right)\right\} -\frac{2}{N} \sum_{i=1}^{N}\E_{ \wt{\bfZ}}\left\{k\left(\wt{\bfw}_{i}, \wt{\bfZ}\right)\right\}\right] \nonumber \\
&=2\E_{\bfZ, \wt{\bfZ}, \wt{\bfZ}'}\left[\left\{k(\wt{\bfZ},\wt{\bfZ}')- \frac{1}{N} \sum_{i=1}^{N}k(\wt{\bfZ}, \wt{\bfw}_{i})\right\}\nabla_{\bftheta}\log q_{\bftheta}(\bfZ)\right].
\end{align*}
\end{proposition}

The unbiased estimate of this gradient can be obtained by using Monte Carlo. At each iteration of our \textsc{sgd} algorithm, we generate a batch of M latent samples  $\{z_{i}\}_{i=1}^{M} \overset{\mathrm{iid}}{\sim} q_{\bftheta}$, compute their scores $\nabla_{\bftheta} \text{log}q_{\bftheta}(y_j)$, and convolve with simulated noise to obtain  
\begin{equation*}
    \wh{\bfJ}_{\bftheta} = \frac{2}{M}\sum_{j=1}^{M}\left[\frac{1}{(M-1)}\sum_{j\neq l}k(\wt{\bfz}_{j},\wt{\bfz}_l) - \frac{1}{N} \sum_{i=1}^{N}k(\wt{\bfz}_{j}, \wt{\bfw}_{i})\right]\nabla_{\bftheta} \text{log}q_{\bftheta}(z_j). 
\end{equation*}
This allows for the construction of an unbiased gradient estimator, $\wh{\bfJ}_{\bftheta}$, as detailed in Algorithm~\ref{alg:optimization} in Section~\ref{sec: Alg} of the Supplementary Material. The algorithm provides a computationally tractable procedure for finding the minimizer $\wh{\bftheta}_N$ defined in \eqref{eq:estimator_definition} by iteratively updating the parameters in the direction of the negative gradient estimate.  We now analyse the properties of the proposed estimator $\wh{\bftheta}_N$. We establish its consistency and derive its limiting distribution, providing a theoretical foundation for its use in statistical inference.
Our first step is to guarantee that the population objective function is well-behaved, meaning it is uniquely minimized at the true parameter value. 
This is guaranteed by the metric property of the \convMMD and standard assumptions on the model.

\begin{lemma} \label{lemma: identifiability-learning}
Suppose $\bftheta_{0} = \underset{\bftheta \in \bfTheta}{\arg \min}\ {\convMMD}^2(p, q_{\bftheta}, m)$, there exists a $\bftheta^{\star}$ such that $q_{\bftheta^{\star}} = p$ and $\bftheta \mapsto q(\bftheta)$ is injective and continuous for almost every $\bftheta \in \bfTheta$. Under Assumptions \ref{assump:iid}--\ref{assump: mmd-assump}, ${\bftheta_{0}} = {\bftheta^{\star}}$ if and only if $\convMMD(p, q_{{\bftheta_{0}}}, m) = 0$.
\end{lemma}
Lemma~\ref{lemma: identifiability-learning} represents the identifiability condition, guaranteeing a well-defined objective and the minimizer to indeed be $\bftheta^{\star}$. Now, we consider the sample objective $L_{N}(\bftheta)$ and consider the convergence of $\wh{\bftheta}_{N}$.

\begin{theorem}[Generalization bound and consistency] \label{th: bound and consistency}
Suppose that a unique minimizer $\bftheta^{\star} \in \bfTheta$ exists such that $\bftheta^{\star} = \underset{\bftheta \in \bfTheta}{\inf}\convMMD\left(p, q_{\bftheta}, m\right)$, then under Assumptions \ref{assump:iid}--\ref{assump: mmd-assump} and Lemma \ref{lemma: identifiability-learning}, we have 
\begin{enumerate}
   \item (Generalization bound) With probability $1 - \gamma$, we have 
\begin{equation*}
\convMMD\left(p, q_{\wh{\bftheta}_{N}}, m\right) \leq \underset{\bftheta \in \bfTheta}{\inf}\convMMD\left(p, q_{\bftheta}, m\right)+  4\left({2K}/{N}\right)^{1/2} \left[2+\{\log (1 / \gamma)\}^{1/2}\right].\nonumber 
\end{equation*}
\item (Consistency) The empirical \convMMD estimator $\wh{\bftheta}_{N}$ converges almost surely to $\bftheta^{\star}$ if it is fitted on noisy data and samples of size $N$ from $q_{\bftheta}*m$. That is $\underset{N \rightarrow \infty}{\lim} \wh{\bftheta}_{N} = \bftheta^{\star}$ almost surely. 
\end{enumerate}
\end{theorem}

Theorem \ref{th: bound and consistency} provides a probabilistic bound on convergence and confirms that $\wh{\bftheta}_N$ remains a $N^{1/2}$-consistent estimator despite the presence of measurement error. This might appear to be counterintuitive, as one would expect the rate of convergence of the estimator to degrade in the presence of noise in the data. However, the persistence of the $N^{1/2}$-rate is primarily because of our assumption that the true data-generating distribution $p$ belongs to the specified parametric family $\{q_{\bftheta}\}$. This assumption constrains the estimation problem to a finite-dimensional space, preventing degradation of the rate. If we relaxed this assumption and worked in a nonparametric setup, then the corresponding estimator would exhibit a slower rate of convergence with dependence on the noise distribution ($m$). 
We now consider the asymptotic normality of our estimator and characterize the efficiency of the estimator in terms of its asymptotic variance. 

\begin{theorem}[Central limit theorem]
\label{th:clt}
Let $\{\wt{\bfw}_{i}\}_{i=1}^{N}$ be i.i.d. samples from the noisy distribution $p*m$. Let the candidate models be $\{q_{\bftheta} : \bftheta \in \bfTheta\}$, where $\bfTheta \subseteq \R^{d_{\bftheta}}$ is an open set. Let $\wh{\bftheta}_N$ be the empirical \convMMD estimator obtained by minimizing the empirical objective function $L_{N}(\bftheta) = \MMD^{2}\{\wh{(p*m)}_{N}, q_{\bftheta}*m\}$. 
Assume the following regularity conditions hold: 
\begin{enumerate}
    \item[(i)] The true distribution $p$ is in the model class: $p = q_{\bftheta^{\star}}$ for a unique $\bftheta^{\star} \in \bfTheta$.
    \item[(ii)] The estimator is consistent: $\wh{\bftheta}_N \xrightarrow{p} \bftheta^{\star}$ as $N \to \infty$.
    \item[(iii)] The log-density $\log q_{\bftheta}(\bfz)$ is twice continuously differentiable with respect to $\bftheta$ in a neighbourhood of $\bftheta^{\star}$.  
    \item[(iv)] The kernel $k(\cdot, \cdot)$ and expectations involving it are sufficiently smooth to justify interchanging differentiation and integration. 
    \item[(v)] The matrices $\wt{\bfH}(\bftheta^{\star})$ and $\wt{\bfSigma}_{\rm score}$ defined below are finite and $\wt{\bfH}(\bftheta^{\star})$ is positive definite.
\end{enumerate}
Then, as $N \to \infty$, the estimator is asymptotically normal $N^{1/2} (\wh{\bftheta}_N - \bftheta^{\star}) \xrightarrow{d} \N(\bfzero, \wt{\bfC}_{\rm score})$. 
The asymptotic covariance matrix is the Godambe information matrix $\wt{\bfC}_{\rm score} = \wt{\bfH}(\bftheta^{\star})^{-1} \wt{\bfSigma}_{\rm score} \wt{\bfH}(\bftheta^{\star})^{-1}$, where \citep{ferguson2017course}:  
\begin{enumerate}
    \item \textit{The curvature matrix}: $\wt{\bfH}(\bftheta^{\star}) = \nabla_{\bftheta}^{2} L(\bftheta) \big|_{\bftheta=\bftheta^{\star}}$, where $L(\bftheta)$ is the population objective function $L(\bftheta) = \convMMD^2(p, q_{\bftheta}, m)$.
    
    \item \textit{The gradient variance matrix}: $\wt{\bfSigma}_{\rm score} = \E_{\wt{\bfW}\sim p*m} \left[ s(\wt{\bfW}, \bftheta^{\star}) s(\wt{\bfW}, \bftheta^{\star})\trans \right]$, where the score vector $s(\wt{\bfW}, \bftheta)$ is the influence of a single data point $\wt{\bfW}$ on the population gradient, defined as
    $s(\wt{\bfW}, \bftheta) = \nabla_{\bftheta} \left[ \E_{\wt{\bfZ},\wt{\bfZ}' \sim q_{\bftheta}*m}\{k(\wt{\bfZ}, \wt{\bfZ}')\} - 2\E_{\wt{\bfZ} \sim q_{\bftheta}*m}\{k(\wt{\bfW}, \wt{\bfZ})\} \right]$. 
\end{enumerate}
\end{theorem}
Theorem \ref{th:clt} provides a powerful characterization of our estimator's behaviour at large sample sizes. It shows that the measurement error, characterized by the noise distribution $m$, does not change the rate of convergence. The asymptotic variance ($\wt{\bfC}_{\rm{score}}$) quantifies the estimator's precision for a large but finite $N$. Our analysis reveals that the noise distribution ($m$) directly influences the components of the Godambe matrix, $\wt{\bfH}$ and $\wt{\bfSigma}_{\rm{score}}$. Generally, higher noise levels will lead to a larger $\wt{\bfC}_{\rm{score}}$, resulting in wider confidence intervals and a less precise estimator for a given sample size $N$. In essence, the noise degrades the quality of the information from each sample, but not the fundamental scaling law governing the rate of convergence.  
\begin{figure}
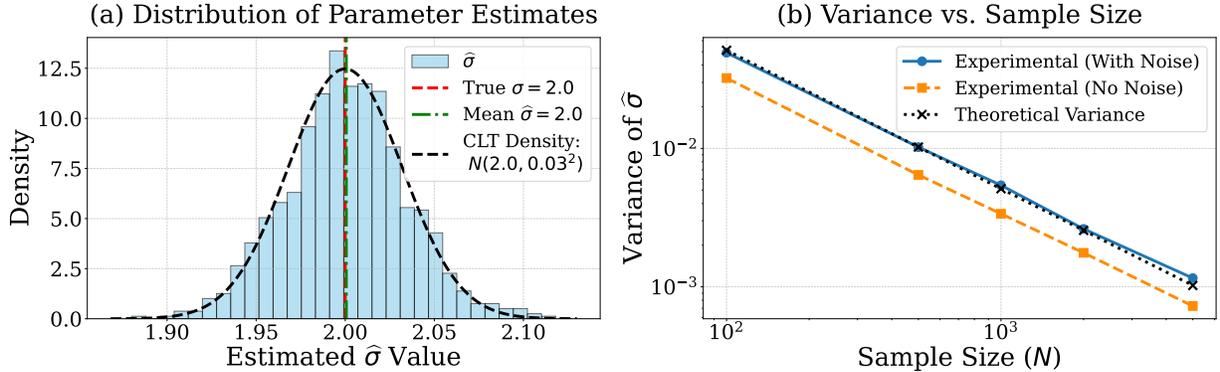

\figuresize{0.2235}
\figurebox{20pc}{25pc}{}[clt_sigma]
\caption{Simulation experiments: Illustration of Theorem \ref{th:clt} in a problem of scale estimation of a Gaussian distribution in the presence of Gaussian noise: distribution of $\wh{\sigma}$ for $N= 5000$ (left), rate of convergence of $\wh{\sigma}$ (right).}
\label{fig:clt-experiment}
\alttext{Two panels. Left: a bell-shaped histogram of estimates of the scale parameter, centred at the true value 2.0 and closely matched by a normal density curve. Right: log-log plot of the variance of the estimate against sample size; the variance with noise (blue) lies above the variance without noise (orange), both decrease linearly, and the theoretical variance (black) follows the blue line.}
\end{figure}

We make our theoretical results concrete using two examples. We consider the problem of estimating the scale parameter $\sigma$ of a zero-mean normal distribution, $p = N(0,\sigma^{2}_{*})$, from noisy observations.  Figure \ref{fig:clt-experiment}(a) shows the empirical distribution of the \convMMD estimates for a sample size $N = 5000$. The distribution follows the normal density derived from Theorem \ref{th:clt}, confirming the validity of our asymptotic result. Figure \ref{fig:clt-experiment}(b) illustrates the convergence of the estimator by plotting the empirical variance of the estimates against the sample size. The variance is inversely proportional to $N$, confirming that the estimator is $N^{1/2}$-consistent despite the presence of noise in the data. The figure also includes, for comparison, the variance of a standard \MMD estimator applied to noise-free data. While both estimators exhibit the same convergence rate, the variance of the \convMMD estimator is uniformly higher across all sample sizes. This result provides a clear, empirical illustration of the statistical cost of measurement error: a loss in efficiency (higher variance), but not a degradation of the (parametric) convergence rate.

Next, we look at the problem of estimating the mean of a Gaussian distribution in the presence of noise, where we are able to obtain an exact expression for the asymptotic covariance matrix $\wt{\bfC}_{\rm score}$. This provides a direct illustration of how the data variance, noise variance, and kernel bandwidth interact to determine the estimator's efficiency.

\begin{example} \label{ex: clt-illust}
Consider the data-generating distribution is $p = \mathcal{N}(\bftheta^{\star}, \sigma^{2} \bfI_{d})$ and the candidate family is $q_{\bftheta} = \mathcal{N}(\bftheta, \sigma^{2} \bfI_{d})$, for $\bftheta, \bftheta^{\star} \in \R^d$. The data and model samples are convolved with known isotropic, homogeneous Gaussian noise $m = \mathcal{N}(\bfzero, \tau^{2} \bfI_{d})$. Using a Gaussian kernel $k(\bfw, \bfz) \propto \exp\left(-\|\bfw-\bfz\|^2/(2l^2)\right)$, the asymptotic covariance of the \convMMD estimator $\wh{\bftheta}_N$ is given by
$\wt{\bfC}_{\rm score} = (\sigma^{2} + \tau^{2})\left\{(l^2+\sigma^{2} + \tau^{2})(l^2+3\sigma^{2} + 3\tau^{2})\right\}^{-d/2-1}\left(l^2+2\sigma^{2} + 2\tau^{2}\right)^{d+2} \bfI_{d}$.
\end{example}
This closed-form expression for the asymptotic covariance is derived by computing the curvature matrix $\wt{\bfH}(\bftheta^{\star})$ and gradient variance matrix $\wt{\bfSigma}_{\rm score}$ for the Gaussian setting and inserting them into the sandwich formula of Theorem \ref{th:clt}. 
The key insights are:  
\begin{enumerate}
    \item \textit{Noise level}: The expression features the term $\sigma^{2} +\tau^{2}$ where the variance of clean data $\sigma^{2}$ would typically appear. This confirms the intuition that the uncertainty of our estimator is driven by the sum of the true data variance $\sigma^{2}$ and the noise variance $\tau^{2}$, where more noise leads directly to a less precise parameter estimate.  
    \item \textit{Kernel bandwidth}: Kernel bandwidth $l$ affects the variance estimate non-monotonically. A small $l$ makes the kernel highly sensitive and leads to a noisy objective function with high gradient variance $\wt{\bfSigma}_{\rm score}$. A large $l$ oversmooths the distributions and flattens the loss landscape, making  $\wt{\bfH}(\bftheta^{\star})$ small and its inverse large.
\end{enumerate}
This implies that there is an optimal intermediate kernel bandwidth $l$ that minimizes the asymptotic variance $\wt{\bfC}_{\rm score}$, thus maximizing the statistical efficiency of the estimator.

\section{Numerical experiments}

\subsection{Gaussian mixture model} \label{sec: gmm-estimation}

We evaluate the \convMMD estimator on estimating Gaussian mixture models (\textsc{gmm}s) and error-in-variables regression (\textsc{eivr}). 
We compare against likelihood-based alternatives: extreme deconvolution \citep[\textsc{xdgmm},][]{bovy2011extreme}, \textsc{simex} \citep{cook1994simulation}, and \texttt{linmix} \citep[Bayesian \textsc{eivr} model,][]{kelly2007some}. 
We also include naive estimators that ignore noise. 
Additional details are provided in Section~\ref{sec: sm sim} of the Supplementary Material.

The core task in this study is to estimate the parameters of the three-component \textsc{gmm} of a latent variable $X$ from their noisy observations $\wt{X}$. 
We minimize the \convMMD objective between the  distribution of the observed data, $\wt{W} \equiv \wt{X}$, and the noisy samples simulated from the model, $\wt{Z} \equiv \wt{X}_{\text{sim}}$.
We compare our methods against two likelihood-based benchmarks: a standard \textsc{gmm} that ignores measurement error, and \textsc{xdgmm}, a Bayesian method that handles Gaussian measurement errors \citep{bovy2011extreme}. 

\textit{Simulation design.}
We first generate $\{X_{i}\}_{i=1}^{N}$ of size $N = 1000$ from a \textsc{gmm} as 
\begin{equation*}
\textstyle X_{i} \sim \sum_{k=1}^3 \pi_{k}^{\star}\, \mathcal{N}(\mu_{k}^{\star}, \sigma_{k}^{\star}),
\qquad 
\pi_{k}^{\star}>0,\;\; \sum_{k=1}^3 \pi_{k}^{\star}= 1, 
\end{equation*}
where the component means $\mu_{k}^{\star}$, mixture weights $\pi_{k}^{\star}$ and component standard deviation $\sigma_{k}^{\star}$ are the parameters of interest. 
The $X_{i}$'s are then corrupted by additive noise as 
%
$\wt{X}_{i} = X_{i} + U_{i}$,  $U_{i} \stackrel{\text{iid}}{\sim} m(\cdot)$, where $m$ denotes some noise distribution, to produce the observed data. 
We vary the noise distribution and consider 5 different cases: (i) Homoscedastic: Standard i.i.d. Gaussian and uniform noise
(ii) Heteroscedastic: Hierarchical noise model with Gaussian, Laplace, and Student's t distributions. All methods are provided with the true parameters of the noise model. However, the likelihood-based methods (\textsc{gmm} and \textsc{xdgmm}) assume Gaussian noise. 
For each case, we perform 50 independent runs and report the mean absolute error (\textsc{mae}) for the estimated component means, weights, and standard deviations. 
We also report \textsc{mae} for density estimates.  

\begin{table}
\tbl{Simulation experiments: \textsc{mae} for estimates of \textsc{gmm} mean and standard deviations (across components), and density estimates. The lowest values highlighted in bold.}
{\setlength{\tabcolsep}{5pt}\begin{tabular}{llccccc}
\hline
 & & \multicolumn{5}{c}{Noise} \\
\cline{3-7}
 & & \multicolumn{2}{c}{Homoscedastic} & \multicolumn{3}{c}{Heteroscedastic} \\
\cline{3-4}\cline{5-7}
 & Method & Gaussian & Uniform & Gaussian & Laplace & Student's t \\
\hline
\multirow{3}{*}{Means \textsc{mae}} & \convMMD & 0.255 $\pm$ 0.057 & \textbf{0.471 $\pm$ 0.100} & 0.255 $\pm$ 0.057 & \textbf{0.182 $\pm$ 0.048} & \textbf{0.350} $\pm$ \textbf{0.100} \\
 & \textsc{xdgmm} & 0.345 $\pm$ 0.094 & 1.166 $\pm$ 0.186 & 0.281 $\pm$ 0.078 & 0.275 $\pm$ 0.080 & 0.395 $\pm$ 0.106 \\
 & \textsc{gmm} & \textbf{0.222 $\pm$ 0.018} & 0.576 $\pm$ 0.143 & \textbf{0.216 $\pm$ 0.018} & 0.221 $\pm$ 0.067 & 0.573 $\pm$ 0.132 \\
\noalign{\vskip3pt}
\multirow{3}{*}{Std. dev. \textsc{mae}} & \convMMD & 0.158 $\pm$ 0.010 & \textbf{0.175 $\pm$ 0.010} & 0.157 $\pm$ 0.009 & \textbf{0.115 $\pm$ 0.007} & \textbf{0.141 $\pm$ 0.012} \\
 & \textsc{xdgmm} & \textbf{0.143 $\pm$ 0.009} & 0.241 $\pm$ 0.011 & \textbf{0.148 $\pm$ 0.008} & 0.213 $\pm$ 0.008 & 0.284 $\pm$ 0.012 \\
 & \textsc{gmm} & 0.563 $\pm$ 0.007 & 0.557 $\pm$ 0.005 & 0.560 $\pm$ 0.007 & 0.480 $\pm$ 0.007 & 0.473 $\pm$ 0.016 \\
\noalign{\vskip3pt}
\multirow{3}{*}{Density \textsc{mae}} & \convMMD & 0.011 $\pm$ 0.001 & 0.014 $\pm$ 0.001 & 0.010 $\pm$ 0.000 & \textbf{0.007 $\pm$ 0.000} & \textbf{0.006 $\pm$ 0.000} \\
 & \textsc{xdgmm} & \textbf{0.010 $\pm$ 0.000} & \textbf{0.013 $\pm$ 0.001} & \textbf{0.009 $\pm$ 0.000} & 0.012 $\pm$ 0.001 & 0.014 $\pm$ 0.001 \\
 & \textsc{gmm} & 0.021 $\pm$ 0.000 & 0.022 $\pm$ 0.000 & 0.020 $\pm$ 0.000 & 0.017 $\pm$ 0.000 & 0.013 $\pm$ 0.000 \\
\hline
\end{tabular}}
\label{tab:overall_results_gmm_sem}
\begin{tabnote}
\textsc{mae}, mean absolute error; \textsc{gmm}, Gaussian mixture model; \textsc{xdgmm}, extreme deconvolution; \convMMD, convolutional maximum mean discrepancy; Std. dev., standard deviation.
\end{tabnote}
\end{table}

\textit{Results.}
The results of our comparative study are summarized in Table \ref{tab:overall_results_gmm_sem}. 
We observe that when the noise is Gaussian, \textsc{xdgmm} performs best, achieving the lowest \textsc{mae} for both standard deviation and density. Our \convMMD method performs comparably and is an improvement over the naive \textsc{gmm}. 
However, when we consider heavier-tailed distributions such as Laplace or Student's t, the performance of both \textsc{xdgmm} and naive \textsc{gmm} degrades. The presence of outliers, for which the Gaussian likelihood assigns small probabilities, leads to poor parameter estimation, which leads to high \textsc{mae}s. In contrast, our \convMMD estimator remains stable. This highlights the flexibility and robustness of our kernel-based method compared to rigid likelihood assumptions.

\begin{figure}
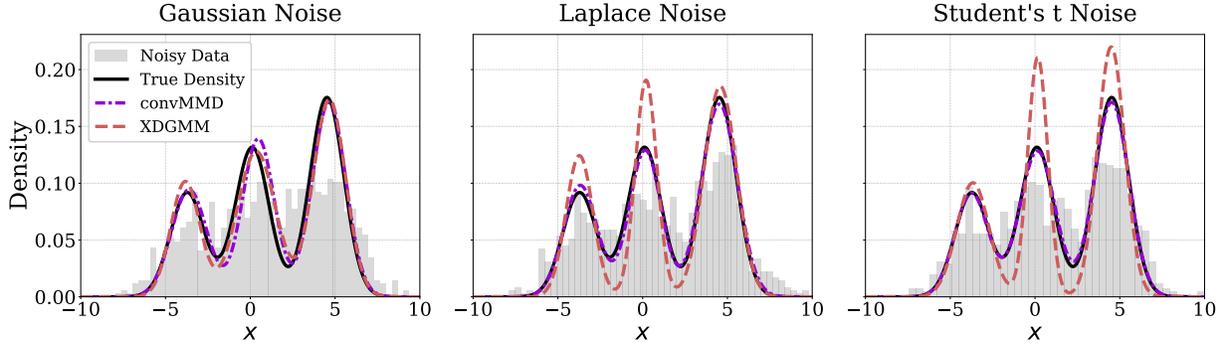

\figuresize{0.2059}
\figurebox{20pc}{25pc}{}[gmm_density_comparison]
\caption{Simulation experiments: Estimated density using different methods across three different noise distributions: Gaussian (left), Laplace (middle), Student's t (right).}
\label{fig: mixgauss-estimation}
\alttext{Three density plots, for Gaussian, Laplace and Student's t noise. Each shows a grey histogram of the noisy data, the true three-peak density (black) and the estimates from \convMMD (purple) and \textsc{xdgmm} (red). The \convMMD curve stays close to the true density in all three panels; the \textsc{xdgmm} curve is close under Gaussian noise but exaggerates the peaks and troughs under Laplace and Student's t noise.}
\end{figure}

\begin{table}
\tbl{Simulation experiments: \textsc{mae} for \textsc{eivr} parameter estimates, with the lowest values highlighted in bold.}
{\begin{tabular}{llccccc}
\hline
 & & \multicolumn{5}{c}{Noise} \\
\cline{3-7}
 & & \multicolumn{2}{c}{Homoscedastic} & \multicolumn{3}{c}{Heteroscedastic} \\
\cline{3-4}\cline{5-7}
 & Method & Gaussian & Uniform & Gaussian & Laplace & Student's t \\
\hline
\multirow{4}{*}{$\alpha$ \textsc{mae}} & \convMMD & 0.291 $\pm$ 0.026 & \textbf{0.260 $\pm$ 0.027} & \textbf{0.266 $\pm$ 0.024} & \textbf{0.241} $\pm$ \textbf{0.028} & \textbf{0.204} $\pm$ \textbf{0.022} \\
 & \textsc{simex} & 1.051 $\pm$ 0.019 & 0.926 $\pm$ 0.024 & 1.046 $\pm$ 0.021 & 1.065 $\pm$ 0.029 & 1.011 $\pm$ 0.044 \\
 & \texttt{linmix} & \textbf{0.274}  $\pm$ \textbf{0.026} & 0.274 $\pm$ 0.031 & 0.274 $\pm$ 0.027 & 0.393 $\pm$ 0.045 & 1.714 $\pm$ 0.227 \\
 & \textsc{ols} & 1.701 $\pm$ 0.013 & 1.601 $\pm$ 0.016 & 2.202 $\pm$ 0.014 & 2.197 $\pm$ 0.022 & 2.191 $\pm$ 0.024 \\
\noalign{\vskip3pt}
\multirow{4}{*}{$\beta$ \textsc{mae}} & \convMMD & 0.100 $\pm$ 0.009 & \textbf{0.095} $\pm$ \textbf{0.009} & \textbf{0.093 $\pm$ 0.008} & \textbf{0.086} $\pm$ \textbf{0.010} & \textbf{0.068} $\pm$ \textbf{0.007} \\
 & \textsc{simex} & 0.368 $\pm$ 0.006 & 0.323 $\pm$ 0.008 & 0.366 $\pm$ 0.006 & 0.374 $\pm$ 0.009 & 0.350 $\pm$ 0.014 \\
 & \texttt{linmix} & \textbf{0.094} $\pm$ \textbf{0.009} & 0.120 $\pm$ 0.011 & \textbf{0.093} $\pm$ \textbf{0.009} & 0.135 $\pm$ 0.016 & 0.598 $\pm$ 0.078 \\
 & \textsc{ols} & 0.596 $\pm$ 0.004 & 0.559 $\pm$ 0.005 & 0.596 $\pm$ 0.004 & 0.590 $\pm$ 0.007 & 0.590 $\pm$ 0.007 \\
\hline
\end{tabular}}
\label{tab:regression_results}
\begin{tabnote}
\textsc{mae}, mean absolute error; \textsc{ols}, ordinary least squares; \textsc{simex}, simulation extrapolation; \convMMD, convolutional maximum mean discrepancy.
\end{tabnote}
\end{table}

\subsection{Error-in-variables regression}

We apply the \convMMD framework to the problem of linear \textsc{eivr} of an outcome $Y$ on a covariate $X$ to recover the intercept and slope using their noisy observations $\wt{Y}$ and $\wt{X}$. 
We minimize the \convMMD objective between the joint distribution of the observed data, $\wt{\bfW} \equiv \{(\wt{X}, \wt{Y})\trans\}$, and the noisy samples simulated from the model, $\wt{\bfZ} \equiv \{(\wt{X}_{\text{sim}}, \wt{Y}_{\text{sim}})\trans\}$. 

\begin{figure}
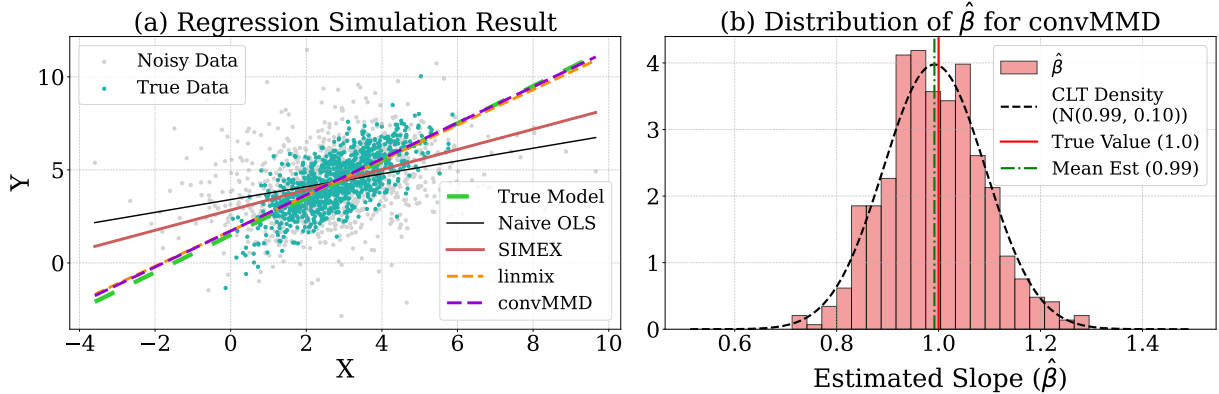

\figuresize{0.2235}
\figurebox{20pc}{25pc}{}[regression_sim_result]
\caption{Simulation experiments: Performance of our proposed \convMMD Regression compared with other approaches on simulated data with heteroscedastic Laplace noise (left), along with the distribution of the \convMMD estimates of $\hat{\beta}$ (right).}
\label{fig: mmd-reg_asymp}
\alttext{Left: scatter plot of Y against X with noisy points in grey, true points in teal and five fitted lines. The true model, \convMMD and linmix lines nearly coincide, while the naive \textsc{ols} and \textsc{simex} lines are visibly flatter. Right: histogram of the \convMMD slope estimates centred near the true value 1.0, with a normal density curve (mean 0.99).}
\end{figure}

\textit{Simulation design.}
We first generate 
$\{(X_{i}, Y_{i})\}_{i=1}^{N}$ of size $N=1000$ as 
\begin{equation*}
Y_{i} =\alpha^{\star} + \beta^{\star} X_{i} + \varepsilon_i, \qquad X_{i} \sim 0.3\,\mathcal{N}(2.5, 1) + 0.7\,\mathcal{N}(3, 1), \qquad \varepsilon_i\!\sim\!\mathcal{N}(0, 1),
\end{equation*}
with $\alpha^{\star}\!=\!1.5$ and $\beta^{\star}\!=\!1$.
The observed data $\{(\wt{X}_{i}, \wt{Y}_{i})\}_{i=1}^{N}$ are $X$'s $Y$'s with measurement errors. 
We assess the performance of our \convMMD estimator under various noise models, including light-tailed Gaussian and uniform, and heavy-tailed Student's t distributions, mirroring the \textsc{gmm} study. 
\textsc{mae} measures performance for the estimated intercept and slope, averaged over 50 replications.

\textit{Results.}
The key results of the regression study are summarized in Table \ref{tab:regression_results}, Fig.~\ref{fig: mmd-reg_asymp} and Fig.~\ref{fig: mmd-reg_asymp2}. We see that the naive ordinary least squares (\textsc{ols}) estimator that ignores measurement error performs poorly in all cases, as expected, consistently underestimating the slope. 
This demonstrates attenuation bias, where measurement error in the covariate shrinks the estimated coefficient toward zero. Under Gaussian measurement error, the ideal setting for \texttt{SIMEX} and \texttt{linmix} regression, our \convMMD method is highly competitive.
However, in the presence of heavy-tailed noise, both benchmarks show a sharp decline in performance. Predicated on Gaussian assumptions, these methods are sensitive to frequent outliers, leading to an increase in estimation error. These results are consistent with the \textsc{gmm} parameter estimation findings in Section \ref{sec: gmm-estimation}.

\section{Real data illustrations}

We demonstrate our method's versatility across diverse applications. 
Extended details are available in Section~\ref{sec: sm real} of the Supplementary Material.

\subsection{Astronomy}

We apply our method to galaxy cluster data from the Dark Energy Survey \citep[\textsc{des},][]{bacon2021dark} to examine the scaling relation between two key cluster mass proxies: optical richness ($\lambda_{\rm RM}$) and hot gas temperature ($T_{X}$). 
As the most massive gravitationally bound systems in the universe, these clusters serve as powerful cosmological probes \citep{allen2011cosmological}. In this application, both observables are subject to measurement error, presenting an \textsc{eivr} problem \citep{mantz2010observed}.

\begin{figure}
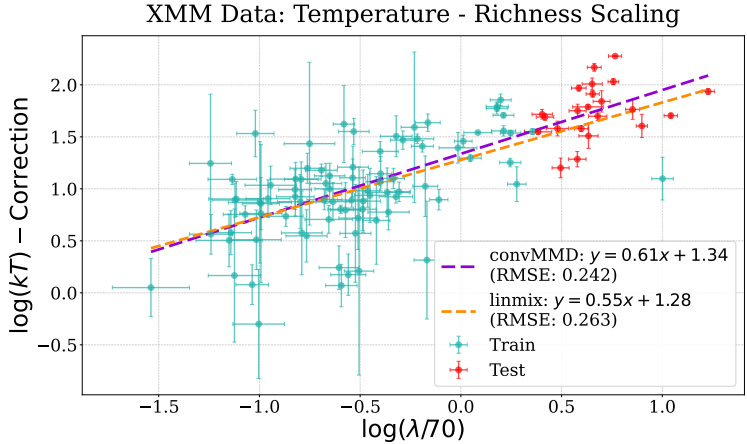

\figuresize{0.2539}
\figurebox{20pc}{25pc}{}[regression_xmm_result]
\caption{Astronomy: Empirical scaling between cluster hot gas temperature and optical richness with the \convMMD regression fit.}
\label{fig:xmm-reg}
\alttext{Scatter plot of log temperature, corrected, against log richness, with error bars: training clusters in teal, widely scattered, and test clusters in red at higher richness. Two upward-sloping lines, \convMMD (purple, slope 0.61, \textsc{rmse} 0.242) and linmix (orange, slope 0.55, \textsc{rmse} 0.263), lie close together.}
\end{figure}

The dataset, compiled by \cite{farahi2019mass}, consists of 110 clusters from the \textsc{des} Year 1 catalogue \citep{abbott2020dark} within the redshift range $0.2 < z < 0.7$. 
The regression uses two observables: (1) optical richness $\lambda_{\rm RM}$, a probability-weighted count of likely member galaxies that serves as a proxy for total cluster mass; and (2) hot gas temperature $T_{X}$ (keV), which traces the gravitational potential of the intracluster medium and is likewise a mass proxy. 
Additionally, the data pipelines provide object-specific, quantitative estimates of statistical uncertainty associated with each measurement ($\bfphi_{i}$). 

\textit{Model and results.}
Astrophysical ``scaling relations'' are commonly modelled as power laws \citep{mulroy2019locuss}, which become linear after taking logarithms. To factor out expected redshift evolution, we include the evolution factor $E(z)\!=\!H(z)/H_0$. 
Our regression is $\E\{\log(T_{X})\!\mid\!\lambda_{RM}, z\}\!=\!\alpha\!+\!(2/3)\log\{E(z)\}\!+\!\beta \log\left(\lambda_{RM}/\lambda_{\text{piv}}\right)$, where $\lambda_{\rm piv}\!=\!70$ and $z\!=\!0$ define reference points; 
$\alpha$ is the intercept (log-temperature at the pivot), and $\beta$ is the slope. 
{Both $\lambda_{RM}$ and $T_{X}$ are assumed to be measured with error.} 
\cite{farahi2019mass} modelled this relationship with \texttt{linmix} \citep{kelly2007some}. 
Here we apply the \convMMD framework developed in Section~\ref{sec: estimation}, 
{with $\wt{\bfW} \equiv (\wt{X},\wt{Y})\trans \equiv (\log \lambda_{RM},\log T_{X})\trans$ and $\wt{\bfZ}\!\equiv\!(\wt{X}_{\text{sim}},\wt{Y}_{\text{sim}})\trans$, assuming $p(X)$ to be a \textsc{gmm} and $p(X\!\mid\!Y)$ to be normal},  
accommodating heteroscedastic uncertainties and intrinsic scatter, and using the pipeline-reported, cluster-specific uncertainties as parameters of a mean-zero Gaussian noise model for both richness and temperature. Figure~\ref{fig:xmm-reg} shows the results. 
Consistent with \cite{farahi2019mass}, our method faithfully captures the underlying relationship. 
For a direct comparison, we evaluated both approaches on a test set consisting of observations with the smallest measurement uncertainties. On this set, our method achieves a lower \textsc{rmse} and provides a better fit (\textsc{rmse}=0.24) than that of \cite{farahi2019mass} (\textsc{rmse} = 0.26).

\subsection{Anthropometry}

Next, we consider an anthropometric regression task using the Davis dataset \citep{davis1990body}, which provides self-reported and measured body metrics for $N=183$ individuals. The presence of both validated and reported measurements makes this dataset well-suited to studying \textsc{eivr} under realistic reporting noise. The dataset consists of (measured) heights ($X$), their self-reported proxies ($\wt{X}$), and (measured) weights ($Y$). It exhibits two salient challenges: (i) \emph{heteroscedastic reporting error}, where discrepancies between $\wt{X}$ and $X$ vary systematically across the range of body measurements, and (ii) a known \emph{data-entry outlier} where the height and weight values were swapped. We treat the paired observations $(\wt{X},X)$ as a validation sub-study \citep{carroll2006measurement}, using them to characterize the empirical measurement error distribution.

\begin{figure}
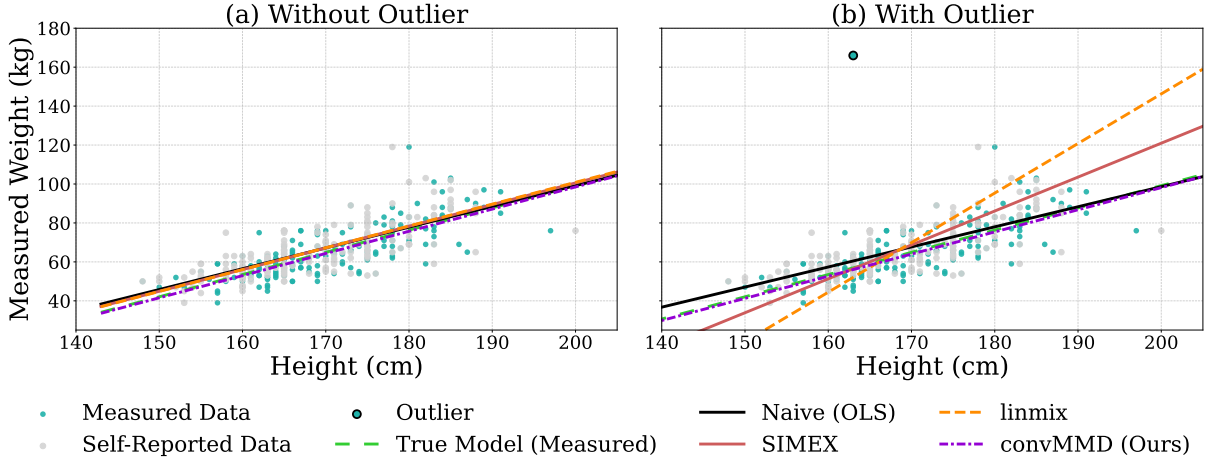

\figuresize{0.225}
\figurebox{20pc}{25pc}{}[davis_height_weight]
\caption{Anthropometry: Fitted regression lines without (left) and with outlier (right).}
\label{fig:reg_davis_new}
\alttext{Two scatter plots of measured weight against height for measured (teal) and self-reported (grey) data, with five fitted lines. Left, without the outlier: all lines almost coincide. Right, with an outlying point at about 163 cm and 166 kg: the linmix and \textsc{simex} lines tilt steeply upwards, while the naive \textsc{ols}, \convMMD and true-model lines stay close to their positions in the left panel.}
\end{figure}

\begin{table}
\tbl{Anthropometry: \textsc{ae} for \textsc{eivr}, with the lowest values highlighted in bold.}
{\setlength{\tabcolsep}{15pt}\begin{tabular}{lcccc}
\hline
 & \multicolumn{4}{c}{Setting} \\
\cline{2-5}
 & \multicolumn{2}{c}{Without Outlier} & \multicolumn{2}{c}{With Outlier} \\
\cline{2-3}\cline{4-5}
Methods & $\alpha$ & $\beta$ & $\alpha$ & $\beta$ \\
\hline
\convMMD & \textbf{0.136} & \textbf{0.003} & \textbf{0.222} & \textbf{0.004} \\
\textsc{simex} & 7.133 & 0.028 & 98.218 & 0.600 \\
\texttt{linmix} & 5.680 & 0.020 & 233.766 & 1.404 \\
Naive & 14.851 & 0.074 & 21.469 & 0.109 \\
\hline
\end{tabular}}
\label{tab:regression_results_davis_new}
\begin{tabnote}
\textsc{ae}, absolute error; \textsc{simex}, simulation extrapolation; \convMMD, convolutional maximum mean discrepancy.
\end{tabnote}
\end{table}

\textit{Model and results.}
We study two linear \textsc{eivr} settings of $Y$ on $\wt{X}$ with and without the outlier. 
We have $\wt{\bfW} \equiv (\wt{X},Y)\trans$ and $\wt{\bfZ} \equiv (\wt{X}_{\text{sim}},Y_{\text{sim}})\trans$; 
we model $p(X)$ as a \textsc{gmm} 
and $p(Y \mid X)$ assuming normally distributed regression errors with an unknown variance. 
We evaluate the performance using the absolute error (\textsc{ae}) of the estimated regression coefficients, reported in Table~\ref{tab:regression_results_davis_new}. 
In both settings,  with and without the outlier, the naive estimator that regresses on $\wt{X}$ suffers from attenuation bias. 
\textsc{simex} and the \texttt{linmix} partially correct this bias but are sensitive to misspecification of noise and degrade sharply in the presence of the outlier. 
In contrast, \convMMD consistently achieves the lowest \textsc{ae} and remains stable even when the outlier is included, indicating robustness. 
Figure~\ref{fig:reg_davis_new} shows the fitted regression lines for both cases. The plots highlight the estimated regression lines, where \convMMD more accurately captures the relationship by effectively modelling error without being affected by the outlier.

\subsection{Homeownership}

In this example, we examine the dependence of homeownership status on two fundamental household characteristics: annual household income and the age of the householder. We use data from the 2019 American Housing Survey (\textsc{ahs}), which is conducted by the Census Bureau and contains detailed and comprehensive data on housing in the U.S. 

To validate our method, we treat the reported \textsc{ahs} values as \textit{clean} ground truth and add noise to the covariates simulate an \textsc{eivr} problem. 
This allows us to compare our estimates against the true parameters derived from the clean data \citep{farahi2024analyzing}, which is usually not possible with noisy data.  We compare the performance of \convMMD (Section~\ref{sec: estimation}) against a naive \textsc{glm} (which ignores noise) and the \textsc{simex} method.

\begin{table}
\tbl{Homeownership: \textsc{mae} and Brier scores. The lowest values highlighted in bold.}
{\begin{tabular}{lcccc}
\hline
Method & Intercept \textsc{mae} & Age \textsc{mae} & Income \textsc{mae} & Brier score \\ \hline
Naive &  0.063 $\pm$ 0.000 & 0.470 $\pm$ 0.001 & 0.478 $\pm$ 0.001 & 0.191 $\pm$ 0.000 \\
\textsc{simex} &  \textbf{0.042} $\pm$ \textbf{0.001} & 0.279 $\pm$ 0.002 & 0.364 $\pm$ 0.001 & 0.185 $\pm$ 0.000 \\
\convMMD (Ours) &  0.044 $\pm$ 0.002 & \textbf{0.104} $\pm$ \textbf{0.003} & \textbf{0.164} $\pm$ \textbf{0.003} & \textbf{0.182} $\pm$ \textbf{0.000} \\ \hline
\end{tabular}}
\label{tab:ahs_results}
\begin{tabnote}
\textsc{mae}, mean absolute error; \textsc{simex}, simulation extrapolation; \convMMD, convolutional maximum mean discrepancy.
\end{tabnote}
\end{table}

\textit{Model and results.}
We model homeownership using a logistic regression model 
\begin{equation*}
\log\left\{\frac{\pr(Y=1 \mid X)}{1 - \pr(Y=1 \mid X)}\right\} = \alpha + \beta_{\text{age}} X_{\text{age}} + \beta_{\text{inc}} X_{\text{inc}}, 
\end{equation*} 
where $\alpha$ is the intercept, and $\beta_{age}$ and $\beta_{inc}$ represent the log-odds change associated with age and income, respectively. 
Both covariates are standardized. 
We generate $50$ independent Monte Carlo realizations. In each run, the data is split into training and testing sets. Noise is injected into the training features, while the testing set remains clean for model evaluation. 
We introduce Poisson noise to $X_{\text{age}}$ and Gaussian noise to $X_{\text{inc}}$ (details regarding noise generation are provided in Section~\ref{sec: sm real} of the Supplementary Material).  
{In our \convMMD framework, we have $\wt{\bfW} \equiv (\wt{X}_{\text{age}},\wt{X}_{\text{inc}},Y)\trans \equiv (\wt{\bfX}\trans,Y)\trans$ and $\wt{\bfZ} \equiv (\wt{\bfX}_{\text{sim}}\trans,Y_{\text{sim}})\trans$; and we assume $p(\bfX)$ to be bivariate normal.} 
We quantify performance using \textsc{mae} of the estimated coefficients relative to the true model trained on clean data. We also look at the Brier score to quantify the predictive performance of the methods \citep{glenn1950verification}. The results are shown in Table \ref{tab:ahs_results}. 
Our proposed method achieves a lower \textsc{mae} for regression parameter estimates ($\beta_{age}$ and $\beta_{inc}$) than the naive baseline and \textsc{simex}. Our method also yields a lower Brier score on the held-out clean test sets. Our results suggest that correcting for measurement error improves both parameter estimation and the model's generalization error on ground-truth data.

\section{Discussion}

We proposed a unified framework for parametric inference in the presence of classical measurement error. 
Our method, based on minimizing \MMD between a convolved model and the observed noisy data, provides a flexible alternative to classical likelihood and Fourier inversion-based approaches. 
We introduce a simulation-based, likelihood-free objective function, through which our estimator implicitly performs deconvolution without going through intractable integrals, making it applicable to a wide range of models and noise processes. The proposed approach establishes a robust, computationally tractable, and theoretically grounded framework for inference in noisy data across a wide range of scenarios.  A primary contribution of this work is establishing the theoretical foundations of the \convMMD estimator, proving its consistency and asymptotic normality. A key result is that the estimator achieves the parametric $N^{1/2}$ convergence rate. Notably, measurement error does not degrade the rate but instead results in an explicitly characterized inflation of the asymptotic variance, ensuring reliable statistical inference. Empirically, \convMMD is competitive with specialized likelihood-based benchmarks under Gaussian noise and demonstrates superior robustness when the noise is non-Gaussian. This flexibility is further validated on real-world applications, where our results align with established findings. 

While our work provides formal asymptotic theory absent in related \MMD-based approaches, our current exposition is limited to known noise models within a well-specified (M-closed) parametric setting. In a companion paper \citep{vashistha2026nonparametric}, we relax the well-specified parametric setting by extending \convMMD to nonparametric deconvolution and denoising, with finite-sample oracle inequalities and $L_{2}$ convergence rates. 
Our ongoing work aims to relax the known-noise assumption by incorporating replicate measurements to learn the noise distribution. 
We have not optimized kernel selection, relying instead on the standard median heuristic \citep[][]{gretton12a}. A formal procedure to select kernels that minimize the asymptotic variance of the \convMMD estimator remains an open research direction.

\section*{Acknowledgement}
We thank Dr. Antonio Linero, Dr. Connor Jerzak, and Ikjun Choi for their valuable suggestions. 
AF, RV, and JP are supported by NSF under Cooperative Agreement 2421782, and the Simons Foundation grant MPS-AI-00010515 awarded to NSF-Simons AI Institute for Cosmic Origins (CosmicAI, \texttt{https://www.cosmicai.org/}). AS was supported in part by NSF grant DMS-2515902. The authors declare no competing interests.

\section*{Supplementary material}
The Supplementary Material includes the proofs of the theoretical results, further details of the algorithm, simulations and real-data applications.

\bibliographystyle{natbib}
\bibliography{references}

\suppstart
\input{supplement_arxiv}

\bibliographystylesupp{natbib}
\bibliographysupp{references}

\end{document}

%% file: supplement_arxiv.tex
This Supplementary Material provides detailed proofs of the theoretical results; additional numerical illustrations of the main theorems discussed in the main paper; details of the main implementation algorithm; and additional information on the simulation experiments, and the real datasets discussed in the main paper.

\section{Proofs of theoretical results}

\begin{customlem}{\ref{lemma: identifiability}}[cancellation property] 
Let $p$, $q$ be two Borel probability measures defined on $\mathbb{R}^{d}$ (with supports potentially restricted to a subset $\calS \subseteq \mathbb{R}^{d}$). 
Let $m$ be a Borel probability measure on $\mathbb{R}^{d}$ such that
$m(\cdot) = \int r(\cdot \mid \bfphi)g(\bfphi \mid \bfpsi) d\bfphi$,  
where $r(\cdot  \mid \bfphi)$ is a family of Borel probability measures parametrized by a random variable  $\bfphi \in \Phi$ and $g(\bfphi \mid \bfpsi)$ is a probability measure on the parameter space $\Phi$. Then, $p = q$ if and only if $\E_{p * m} [f(\wt{\bfW})] = \E_{q * m} [f(\wt{\bfZ})]$ for all $f \in C(\mathbb{R}^{d})$, where $C(\mathbb{R}^{d})$ is the space of bounded continuous functions on $\mathbb{R}^{d}$.
\end{customlem}

\begin{proof}
Following Lemma 1 in  \citet{gretton12a}, we can say that the distributions $p*m = q*m$ if and only if $\E_{p * m} \bracket{f(\wt{\bfW})} = \E_{q * m} \bracket{f(\wt{\bfZ})}$ for all $f \in C(\mathbb{R}^{d})$, where $C(\mathbb{R}^{d})$ is the space of bounded continuous functions on $\mathbb{R}^{d}$. 
Now, the problem is further reduced to showing $p = q$ if and only if $p*m = q*m$. If $p = q$, then it easily follows that $p*m = q*m$. 

Now, suppose that $p*m = q*m$. 
Let $\mathcal{\varphi}_{p*m}(\bft)$ denote the characteristic function of the convolution $p*m$. 
We have
\begin{align*}
    \mathcal{\varphi}_{p*m}(\bft) &= \E\left(e^{i\bft\trans \wt{\bfW}}\right) 
                              =  \E\left(e^{i\bft\trans (\bfW + \bfU_{W})}\right) 
                              = \E\left(e^{i\bft\trans \bfW + i\bft\trans \bfU_{W}}\right) \\
                              &= \E\left(e^{i\bft\trans \bfW}\right) \cdot \E\left(e^{i\bft\trans \bfU_{W}}\right)  \text{ (using } \bfW \perp \!\!\! \perp \bfU_{W})\\
                              &= \mathcal{\varphi}_{p}(\bft) \cdot \mathcal{\varphi}_{m}(\bft).
\end{align*}
Here, $\mathcal{\varphi}_{p}(\bft)$ and $\mathcal{\varphi}_{m}(\bft)$ denote the characteristic function of $p$ and $m$, respectively. 
Similarly, it can be shown that $\mathcal{\varphi}_{q*m}(\bft)  = \mathcal{\varphi}_{q}(\bft) \cdot \mathcal{\varphi}_{m}(\bft)$ using $\bfZ \perp \!\!\! \perp \bfU_{Z}$. Since characteristic functions uniquely identify distributions, we have that 
\begin{align*}
    \mathcal{\varphi}_{p*m}(\bft) &= \mathcal{\varphi}_{q*m}(\bft) \text{ (} p*m = q*m)\\
    \mathcal{\varphi}_{p}(\bft) \cdot \mathcal{\varphi}_{m}(\bft) &= \mathcal{\varphi}_{q}(\bft) \cdot \mathcal{\varphi}_{m}(\bft) \\ 
    \mathcal{\varphi}_{p}(\bft) &= \mathcal{\varphi}_{q}(\bft) \\ 
 \implies p &= q.
\end{align*} 
In line two, we use Assumption \ref{assump:convolution-invertibility}, which says that the set of zeros $\mathcal{T}_{0} = \{\bft : \varphi_{m}(\bft) = 0\}$ has Lebesgue measure zero. Thus, its complement $\mathcal{T}_{0}^c$ is dense in $\mathbb{R}^{d}$ and cancellation holds. 
\begin{equation*}
     \varphi_{p}(\bft)\,\varphi_{m}(\bft) \;=\; \varphi_{q}(\bft)\,\varphi_{m}(\bft)
     \;\;\Longrightarrow\;\;
     \varphi_{p}(\bft) \;=\; \varphi_{q}(\bft).
\end{equation*}

Because characteristic functions are continuous, equality on a dense set implies equality everywhere. Thus, we have that $\varphi_{p}(\bft) = \varphi_{q}(\bft)$ for all $\bft\in\mathbb{R}^{d}$. 
By the usual uniqueness theorem for characteristic functions, that implies $p = q$.

\end{proof}

\begin{customthm}{\ref{th: sufficient condition}}
Under Assumptions \ref{assump:iid}--\ref{assump: mmd-assump},  $\convMMD(p,q,m) = 0$ if and only if $p = q$.
\end{customthm} 

\begin{proof}
First, let us assume $p = q$. In this case, $\convMMD(p*m,q*m)= 0$ since $p=q \Rightarrow p*m = q*m$. 

Now, suppose $\convMMD(p,q,m) \equiv \MMD(p*m,q*m) = 0$. 
Using Theorem 5 in \citet{gretton12a}, we can say that $p*m = q*m$. 
Additionally, using Lemma \ref{lemma: identifiability}, we can say that $p*m = q*m \Rightarrow p = q$. Thus, $\convMMD(p,q,m) = 0$ if and only if $p = q$.
\end{proof}

\begin{customthm}{\ref{th: mmd-equivalence}}
Let the kernel $k: \mathbb{R} \times \mathbb{R} \to \mathbb{R}$ be characteristic and translation-invariant such that $k(\bfw,\bfz) = \kappa(\bfw-\bfz)$ for some function $\kappa$. Under Assumptions \ref{assump:iid}--\ref{assump: mmd-assump}, the \MMD between the noisy distributions $p*m$ and $q*m$ with respect to the kernel $k$ is equal to the \MMD between the true, noiseless distributions $p$ and $q$ with respect to a modified kernel, $\wt{k}$:
\[
    \convMMD_k(p, q, m) = \MMD_{\wt{k}}(p, q),
\]
where the modified kernel $\wt{k}$ is defined as the expectation of the original kernel over the noise distribution:
\[
    \wt{k}(\bfw, \bfz) = \E_{\bfU_{W}, \bfU_{Z} \sim m}
    \left[k\left(\bfw+\bfU_{W}, \bfz+\bfU_{Z}\right)\right].
\]
Here, $\bfU_{W}$ and $\bfU_{Z}$ are i.i.d. random variables drawn from the noise distribution $m$.
\end{customthm}
\begin{proof}
We begin with the expression for the squared \MMD applied to the noisy data:
\[
\begin{aligned}
 \convMMD^2_k(p,q,m) = & \E_{\wt{\bfW}, \wt{\bfW}'}\left[k\left(\wt{\bfW}, \wt{\bfW}'\right)\right]+\E_{\wt{\bfZ}, \wt{\bfZ}'}\left[k\left(\wt{\bfZ}, \wt{\bfZ}'\right)\right] \\
 & -2 \E_{\wt{\bfW}, \wt{\bfZ}}\left[k\left(\wt{\bfW}, \wt{\bfZ}\right)\right].
\end{aligned}
\]
Consider the first term. The variables $\wt{\bfW} = \bfW+\bfU_{W}$ and $\wt{\bfW}' = \bfW'+\bfU'_{W}$ are drawn from the convolved distribution, where $\bfW,\bfW' \sim p$, and $\bfU_{W}, \bfU'_{W} \sim m$;  
similarly, $\wt{\bfZ} = \bfZ+\bfU_{Z}$ and $\wt{\bfZ}' = \bfZ'+\bfU'_{Z}$, where $\bfZ,\bfZ' \sim q$, and $\bfU_{Z}, \bfU'_{Z} \sim m$; 
and all variables are mutually independent. 
Using the law of total expectation, we have the following:
\[
\begin{aligned}
    \E_{\wt{\bfW}, \wt{\bfW}'}\left[k\left(\wt{\bfW}, \wt{\bfW}'\right)\right] &= \E_{\bfW, \bfW', \bfU_{W}, \bfU'_{W}}\left[k\left(\bfW + \bfU_{W}, \bfW' + \bfU'_{W} \right)\right] \\
    &= \E_{\bfW, \bfW' \sim p}\left[\E_{\bfU_{W}, \bfU'_{W} \sim m}\left[k\left(\bfW + \bfU_{W}, \bfW' + \bfU'_{W} \right)\right]\right].
\end{aligned}
\]
We define a new function, $\wt{k}(\bfw, \bfz)$, as:
\[
    \wt{k}(\bfw, \bfz) := \E_{\bfU_{W}, \bfU_{Z} \sim m}\left[k(\bfw+\bfU_{W}, \bfz+\bfU_{Z})\right].
\]
Since kernel $k$ is translation-invariant, we can write $k(\bfw+\bfU_{W}, \bfz+\bfU_{Z}) = \kappa((\bfw+\bfU_{W}) - (\bfz+\bfU_{Z})) = \kappa((\bfw-\bfz) + (\bfU_{W}-\bfU_{Z}))$. 
The new function $\wt{k}(\bfw, \bfz)$ can thus be expressed as:
\[
    \wt{k}(\bfw, \bfz) = \E_{\bfU_{W}, \bfU_{Z} \sim m}\left[\kappa(\left(\bfw-\bfz) + (\bfU_{W}-\bfU_{Z})\right)\right].
\]
We hence must show that this new function, $\wt{k}(\bfw, \bfz) = \wt{\kappa}(\bfw-\bfz)$, must also be a valid positive definite kernel; we use Bochner's theorem to demonstrate this formally.

Bochner's theorem states that a continuous, translation-invariant function $g(\bfw,\bfz) = \kappa(\bfw-\bfz)$ on $\mathbb{R}^{d}$ is a positive definite kernel if and only if the Fourier transform of $\kappa$ is a non-negative finite measure. Since our original kernel $k(\bfw,\bfz) = \kappa(\bfw-\bfz)$ is positive definite, the Fourier transform of $\kappa$, denoted $\calF\{\kappa\}(\bft)$, is non-negative for all $\bft \in \mathbb{R}^{d}$.

Our goal is to show that the Fourier transform of the new function, $\wt{\kappa}(\bfx) = \wt{k}(\bfx, \bfzero)$, is also non-negative. From its definition:
\[
    \wt{\kappa}(\bfx) = \E_{\bfU, \bfU' \sim m}\left[\kappa(\bfx + \bfU - \bfU')\right].
\]
Let $\bfdelta = \bfU - \bfU'$ be a new random variable representing the difference of two i.i.d. noise samples $\bfU,\bfU' \sim m$. 
Let the probability distribution of $\bfdelta$ be $m_{\delta}$. 
The expectation can then be written as an integral, which is the definition of a convolution:
\[
    \wt{\kappa}(\bfx) = \int_{\mathbb{R}^{d}} \kappa(\bfx - \bfdelta) \, dm_{\delta}(\bfdelta) = (\kappa * m_{\delta})(\bfx).
\]
We now apply the convolution theorem, which states that the Fourier transform of a convolution of two functions is the product of their individual Fourier transforms:
\[
    \calF\{\wt{\kappa}\}(\bft) = \calF\{\kappa* m_{\delta}\}(\bft) = \calF\{\kappa\}(\bft) \cdot \calF\{m_{\delta}\}(\bft).
\]
The second term, $\calF\{m_{\delta}\}(\bft)$, is the characteristic function of $\bfdelta$, which we also denote by $\varphi_{\delta}(\bft)$. 
Since $\bfdelta = \bfU - \bfU'$, where $\bfU$ and $\bfU'$, are i.i.d. with characteristic function $\varphi_{m}(\bft)$, we have:
\[
\begin{aligned}
    \varphi_{\delta}(\bft) &= \E[e^{i\bft\trans \bfdelta}] = \E[e^{i\bft\trans (\bfU - \bfU')}] = \E[e^{i\bft\trans \bfU} e^{-i\bft\trans \bfU'}] \\
    &= \E[e^{i\bft\trans \bfU}] \E[e^{-i\bft\trans \bfU'}] \quad (\text{by independence}) \\
    &= \varphi_{m}(\bft) \cdot \varphi_{m}(-\bft) \\
    &= \varphi_{m}(\bft) \cdot \overline{\varphi_{m}(\bft)} = |\varphi_{m}(\bft)|^2.
\end{aligned}
\]
The characteristic function of $\bfdelta$ is the squared magnitude of the characteristic function of the original noise distribution, which is always non-negative.

Substituting this back into the expression for the Fourier transform of $\wt{\kappa}$:
\[
    \calF\{\wt{\kappa}\}(\bft) = \calF\{\kappa\}(\bft) \cdot |\varphi_{m}(\bft)|^2.
\]
Since $\calF\{\kappa\}(\bft) \ge 0$ (by the positive definiteness of $k$) and $|\varphi_{m}(\bft)|^2 \ge 0$, their product is also non-negative for all $t$. By Bochner's theorem, this implies that $\wt{\kappa}$ corresponds to a positive semi-definite kernel. Thus, $\wt{k}$ is a valid kernel.

By substituting the definition of $\wt{k}$ into our expression, the first term of $\convMMD^2_k(p,q,m)$ simplifies to an expectation over the true data distribution $p$ with the new kernel:
\[
 \E_{\wt{\bfW}, \wt{\bfW}'}\left[k\left(\wt{\bfW}, \wt{\bfW}'\right)\right] = \E_{\bfW, \bfW' \sim p}\left[\wt{k}(\bfW, \bfW')\right].
\]
Similarly, we have:
\[
\begin{aligned}
    \E_{\wt{\bfZ}, \wt{\bfZ}'}\left[k\left(\wt{\bfZ}, \wt{\bfZ}'\right)\right] &= \E_{\bfZ, \bfZ' \sim q}\left[\wt{k}(\bfZ, \bfZ')\right], \\
    \E_{\wt{\bfW}, \wt{\bfZ}}\left[k\left(\wt{\bfW}, \wt{\bfZ}\right)\right] &= \E_{\bfW \sim p, \bfZ \sim q}\left[\wt{k}(\bfW, \bfZ)\right].
\end{aligned}
\]
Finally, substituting these simplified terms back into the \MMD formula gives:
\[
\begin{aligned}
 \convMMD^2_k(p,q,m) = & \E_{\bfW, \bfW' \sim p}\left[\wt{k}\left(\bfW, \bfW'\right)\right]+\E_{\bfZ, \bfZ' \sim q}\left[\wt{k}\left(\bfZ, \bfZ'\right)\right]  -2 \E_{\bfW \sim p, \bfZ \sim q}\left[\wt{k}\left(\bfW, \bfZ\right)\right].
\end{aligned}
\]
This is precisely the definition of the squared \MMD between the true distributions $p$ and $q$, calculated using the modified kernel $\wt{k}$. Therefore, $\convMMD_k(p,q,m) = \MMD_{\wt{k}}(p, q)$.

Now, we must also show that $\wt{k}$ is a characteristic kernel. We can show this by proving $\MMD^2_{\wt{k}}(p, q) = 0 \Leftrightarrow p=q$ \citep{sriperumbudur2010hilbert}.

From the main result of the proposition, we have:\[ \MMD^2_{\wt{k}}(p, q) = \convMMD^2_k(p,q,m) = \MMD^2_k(p*m, q*m).\] If $\MMD^2_{\wt{k}}(p, q) = 0$, this implies:\[ \MMD^2_k(p*m, q*m) = 0.\]

Now, the kernel $k$ is characteristic, and the \MMD between two measures with respect to a characteristic kernel is zero if and only if the two measures are identical. 
Therefore, we must have:
\[ p*m = q*m.\]

We can now reuse the argument made in Lemma \ref{lemma: identifiability} to show $p = q$. 
Let $\mathcal{\varphi}_{p*m}(\bft)$ denote the characteristic function of the convolution $p*m$. Now, we have
\begin{align*}
    \mathcal{\varphi}_{p*m}(\bft) &= \E\left(e^{i\bft\trans \wt{\bfW}}\right) 
                            =  \E\left(e^{i\bft\trans (\bfW + \bfU_{W})}\right) 
                            = \E\left(e^{i\bft\trans \bfW + i\bft\trans \bfU_{W}}\right) \\
                              &= \E\left(e^{i\bft\trans \bfW}\right) \cdot \E\left(e^{i\bft\trans \bfU_{W}}\right)  \text{ (using } \bfW \perp \!\!\! \perp \bfU_{W})\\
                              &= \mathcal{\varphi}_{p}(\bft)\cdot \mathcal{\varphi}_{m}(\bft).
\end{align*}
Here, $\mathcal{\varphi}_{p}(\bft)$ and $\mathcal{\varphi}_{m}(\bft)$ denote the characteristic function of $p$ and $m$, respectively. Similarly, it can be shown that $\mathcal{\varphi}_{q*m}(\bft) = \mathcal{\varphi}_{q}(\bft) \cdot \mathcal{\varphi}_{m}(\bft)$ using $\bfZ \perp \!\!\! \perp \bfU_{Z}$. Since characteristic functions uniquely identify distributions, we have that 
\begin{align*}
    \mathcal{\varphi}_{p*m}(\bft) &= \mathcal{\varphi}_{q*m}(\bft) \text{ (} p*m = q*m)\\
    \mathcal{\varphi}_{p}(\bft) \cdot \mathcal{\varphi}_{m}(\bft) &= \mathcal{\varphi}_{q}(\bft)\cdot \mathcal{\varphi}_{m}(\bft) \\ 
    \mathcal{\varphi}_{p}(\bft) &= \mathcal{\varphi}_{q}(\bft) \\ 
 \implies p &= q.
\end{align*} 
In line two, we use Assumption \ref{assump:convolution-invertibility}, which says that the set of zeros $\mathcal{T}_{0} = \{\bft : \varphi_{m}(\bft) = 0\}$ has Lebesgue measure zero. Thus, its complement $\mathcal{T}_{0}^c$ is dense in $\mathbb{R}^{d}$ and cancellation holds. 
\begin{equation*}
     \varphi_{p}(\bft)\,\varphi_{m}(\bft) \;=\; \varphi_{q}(\bft)\,\varphi_{m}(\bft)
     \;\;\Longrightarrow\;\;
     \varphi_{p}(\bft) \;=\; \varphi_{q}(\bft).
\end{equation*}

Because characteristic functions are continuous, equality on a dense set implies equality everywhere. Thus, we have that $\varphi_{p}(\bft) = \varphi_{q}(\bft)$ for all $\bft\in\mathbb{R}^{d}$. 
By the usual uniqueness theorem for characteristic functions, this implies $p = q$.

Thus, we have:\[ \MMD^2_{\wt{k}}(p, q) = 0 \implies \MMD^2_k(p*m, q*m) = 0 \implies p*m = q*m \implies p = q.\]
\\ 
Now if $p = q$, this implies: 
\[
\MMD^2_{\wt{k}}(p, q) = \MMD^2_{\wt{k}}(p, p)  = \convMMD^2_k(p,p,m) = \MMD^2_k(p*m, p*m) = 0.
\]
Since the kernel $k$ is characteristic, the \MMD between two measures with respect to a characteristic kernel is zero if and only if the two measures are identical. Therefore, we must have that:
\[
p = q \Rightarrow \MMD^2_{\wt{k}}(p, q) = 0.
\]
Thus, kernel $\wt{k}$ is characteristic and $\MMD^2_{\wt{k}}(p, q) = 0 \Leftrightarrow p =q$.
\end{proof}

\begin{customthm}{\ref{th:hypothesis_testing}} [Large deviation bound for estimation error]
Under Assumptions \ref{assump:iid}--\ref{assump: mmd-assump}, for any $\gamma \in (0, 1)$, we have with probability at least $1 - \gamma$: 
\begin{equation*}
\left|\convMMD(p,q,m) -\wh{\convMMD_b}(p,q,m) \right|
\leq (16 K / N)^{1/2}\left\{1 + \left(\frac{1}{4}\log \frac{2}{\gamma}\right)^{1/2}\right\}. \nonumber
\end{equation*}
\end{customthm}
\begin{proof}
\citet{gretton12a} derived a similar  bound for $\wh{\MMD_b}(p,q)$. For any $\gamma \in (0, 1)$, we have with probability at least $1 - \gamma$:
\[
\left|\MMD(p,q) -\wh{\MMD_b}(p,q) \right|
\leq (16 K / N)^{1/2}\left\{1 + \left(\frac{1}{4}\log \frac{2}{\gamma}\right)^{1/2}\right\}.
\]
Since the bound depends only on $K$, we can directly use it and say:
for any $\gamma \in (0, 1)$, we have with probability at least $1 - \gamma$: 
\begin{equation*}
\left|\convMMD(p,q,m) -\wh{\convMMD_b}(p,q,m) \right|
\leq (16 K / N)^{1/2}\left\{1 + \left(\frac{1}{4}\log \frac{2}{\gamma}\right)^{1/2}\right\}.
\end{equation*}
\end{proof}
\begin{theorem}[Large deviation bound for noise induced shift]\label{th: noise-shift}
Under Assumptions \ref{assump:iid}--\ref{assump: mmd-assump} and assuming $k(\cdot,\cdot)$ is translation invariant such that $k(\wt{\bfW},\wt{\bfZ}) = k_{0}(\wt{\bfW} - \wt{\bfZ})$ where $k_{0}$ is a Lipschitz continuous function with constant $L_{k}$,
we have
\begin{equation*}
    \left|\MMD^2(p,q) - \convMMD^2(p,q,m)\right|
\leq 4L_{k}\{2\E_{\bfphi \sim g(.|\bfpsi)}(\alpha(\bfphi))\}^{1/2}.
\end{equation*}
\end{theorem} 
\begin{proof}
We recall 
\[
\begin{aligned}
\MMD^2(p,q) &= \E_{p}\left[k\left(\bfW, \bfW'\right)\right]+\E_{q}\left[k\left(\bfZ, \bfZ'\right)\right] -2 \E_{p, q}\left[k\left(\bfW \cdot \bfZ\right)\right], \\
\convMMD^2(p,q,m) & =  \E_{p*m}\left[k\left(\wt{\bfW}, \wt{\bfW}'\right)\right]+\E_{q*m}\left[k\left(\wt{\bfZ}, \wt{\bfZ}'\right)\right] -2 \E_{p*m, q*m}\left[k\left(\wt{\bfW}, \wt{\bfZ}\right)\right] .
\end{aligned}
\]
Thus, the difference between the two can be expressed as
\[
\begin{aligned}
\left|\MMD^2(p,q) - \convMMD^2(p,q,m)\right| = |\Delta_{W,W} + \Delta_{Z,Z} + 2\Delta_{W,Z}| \leq |\Delta_{W,W}| + |\Delta_{Z,Z}| + 2|\Delta_{W,Z}|,
\end{aligned}
\]
where 
\[
\begin{aligned}
\Delta_{W,W} &= \E_{p}[k(\bfW, \bfW')] - \E_{p*m}\left[k\left(\wt{\bfW}, \wt{\bfW}'\right)\right], \\
\Delta_{Z,Z} &= \E_{q}[k(\bfZ, \bfZ')] - \E_{q*m}\left[k\left(\wt{\bfZ}, \wt{\bfZ}'\right)\right], \\
\Delta_{W,Z} &= \E_{p,q}[k(\bfW,\bfZ)] - \E_{p*m,q*m}\left[k\left(\wt{\bfW}, \wt{\bfZ}\right)\right].
\end{aligned}
\]

Now, we focus on bounding $\Delta_{W,W}$, $\Delta_{Z,Z}$ and $\Delta_{W,Z}$. Let us consider $\Delta_{W,W}$ first
\begin{align*}
\Delta_{W,W} &= \E_{p}[k(\bfW, \bfW')] - \E_{p*m}\left[k\left(\wt{\bfW}, \wt{\bfW}'\right)\right] \nonumber \\ 
&= \E_{\bfW, \bfW'\sim p}[k(\bfW, \bfW')] - \E_{\tX, \tX' \sim p*m}\left[k\left(\wt{\bfW}, \wt{\bfW}'\right)\right] \ \nonumber \\
&= \E_{\bfW, \bfW'\sim p}[k(\bfW, \bfW')] - \E_{\bfW, \bfW' \sim p, \bfU_{W}, \bfU_{W'} \sim r(\cdot \mid \bfphi), \bfphi \sim g(\bfphi \mid \bfpsi)}\left[k\left(\bfW + \bfU_{W}, \bfW' + \bfU'_{W}\right)\right] \nonumber \\
&= \E_{\bfW, \bfW'\sim p}\left[k(\bfW, \bfW') - \E_{\bfU_{W}, \bfU_{W'} \sim r(\cdot \mid \bfphi), \bfphi \sim g(\bfphi \mid \bfpsi)}[k(\bfW + \bfU_{W}, \bfW' + \bfU'_{W})] \right] \nonumber \\
&= \E_{\bfW, \bfW'\sim p}\left[ \E_{\bfU_{W}, \bfU_{W'} \sim r(\cdot \mid \bfphi), \bfphi \sim g(\bfphi \mid \bfpsi)}\left[k(\bfW, \bfW') - k(\bfW + \bfU_{W}, \bfW' + \bfU'_{W})\right] \right].
\end{align*}
The problem has now narrowed down to bounding the kernel difference. 
If k is translation-invariant, we can write \[
\begin{aligned}
k(\bfW,\bfW') &= k_{0}(\bfW-\bfW') = k_{0}(\bfX), \\
k(\bfW + \bfU_{W}, \bfW' + \bfU'_{W}) &= k_{0}\left\{(\bfW+\bfU_{W}) - (\bfW'+\bfU'_{W})\right\} = k_{0}(\bfX+\bfdelta),
\end{aligned}
\] where $k_{0}(\cdot)$ is a function of the relative difference with $\bfX = \bfW-\bfW'$ and $\bfdelta = \bfU_{W} - \bfU'_{W}$. 
This allows us to rewrite the kernel difference as follows
\[
k(\bfW,\bfW') - k(\bfW + \bfU_{W}, \bfW' + \bfU'_{W}) =  k_{0}(\bfX) - k_{0}(\bfX+\bfdelta).
\]
Furthermore, if $k_{0}(\cdot)$ is Lipschitz with constant $L_{k}$ 
\[
\begin{aligned}
|k_{0}(\bfX) - k_{0}(\bfX+\bfdelta)| &\leq L_{k}\|\bfdelta\| \\
\Rightarrow \E(|k_{0}(\bfX) - k_{0}(\bfX+\bfdelta)|) &\leq L_{k} \E[\|\bfdelta\|] \leq L_{k} (\E[\|\bfdelta\|^2])^{1/2} \text{ (using Jensen's inequality)}.
\end{aligned}
\]
Now, $\bfdelta = \bfU_{W} - \bfU'_{W}$ and $\bfU_{W},\bfU'_{W} \sim r(\cdot \mid \bfphi)$, where $\bfphi \sim g(\bfphi \mid \bfpsi)$. 
Suppose $\E(\bfU_{W}) = \E(\bfU'_{W}) = \bfzero$ and from Assumption \ref{assump:finite-moment-noise}, we know $\E[\|\bfU_{W}\|^{2} \mid \bfphi] = \E[\|\bfU'_{W}\|^{2} \mid \bfphi] = \alpha(\bfphi)$, then:
\[
\begin{aligned}
\E[\|\bfdelta\|^2 \mid \bfphi] &= \E[\|\bfU_{W} - \bfU'_{W}\|^2 \mid \bfphi] = \E[\|\bfU_{W}\|^2 + \|\bfU'_{W}\|^2 \mid \bfphi] = 2\alpha(\bfphi) \\
\Rightarrow \E(|k_{0}(\bfX) - k_{0}(\bfX+\bfdelta)|) &\leq  L_{k} (\E[\|\bfdelta\|^2])^{1/2}  = L_{k}\{2\E_{\bfphi \sim g(\bfphi \mid \bfpsi)}(\alpha(\bfphi))\}^{1/2}   .
\end{aligned}
\]
Using the bound derived above, we can write $|\Delta_{W,W}|$ as follows:
\begin{align*}
    |\Delta_{W,W}| &=  \left|\E_{\bfW, \bfW'\sim p}\left[ \E_{\bfU_{W}, \bfU_{W'} \sim r(\cdot \mid \bfphi), \bfphi \sim g(\bfphi \mid \bfpsi)}\left[k(\bfW, \bfW') - k(\bfW + \bfU_{W}, \bfW' + \bfU'_{W})\right] \right]\right| \nonumber \\
    &\leq \E_{\bfW, \bfW'\sim p}\left[ \E_{\bfU_{W}, \bfU_{W'} \sim r(\cdot \mid \bfphi), \bfphi \sim g(\bfphi \mid \bfpsi)}\left[\left|k(\bfW, \bfW') - k(\bfW + \bfU_{W}, \bfW' + \bfU'_{W})\right|\right]\right] \nonumber \\
    &\leq \E_{\bfW, \bfW'\sim p}[L_{k}\{2\E_{\bfphi \sim g(\bfphi \mid \bfpsi)}(\alpha(\bfphi))\}^{1/2}   ] \nonumber \\
   \Rightarrow |\Delta_{W,W}|  & \leq L_{k}\{2\E_{\bfphi \sim g(\bfphi \mid \bfpsi)}(\alpha(\bfphi))\}^{1/2}  .  \nonumber
\end{align*}
Similarly, we can derive bounds for $|\Delta_{Z,Z}|$ and $|\Delta_{W,Z}|$, and since they are symmetrical, we have the following:
\begin{align*}
\left|\MMD^2(p,q) - \convMMD^2(p,q,m)\right| \leq 4L_{k}\{2\E_{\bfphi \sim g(\bfphi \mid \bfpsi)}(\alpha(\bfphi))\}^{1/2} .
\end{align*}
\end{proof}

\begin{customthm}{\ref{th: mmd-var}}
Let $p\!=\!q$, Assumptions \ref{assump:iid}--\ref{assump: mmd-assump} hold, and $k(\wt{\bfW},\wt{\bfZ})\!=\!k_0(\wt{\bfW}\!-\!\wt{\bfZ})$ be a translation invariant kernel, where $k_{0}$ is an $L_{k}$-Lipschitz continuous function, and a finite non-zero variance of $\wh{\MMD_u^2}(p,q)$. Then, for $(\bfW_{i},\bfZ_{i}) \sim p \times q$ and $\delta = [32 L_{k}^{2} \E_{\bfphi \sim g(\bfphi \mid \bfpsi)}\{\alpha(\bfphi)\}]^{1/2}$, we have
\[
\E\left[\left\{\wh{\convMMD_{u}^{2}}(p,q,m)\right\}^2\right] \leq \E\left[\left\{\wh{\MMD_u^2}(p,q)\right\}^2\right] + \frac{2\delta(\delta + 8K)}{N(N-1)}.
\]
\end{customthm}

\begin{proof}
It is shown in \citet{gretton12a} that:
\begin{equation}
\E\left[\left\{\wh{\MMD_u^2}(p,q)\right\}^2\right] = \frac{2}{N(N-1)}\E_{\bfV,\bfV'}\left[h\left(\bfV,\bfV'\right)^2\right], \label{eq: mmd_var}
\end{equation}
where $\bfV: (\bfW,\bfZ) \sim p \times q$, $\bfV'$ is an independent copy of $\bfV$, and $h(\bfV, \bfV')$ is the kernel of the one-sample U-statistic $\wh{\MMD_u^2}(p,q) = \{N(N-1)\}^{-1}\sum_{i \neq j} h(\bfV_{i},\bfV_{j})$, with $\bfV_{i}=(\bfW_{i},\bfZ_{i}) \sim p \times q$:

\begin{equation}
h(\bfV, \bfV') = k(\bfW, \bfW') + k(\bfZ, \bfZ') - k(\bfW, \bfZ') - k(\bfZ, \bfW'). \label{eq: u-stat}
\end{equation}
Similarly, we can write the second moment for $\wh{\convMMD_{u}^2}(p,q,m)$:
\[
\E\left[\left\{\wh{\convMMD_{u}^{2}}(p,q, m)\right\}^2\right] = \frac{2}{N(N-1)}\E_{\wt{\bfV},\wt{\bfV'}}\left[h\left(\wt{\bfV},\wt{\bfV'}\right)^2\right],
\]
where $\wt{\bfV}: (\wt{\bfW},\wt{\bfZ}) \sim p*m \otimes q*m$ and $h(\wt{\bfV},\wt{\bfV'})$ is the one-sample U-statistic defined in Equation~\eqref{eq: u-stat}. Under the null hypothesis, $p = q$, $\bfV \sim p $, and $\wt{\bfV} \sim p*m$. 
Thus, we can say 
\[
\frac{\E\left[\left\{\wh{\convMMD_{u}^{2}}(p,q, m)\right\}^2\right]}{\E\left[\left\{\wh{\MMD_u^2}(p,q)\right\}^2\right]} = \frac{\E_{\wt{\bfV},\wt{\bfV'}}\left[h\left(\wt{\bfV},\wt{\bfV'}\right)^2\right]}{\E_{\bfV,\bfV'}\left[h\left(\bfV,\bfV'\right)^2\right]}.
\]
Now, we can rewrite $h\left(\wt{\bfV},\wt{\bfV'}\right)$ as
\[
\begin{aligned}
h\left(\wt{\bfV},\wt{\bfV'}\right) &= h\left(\bfV, \bfV'\right) + \left(h\left(\wt{\bfV},\wt{\bfV'}\right) - h\left(\bfV, \bfV'\right)\right) \\
    \Rightarrow h\left(\wt{\bfV},\wt{\bfV'}\right)^2 &= \left[h\left(\bfV, \bfV'\right) + \Delta'' \right]^2 = h\left(\bfV, \bfV'\right)^2 + \Delta''^2 + 2h\left(\bfV, \bfV'\right)\Delta'' \\
    \Rightarrow \E\left[ h\left(\wt{\bfV},\wt{\bfV'}\right)^2 \right] &= \E\left[h\left(\bfV, \bfV'\right)^2\right]+ \E\left(\Delta''^2\right) + 2\E\left(h\left(\bfV, \bfV'\right)\Delta''\right),
\end{aligned}
\]
where $\Delta'' = h\left(\wt{\bfV},\wt{\bfV'}\right) - h\left(\bfV, \bfV'\right)$. This implies that,
\[
\begin{aligned}
\frac{\E\left[\left\{\wh{\convMMD_{u}^{2}}(p,q,m)\right\}^2\right]}{\E\left[\left\{\wh{\MMD_u^2}(p,q)\right\}^2\right]} &= 1 + \frac{\E\left(\Delta''^2\right) + 2\E\left(h\left(\bfV, \bfV'\right)\Delta''\right)}{\E\left[h\left(\bfV, \bfV'\right)^2\right]}.
\end{aligned}
\]
Now, using the Cauchy-Schwarz inequality, we can say that
\[
\left|\E\left(h\left(\bfV, \bfV'\right)\Delta''\right)\right| \leq \left\{\E\left[h\left(\bfV, \bfV'\right)^2\right]\right\}^{1/2} \cdot \left\{\E\left(\Delta''^2\right)\right\}^{1/2}.
\]
Thus, we can say
\begin{equation}
\frac{\E\left[\left\{\wh{\convMMD_{u}^{2}}(p,q,m)\right\}^2\right]}{\E\left[\left\{\wh{\MMD_u^2}(p,q)\right\}^2\right]} \leq 1 + \frac{\E\left(\Delta''^2\right) + 2\left\{\E\left[h\left(\bfV, \bfV'\right)^2\right]\right\}^{1/2} \cdot \left\{\E\left(\Delta''^2\right)\right\}^{1/2}}{\E\left[h\left(\bfV, \bfV'\right)^2\right]}. \label{eq: var_ratio}
\end{equation}
We can rewrite $\Delta''$ as
\[
\begin{aligned}
    \Delta'' &= [k\left(\wt{\bfW}, \wt{\bfW}'\right) - k\left(\bfW, \bfW'\right)] + [k\left(\wt{\bfZ}, \wt{\bfZ}'\right) - k\left(\bfZ, \bfZ'\right)] \\
    &\quad\quad - {[k\left(\wt{\bfW}, \wt{\bfZ}'\right) - k\left(\bfW, \bfZ'\right)] - [k\left(\wt{\bfZ}, \wt{\bfW}'\right) - k\left(\bfZ, \bfW'\right)]}.
\end{aligned}
\]
Using Lipschitz continuity and translation invariance, as in Theorem \ref{th: noise-shift}, we can say $\left|k\left(\wt{\bfW}, \wt{\bfW}'\right) - k\left(\bfW, \bfW'\right)\right| \leq L_{k} \| \bfU_{W} - \bfU'_{W} \|$, where $L_{k}$ is the Lipschitz constant; the other three terms are bounded in the same way by $L_{k}\|\bfU_{Z}-\bfU'_{Z}\|$, $L_{k}\|\bfU_{W}-\bfU'_{Z}\|$ and $L_{k}\|\bfU_{Z}-\bfU'_{W}\|$.
Using $(a_{1}+a_{2}+a_{3}+a_{4})^{2} \leq 4(a_{1}^{2}+a_{2}^{2}+a_{3}^{2}+a_{4}^{2})$, we have
\[
\Delta''^{2} \leq 4 L_{k}^{2} \left\{\| \bfU_{W} - \bfU'_{W} \|^{2} + \| \bfU_{Z} - \bfU'_{Z} \|^{2} + \| \bfU_{W} - \bfU'_{Z} \|^{2} + \| \bfU_{Z} - \bfU'_{W} \|^{2}\right\}.
\]
In each term the two noise vectors come from different observations, so they are independent, and both have distribution $m$ with $\E\|\bfU\|^{2} = \E_{\bfphi \sim g(\bfphi \mid \bfpsi)}\left(\alpha(\bfphi)\right)$. Hence $\E\|\bfU - \bfU'\|^{2} \leq 2 \E_{\bfphi \sim g(\bfphi \mid \bfpsi)}\left(\alpha(\bfphi)\right)$, with equality for mean-zero noise, and
\[
\E\left(\Delta''^2\right) \leq 4 L_{k}^{2} \cdot 4 \cdot 2 \E_{\bfphi \sim g(\bfphi \mid \bfpsi)}\left(\alpha(\bfphi)\right) = 32 L_{k}^{2} \E_{\bfphi \sim g(\bfphi \mid \bfpsi)}\left(\alpha(\bfphi)\right).
\]
Under Assumption \ref{assump: mmd-assump}, the kernel $k$ is bounded such that $|k(\cdot, \cdot)| \leq K$. 
By the triangle inequality, the function $h(\bfV, \bfV') = k(\bfW, \bfW') + k(\bfZ, \bfZ') - k(\bfW, \bfZ') - k(\bfZ, \bfW')$ is bounded such that $|h(\bfV, \bfV')| \leq 4K$. 
It follows that $\{\E[h(\bfV, \bfV')^2]\}^{1/2} \leq 4K$. We use this to bound the cross-term in Equation~\eqref{eq: var_ratio} via the Cauchy-Schwarz inequality:
\begin{equation*}
2\E\left[h(\bfV, \bfV')\Delta''\right] \leq 2\{\E[h(\bfV, \bfV')^2]\}^{1/2} (\E[\Delta''^2])^{1/2} \leq 8K(\E[\Delta''^2])^{1/2}.
\end{equation*}
Now, we substitute the bound $\E[\Delta''^2] \leq 32 L_{k}^{2} \E_{\bfphi \sim g(\bfphi \mid \bfpsi)}[\alpha(\bfphi)]$ into Equation~\eqref{eq: var_ratio}:
\[
\begin{aligned}
\frac{\E\left[\left\{\wh{\convMMD_{u}^{2}}(p,q,m)\right\}^2\right]}{\E\left[\left\{\wh{\MMD_u^2}(p,q)\right\}^2\right]} &\leq 1 + \frac{32 L_{k}^{2} \E_{\bfphi \sim g(\bfphi \mid \bfpsi)}\left(\alpha(\bfphi)\right) + 8 K \left\{32 L_{k}^{2} \E_{\bfphi \sim g(\bfphi \mid \bfpsi)}\left(\alpha(\bfphi)\right)\right\}^{1/2}}{\E_{\bfV,\bfV'}\left[h\left(\bfV,\bfV'\right)^2\right]} \\
    \Rightarrow  \E\left[\left\{\wh{\convMMD_{u}^{2}}(p,q,m)\right\}^2\right] &\leq \E\left[\left\{\wh{\MMD_u^2}(p,q)\right\}^2\right] + \frac{\E\left[\left\{\wh{\MMD_u^2}(p,q)\right\}^2\right]}{\E_{\bfV,\bfV'}(h(\bfV, \bfV')^2)} \times \\ & \left ( 32 L_{k}^{2} \E_{\bfphi \sim g(\bfphi \mid \bfpsi)}\left(\alpha(\bfphi)\right)+ 8 K \left\{32 L_{k}^{2} \E_{\bfphi \sim g(\bfphi \mid \bfpsi)}\left(\alpha(\bfphi)\right)\right\}^{1/2} \right) \\
   \Rightarrow  \E\left[\left\{\wh{\convMMD_{u}^{2}}(p,q,m)\right\}^2\right] &\leq \E\left[\left\{\wh{\MMD_u^2}(p,q)\right\}^2\right] + \\ &\frac{2}{N(N-1)}\left ( 32 L_{k}^{2} \E_{\bfphi \sim g(\bfphi \mid \bfpsi)}\left(\alpha(\bfphi)\right) \right. \\ & \qquad \left. {} + 8 K \left\{32 L_{k}^{2} \E_{\bfphi \sim g(\bfphi \mid \bfpsi)}\left(\alpha(\bfphi)\right)\right\}^{1/2} \right),
\end{aligned}
\]
using Equation~\eqref{eq: mmd_var}, where $\E\left[\left\{\wh{\MMD_{u}^{2}}(p,q)\right\}^2\right] = (2/\{N(N-1)\})\E_{\bfV,\bfV'}\left[h\left(\bfV,\bfV'\right)^2\right]$.

Since $32 L_{k}^{2} \E_{\bfphi \sim g(\bfphi \mid \bfpsi)}\left(\alpha(\bfphi)\right) = \delta^{2}$, we have
\[
\E\left[\left\{\wh{\convMMD_{u}^{2}}(p,q,m)\right\}^2\right] \leq \E\left[\left\{\wh{\MMD_u^2}(p,q)\right\}^2\right] + \frac{2\delta(\delta + 8K)}{N(N-1)}.
\] 
\end{proof}

\begin{customprop}{\ref{prop: score-gradient}}
Suppose that $q_{\bftheta}$ has a valid density with respect to Lebesgue measure and $\bftheta \rightarrow q_{\bftheta}(\bfw)$ is differentiable with respect to $\bftheta$ for all $\bfw \in \calS$. Then, under regularity conditions that permit the exchange of differentiation and integration, the gradient of the \MMD objective with respect to $\bftheta$ is given by:
\begin{align*}
\nabla_{\bftheta} L_{N}(\bftheta)  &= \nabla_{\bftheta}\left[\E_{\wt{\bfZ}, \wt{\bfZ}'}\left[k\left(\wt{\bfZ}, \wt{\bfZ}'\right)\right] -\frac{2}{N} \sum_{i=1}^{N}\E_{ \wt{\bfZ}}\left[k\left(\wt{\bfw}_{i}, \wt{\bfZ}\right)\right]\right] \nonumber \\
&=2\E_{\bfZ,\wt{\bfZ}, \wt{\bfZ}'}\left[\left(k(\wt{\bfZ},\wt{\bfZ}')- \frac{1}{N} \sum_{i=1}^{N}k(\wt{\bfZ}, \wt{\bfw}_{i})\right)\nabla_{\bftheta}\left[\log q_{\bftheta}(\bfZ)\right]\right],
\end{align*}
where $\bfZ \sim q_{\bftheta}$  and $\wt{\bfZ} \sim q_{\bftheta}*m$ . 
\end{customprop}

\begin{proof}
We start by taking the gradient of the \MMD expression. The first term, $N^{-2}\sum_{i,j=1}^{N} k(\wt{\bfw}_{i},\wt{\bfw}_{j})$, does not depend on $\bftheta$, so its gradient is zero. We thus have:
\begin{align*}
\nabla_{\bftheta} L_{N}(\bftheta)   &= \nabla_{\bftheta}\left[\E_{\wt{\bfZ}, \wt{\bfZ}'}\left[k\left(\wt{\bfZ}, \wt{\bfZ}'\right)\right] -\frac{2}{N} \sum_{i=1}^{N}\E_{ \wt{\bfZ}}\left[k\left(\wt{\bfw}_{i}, \wt{\bfZ}\right)\right]\right].
\end{align*}
We compute the gradient for each expectation term using the log-derivative trick, that is, $\nabla_{\bftheta} \E_{\bfZ \sim q_{\bftheta}}[f(\bfZ)] = \E_{\bfZ \sim q_{\bftheta}}[f(\bfZ)\nabla_{\bftheta} \log q_{\bftheta}(\bfZ)]$.

For the first term, with $\wt{\bfZ}=\bfZ+\bfU_{Z}$ and $\wt{\bfZ}'=\bfZ'+\bfU'_{Z}$, we thus have: 
\begin{align*}
\nabla_{\bftheta} \E_{\wt{\bfZ}, \wt{\bfZ}'}\left[k(\wt{\bfZ}, \wt{\bfZ}')\right] &= \nabla_{\bftheta} \iint k(\wt{\bfz}, \wt{\bfz}') (q_{\bftheta}*m)(\wt{\bfz}) (q_{\bftheta}*m)(\wt{\bfz}') d\wt{\bfz} d\wt{\bfz}' \\
&= \E_{\wt{\bfZ}, \wt{\bfZ}'}\left(k(\wt{\bfZ}, \wt{\bfZ}') [\nabla_{\bftheta} \log\{ (q_{\bftheta}*m)(\wt{\bfZ})\} + \nabla_{\bftheta} \log\{ (q_{\bftheta}*m)(\wt{\bfZ}')\}]\right).
\end{align*}
By symmetry, this equals $2\E_{\wt{\bfZ}, \wt{\bfZ}'}[k(\wt{\bfZ}, \wt{\bfZ}') \nabla_{\bftheta} \log\{(q_{\bftheta}*m)(\wt{\bfZ})\}]$. Using the fact that $\wt{\bfZ}=\bfZ+\bfU_{Z}$, we can express the score function of the convolved distribution in terms of the original one: $\nabla_{\bftheta} \log\{(q_{\bftheta}*m)(\wt{\bfZ})\} = \E_{\bfZ|\wt{\bfZ}}[\nabla_{\bftheta} \log q_{\bftheta}(\bfZ)]$. The gradient term becomes:
\[
2\E_{\wt{\bfZ}, \wt{\bfZ}'}[k(\wt{\bfZ}, \wt{\bfZ}') \E_{\bfZ|\wt{\bfZ}}\{\nabla_{\bftheta} \log q_{\bftheta}(\bfZ)\}] = 2\E_{\bfZ, \wt{\bfZ}, \wt{\bfZ}'}[k(\wt{\bfZ}, \wt{\bfZ}') \nabla_{\bftheta} \log q_{\bftheta}(\bfZ)].
\]
For the second term, we apply the same logic
\[
\nabla_{\bftheta} \E_{\wt{\bfZ}}[k(\wt{\bfw}_{i}, \wt{\bfZ})] = \E_{\wt{\bfZ}}[k(\wt{\bfw}_{i}, \wt{\bfZ}) \nabla_{\bftheta} \log\{(q_{\bftheta}*m)(\wt{\bfZ})\}] = \E_{\bfZ, \wt{\bfZ}}[k(\wt{\bfw}_{i}, \wt{\bfZ}) \nabla_{\bftheta} \log q_{\bftheta}(\bfZ)].
\]
Combining the terms yields the final expression
\begin{align*}
\nabla_{\bftheta} L_{N}(\bftheta)  &= 2\E\left[k(\wt{\bfZ}, \wt{\bfZ}')\nabla_{\bftheta} \log q_{\bftheta}(\bfZ)\right] - \frac{2}{N}\sum_{i=1}^{N} \E\left[k(\wt{\bfw}_{i}, \wt{\bfZ})\nabla_{\bftheta} \log q_{\bftheta}(\bfZ)\right] \\
&= 2\E\left[\left(k(\wt{\bfZ}, \wt{\bfZ}') - \frac{1}{N}\sum_{i=1}^{N} k(\wt{\bfw}_{i}, \wt{\bfZ})\right)\nabla_{\bftheta} \log q_{\bftheta}(\bfZ)\right].
\end{align*}
\end{proof}

\begin{customlem}{\ref{lemma: identifiability-learning}}
Suppose $\bftheta_{0} = \underset{\bftheta \in \bfTheta}{\arg \min}\ {\convMMD}^2(p, q_{\bftheta}, m)$, there exists a $\bftheta^{\star}$ such that $q_{\bftheta^{\star}} = p$, and $\bftheta \mapsto q(\bftheta)$ is injective and continuous for almost every $\bftheta \in \bfTheta$. Under the conditions of Lemma \ref{lemma: identifiability} and Assumptions \ref{assump:iid}--\ref{assump: mmd-assump}, ${\bftheta_{0}} = {\bftheta^{\star}}$ if and only if $\convMMD(p, q_{\bftheta_{0}}, m) = 0$.
\end{customlem}

\begin{proof}
Suppose ${{\bftheta_{0}}} = {\bftheta^{\star}}$. This implies that
\begin{align*}
    \convMMD^2(p, q_{{\bftheta_{0}}},m) &= \convMMD^2(p, q_{{\bftheta^{\star}}},m) = \convMMD^2(p, p, m) = 0.
\end{align*}
Now, suppose $\convMMD^2(p, q_{{\bftheta_{0}}}, m)  =0$, using Theorem~\ref{th:hypothesis_testing}, we can say that  $p  = q_{\bftheta^{\star}} = q_{{\bftheta_{0}}}$. Since $\bftheta \mapsto q(\bftheta)$ is injective and continuous for almost every $\bftheta \in \bfTheta$, we have that ${\bftheta_{0}} = \bftheta^{\star}$.
\end{proof}

\begin{customthm}{\ref{th: bound and consistency}}[Generalization bound and consistency]
Suppose that a unique minimizer $\bftheta^{\star} \in \bfTheta$ exists such that $\bftheta^{\star} = \underset{\bftheta \in \bfTheta}{\inf}\convMMD\left(p, q_{\bftheta}, m\right)$, then under Assumptions \ref{assump:iid}--\ref{assump: mmd-assump} and Lemma \ref{lemma: identifiability-learning}, we have
\begin{enumerate}
    \item (Generalization bound) With probability $1 - \gamma$, we have
        \[
        \convMMD\left(p, q_{\wh{\bftheta}_{N}}, m\right) \leq \underset{\bftheta \in \bfTheta}{\inf}\convMMD\left(p, q_{\bftheta}, m\right)+  4\left(\frac{2K}{N}\right)^{1/2} \left(2+\{\log (1 / \gamma)\}^{1/2}\right),
        \]
    \item (Consistency) The empirical \convMMD estimator $\wh{\bftheta}_{N}$ converges almost surely to $\bftheta^{\star}$ if it is fitted on noisy data and samples of size $N$ from $q_{\bftheta}*m$. That is $\underset{N \rightarrow \infty}{\lim} \wh{\bftheta}_{N} = \bftheta^{\star}$ almost surely. 
\end{enumerate}
\end{customthm}
\begin{proof}
The bound and consistency result follows directly from  Theorem 1 and Proposition 1 in \citet{briol2019statistical}. 
They derived the bound for the case when we have access to samples from $p$, that is, the no-noise case. 
They considered $\wh{\bftheta}_{N} = \underset{\bftheta \in \bfTheta}{\arg \min} {\MMD^2}(\wh{p}_N, {q}_{\bftheta})$, and derived 
\[
\MMD\left({p},{q}_{\wh{\bftheta}_{N}} \right) \leq \MMD\left({p}, {q}_{\bftheta^{\star}}\right)+  4\left(\frac{2K}{N}\right)^{1/2} \left(2+\{\log (1 / \gamma)\}^{1/2}\right).
\]
Since the bound doesn't depend directly on the distributions but on the sample size and kernel properties, we can directly use it for our case. 
Lemma~\ref{lemma: identifiability-learning} guarantees that minimizing the population objective indeed gets us $q_{\bftheta^{\star}} = p$ in the presence of noise. 
Similarly, the consistency result also follows from Proposition 1 of \citet{briol2019statistical}. 
\end{proof}

\begin{customthm}{\ref{th:clt}}[Central limit theorem] Let $\{\wt{\bfw}_{i}\}_{i=1}^{N}$ be i.i.d. samples from the noisy distribution $p*m$. Let the candidate models be $\{q_{\bftheta} : \bftheta \in \bfTheta\}$, where $\bfTheta \subseteq \R^{d_{\bftheta}}$ is an open set. Let $\wh{\bftheta}_{N}$ be the empirical \convMMD estimator obtained by minimizing the empirical objective function $L_{N}(\bftheta) = \MMD^2(\wh{(p*m)}_N, q_{\bftheta}*m)$. 

Assume the following regularity conditions hold:
\begin{enumerate}
    \item[(i)] The true distribution $p$ is in the model class: $p = q_{\bftheta^{\star}}$ for a unique $\bftheta^{\star} \in \bfTheta$.
    \item[(ii)] The estimator is consistent: $\wh{\bftheta}_{N} \xrightarrow{p} \bftheta^{\star}$ as $N \to \infty$.
    \item[(iii)] The log-density $\log q_{\bftheta}(\bfz)$ is twice continuously differentiable with respect to $\bftheta$ in a neighbourhood of $\bftheta^{\star}$.
    \item[(iv)] The kernel $k(\cdot, \cdot)$ and expectations involving it are sufficiently smooth to justify interchanging differentiation and integration.
    \item[(v)] The matrices $\wt{\bfH}(\bftheta^{\star})$ and $\wt{\bfSigma}_{\rm score}$ defined below are finite and $\wt{\bfH}(\bftheta^{\star})$ is positive definite.
\end{enumerate}
Then, as $N \to \infty$, the estimator is asymptotically normal:
\[
N^{1/2} (\wh{\bftheta}_{N} - \bftheta^{\star}) \xrightarrow{d} \N(\bfzero, \wt{\bfC}_{\rm score}).
\]
The asymptotic covariance matrix is the Godambe information matrix $\wt{\bfC}_{\rm score} = \wt{\bfH}(\bftheta^{\star})^{-1} \wt{\bfSigma}_{\rm score} \wt{\bfH}(\bftheta^{\star})^{-1}$, where \citep{ferguson2017course}:

\begin{enumerate}
    \item \textit{The curvature matrix}: $\wt{\bfH}(\bftheta^{\star}) = \nabla_{\bftheta}^{2} L(\bftheta) \big|_{\bftheta=\bftheta^{\star}}$, where $L(\bftheta)$ is the population objective function $L(\bftheta) = \convMMD^{2}(p, q_{\bftheta}, m)$.
    
    \item \textit{The gradient variance matrix}: $\wt{\bfSigma}_{\rm score} = \E_{\wt{\bfW}\sim p*m} \left[ s(\wt{\bfW}, \bftheta^{\star}) s(\wt{\bfW}, \bftheta^{\star})\trans \right]$, where the score vector $s(\wt{\bfW}, \bftheta)$ is the influence of a single data point $\wt{\bfW}$ on the population gradient, defined as:
    \[
    s(\wt{\bfW}, \bftheta) = \nabla_{\bftheta} \left( \E_{\wt{\bfZ},\wt{\bfZ}' \sim q_{\bftheta}*m}[k(\wt{\bfZ}, \wt{\bfZ}')] - 2\E_{\wt{\bfZ} \sim q_{\bftheta}*m}[k(\wt{\bfW}, \wt{\bfZ})] \right).
    \]
\end{enumerate}
\end{customthm}

\begin{proof}
The estimator $\wh{\bftheta}_{N}$ minimizes the empirical objective:
\[
L_{N}(\bftheta) = \frac{1}{N^2}\sum_{i,j=1}^N k(\wt{\bfw}_{i}, \wt{\bfw}_{j}) - \frac{2}{N}\sum_{i=1}^{N} \E_{\wt{\bfZ} \sim q_{\bftheta}*m}[k(\wt{\bfW}_{i}, \wt{\bfZ})] + \E_{\wt{\bfZ},\wt{\bfZ}' \sim q_{\bftheta}*m}[k(\wt{\bfZ}, \wt{\bfZ}')].
\]
As a minimizer, $\wh{\bftheta}_{N}$ satisfies the first-order condition $\nabla_{\bftheta} L_{N}(\wh{\bftheta}_{N}) = \bfzero$.

We perform a first-order Taylor expansion of $\nabla_{\bftheta} L_{N}(\bftheta)$ around the true parameter $\bftheta^{\star}$,
\[
\bfzero = \nabla_{\bftheta} L_{N}(\wh{\bftheta}_{N}) = \nabla_{\bftheta} L_{N}(\bftheta^{\star}) + \nabla_{\bftheta}^{2} L_{N}(\bftheta^{\star\star})(\wh{\bftheta}_{N} - \bftheta^{\star}),
\]
where $\bftheta^{\star\star}$ is a point on the line segment between $\wh{\bftheta}_{N}$ and $\bftheta^{\star}$. Rearranging gives:
\[
N^{1/2}(\wh{\bftheta}_{N} - \bftheta^{\star}) = - \left[ \nabla_{\bftheta}^{2} L_{N}(\bftheta^{\star\star}) \right]^{-1} N^{1/2} \nabla_{\bftheta} L_{N}(\bftheta^{\star}).
\]

The goal is to prove that the empirical Hessian, evaluated at an intermediate point, converges in probability to the deterministic population Hessian at the true parameter: $\nabla_{\bftheta}^{2} L_{N}(\bftheta^{\star\star}) \xrightarrow{p} \wt{\bfH}(\bftheta^{\star})$.

First, let us define the population loss and its Hessian. The population loss is:
\[
L(\bftheta) = \E_{\wt{\bfW}, \wt{\bfW}' \sim p*m}[k(\wt{\bfW}, \wt{\bfW}')] - 2\E_{\wt{\bfW} \sim p*m, \wt{\bfZ} \sim q_{\bftheta}*m}[k(\wt{\bfW}, \wt{\bfZ})] + \E_{\wt{\bfZ}, \wt{\bfZ}' \sim q_{\bftheta}*m}[k(\wt{\bfZ}, \wt{\bfZ}')].
\]
The population Hessian is its second derivative with respect to $\bftheta$. 
Since the first term does not depend on $\bftheta$, so its contribution to the Hessian is zero, and we have:
\[
\wt{\bfH}(\bftheta) = \nabla_{\bftheta}^{2} L(\bftheta) = \nabla_{\bftheta}^{2} \left( \E_{\wt{\bfZ}, \wt{\bfZ}' \sim q_{\bftheta}*m}[k(\wt{\bfZ}, \wt{\bfZ}')] - 2\E_{\wt{\bfW} \sim p*m, \wt{\bfZ} \sim q_{\bftheta}*m}[k(\wt{\bfW}, \wt{\bfZ})] \right).
\]
Now, consider the empirical Hessian, $\nabla_{\bftheta}^{2} L_{N}(\bftheta)$. 
The first term of $L_{N}(\bftheta)$ does not depend on $\bftheta$, so its contribution to the Hessian is zero, and we again have: 
\[
\nabla_{\bftheta}^{2} L_{N}(\bftheta) = \nabla_{\bftheta}^{2} \left( \E_{\wt{\bfZ}, \wt{\bfZ}' \sim q_{\bftheta}*m}[k(\wt{\bfZ}, \wt{\bfZ}')] - \frac{2}{N}\sum_{i=1}^{N} \E_{\wt{\bfZ} \sim q_{\bftheta}*m}[k(\wt{\bfw}_{i}, \wt{\bfZ})] \right).
\]
Let $\wt{\bfG}(\bftheta, \wt{\bfw}) = \nabla_{\bftheta}^{2} \left( -\E_{\wt{\bfZ} \sim q_{\bftheta}*m}[k(\wt{\bfw}, \wt{\bfZ})] \right)$. Then the random part of the empirical Hessian is an average of i.i.d. random matrices: $2N^{-1}\sum_{i=1}^{N} \wt{\bfG}(\bftheta, \wt{\bfw}_{i})$.

By the law of large numbers for random matrices, for any fixed $\bftheta$ in a neighbourhood of $\bftheta^{\star}$, we have: 
\[
\frac{2}{N}\sum_{i=1}^{N} \wt{\bfG}(\bftheta, \wt{\bfw}_{i}) \xrightarrow{p} 2\E_{\wt{\bfW} \sim p*m}[\wt{\bfG}(\bftheta, \wt{\bfW})].
\]
Therefore, the full empirical Hessian converges to
\[
\nabla_{\bftheta}^{2} L_{N}(\bftheta) \xrightarrow{p} \nabla_{\bftheta}^{2} \E_{\wt{\bfZ}, \wt{\bfZ}'}[k(\wt{\bfZ}, \wt{\bfZ}')] - 2\E_{\wt{\bfW} \sim p*m}[\nabla_{\bftheta}^{2}(\E_{\wt{\bfZ}}[k(\wt{\bfW}, \wt{\bfZ})])] = \wt{\bfH}(\bftheta).
\]
This establishes pointwise convergence. 
Now, since we assumed consistency ($\wh{\bftheta}_{N} \xrightarrow{p} \bftheta^{\star}$), it follows that the intermediate point also converges ($\bftheta^{\star\star} \xrightarrow{p} \bftheta^{\star}$). 
Under regularity conditions ensuring the map $\bftheta \mapsto \wt{\bfH}(\bftheta)$ is continuous and that the empirical Hessian converges uniformly in a neighbourhood of $\bftheta^{\star}$ (a standard condition in M-estimator theory), we can conclude
\[
\nabla_{\bftheta}^{2} L_{N}(\bftheta^{\star\star}) \xrightarrow{p} \wt{\bfH}(\bftheta^{\star}).
\]

The gradient of the empirical objective can be written as an average
\[
\nabla_{\bftheta} L_{N}(\bftheta) = \frac{1}{N} \sum_{i=1}^{N} s(\wt{\bfw}_{i}, \bftheta),
\]
where $s(\wt{\bfw}, \bftheta)$ is the score vector defined in the theorem statement. As shown previously, at the true parameter value $\bftheta^{\star}$, this score has an expectation of zero: $\E_{\wt{\bfW} \sim p*m}[s(\wt{\bfW}, \bftheta^{\star})] = 0$.

We have an average of $N$ i.i.d. zero-mean random vectors. 
By the multivariate central limit theorem, we have: 
\[
N^{1/2} \nabla_{\bftheta} L_{N}(\bftheta^{\star}) = \frac{1}{N^{1/2}} \sum_{i=1}^{N} s(\wt{\bfw}_{i}, \bftheta^{\star}) \xrightarrow{d} \N(\bfzero, \mathrm{cov}(s(\wt{\bfW}, \bftheta^{\star}))).
\]
The covariance matrix is $\mathrm{cov}(s(\wt{\bfW}, \bftheta^{\star})) = \E[s(\wt{\bfW}, \bftheta^{\star}) s(\wt{\bfW}, \bftheta^{\star})\trans] = \wt{\bfSigma}_{\rm score}$.

We combine the results into the rearranged Taylor expansion:
\[
N^{1/2}(\wh{\bftheta}_{N} - \bftheta^{\star}) = \underbrace{- \left[ \nabla_{\bftheta}^{2} L_{N}(\bftheta^{\star\star}) \right]^{-1}}_{\xrightarrow{p} \quad -\wt{\bfH}(\bftheta^{\star})^{-1}} \underbrace{N^{1/2} \nabla_{\bftheta} L_{N}(\bftheta^{\star})}_{\xrightarrow{d} \quad \N(0, \wt{\bfSigma}_{\rm score})}.
\]
By Slutsky's theorem, the product converges in distribution to the product of the constant matrix and the normal random variable. The resulting distribution is a multivariate normal with mean zero and covariance matrix
\[
\wt{\bfC}_{\rm score} = (-\wt{\bfH}(\bftheta^{\star})^{-1}) \wt{\bfSigma}_{\rm score} (-\wt{\bfH}(\bftheta^{\star})^{-1})\trans  = \wt{\bfH}(\bftheta^{\star})^{-1} \wt{\bfSigma}_{\rm score} \wt{\bfH}(\bftheta^{\star})^{-1},
\]
where $\wt{\bfH}(\bftheta^{\star})$ is the symmetric Hessian matrix. This completes the proof.
\end{proof}

\begin{customex}{\ref{ex: clt-illust}}
Consider the setting where the true data-generating distribution is $p = \mathcal{N}(\bftheta^{\star}, \sigma^{2} \bfI_{d})$ and the candidate family is $q_{\bftheta} = \mathcal{N}(\bftheta, \sigma^{2} \bfI_{d})$, for $\bftheta, \bftheta^{\star} \in \R^d$. The data and model samples are convolved with known isotropic Gaussian noise $m = \mathcal{N}(\bfzero, \tau^{2} \bfI_{d})$. Using a Gaussian kernel $k(\bfw,\bfz) \propto \exp\left(-\|\bfw-\bfz\|^{2}/(2l^2)\right)$, the asymptotic covariance of the \MMD estimator $\wh{\bftheta}_{N}$ is given by
\[
\wt{\bfC}_{\rm score} = (\sigma^{2} + \tau^{2})\left\{(l^2+\sigma^{2} + \tau^{2})(l^2+3\sigma^{2} + 3\tau^{2})\right\}^{-\frac{d}{2}-1}\left(l^2+2\sigma^{2} + 2\tau^{2}\right)^{d+2} \bfI_{d}.
\]
\end{customex}

\begin{proof}
The asymptotic covariance is given by the sandwich matrix $\wt{\bfC}_{\rm score} = \wt{\bfH}(\bftheta^{\star})^{-1} \wt{\bfSigma}_{\rm score} \wt{\bfH}(\bftheta^{\star})^{-1}$, where $\wt{\bfH}(\bftheta^{\star})$ is the Hessian of the population \MMD loss and $\wt{\bfSigma}_{\rm score}$ is the covariance of the gradient estimator's influence function, both evaluated at the true parameter $\bftheta^{\star}$.

The noisy distributions are convolutions of Gaussians:
\begin{enumerate}
    \item Noisy data distribution: $p*m = \mathcal{N}(\bftheta^{\star}, \sigma^{2} \bfI_{d}) * \mathcal{N}(\bfzero, \tau^{2} \bfI_{d}) = \mathcal{N}(\bftheta^{\star}, (\sigma^{2} + \tau^{2})\bfI_{d})$.
    \item Noisy model distribution: $q_{\bftheta}*m = \mathcal{N}(\bftheta, \sigma^{2} \bfI_{d}) * \mathcal{N}(\bfzero, \tau^{2} \bfI_{d}) = \mathcal{N}(\bftheta, (\sigma^{2} + \tau^{2})\bfI_{d})$.
\end{enumerate}
Let $\sigma_{\text{total}}^{2} = \sigma^{2} + \tau^{2}$. 
As before, we denote samples from these distributions as $\wt{\bfW} \sim \mathcal{N}(\bftheta^{\star}, \sigma_{\text{total}}^{2} \bfI_{d})$ and $\wt{\bfZ} \sim \mathcal{N}(\bftheta, \sigma_{\text{total}}^{2} \bfI_{d})$, respectively. 
A key identity for the expectation of a Gaussian kernel between $\bfX \sim \mathcal{N}(\bfmu_{1}, \sigma_{1}^{2} \bfI_{d})$ and $\bfY \sim \mathcal{N}(\bfmu_{2}, \sigma_{2}^{2} \bfI_{d})$ is
\[
\E[k(\bfX, \bfY)] = \left(1 + \frac{\sigma_{1}^{2} + \sigma_{2}^{2}}{l^2}\right)^{-d/2} \exp\left(-\frac{\|\bfmu_{1} - \bfmu_{2}\|^2}{2(l^2 + \sigma_{1}^{2} + \sigma_{2}^{2})}\right).
\]

The population loss is $L(\bftheta) = \E_{\wt{\bfW}, \wt{\bfW}'}[k(\wt{\bfW}, \wt{\bfW}')] - 2\E_{\wt{\bfW}, \wt{\bfZ}}[k(\wt{\bfW}, \wt{\bfZ})] + \E_{\wt{\bfZ},\wt{\bfZ}'}[k(\wt{\bfZ},\wt{\bfZ}')]$. The curvature is its Hessian:
\[
\wt{\bfH}(\bftheta) = \nabla_{\bftheta}^{2} L(\bftheta) = \nabla_{\bftheta}^{2} \left( \E_{\wt{\bfZ}, \wt{\bfZ}' \sim q_{\bftheta}*m}[k(\wt{\bfZ}, \wt{\bfZ}')] - 2\E_{\wt{\bfW} \sim p*m, \wt{\bfZ} \sim q_{\bftheta}*m}[k(\wt{\bfW}, \wt{\bfZ})] \right).
\]

For Term 1: $\E[k(\wt{\bfZ}, \wt{\bfZ}')]$, we have $\bfmu_{1}=\bfmu_{2}=\bftheta$ and $\sigma_{1}^{2}=\sigma_{2}^{2}=\sigma_{\text{total}}^{2}$. The expectation is $\left(1 + 2\sigma_{\text{total}}^{2}/l^2\right)^{-d/2}$. This term is constant in $\bftheta$, so its contribution to the Hessian is zero.

For Term 2: $-2\E[k(\wt{\bfW}, \wt{\bfZ})]$
Let $L_{2}(\bftheta) = -2\E[k(\wt{\bfW}, \wt{\bfZ})]$, we have $\bfmu_{1}=\bftheta^{\star}, \bfmu_{2}=\bftheta$ and $\sigma_{1}^{2}=\sigma_{2}^{2}=\sigma_{\text{total}}^{2}$.
\[
L_{2}(\bftheta) = -2 \left(1 + \frac{2\sigma_{\text{total}}^{2}}{l^2}\right)^{-d/2} \exp\left(-\frac{\|\bftheta^{\star} - \bftheta\|^2}{2(l^2 + 2\sigma_{\text{total}}^{2})}\right).
\]
The gradient is
\[
\nabla_{\bftheta} L_{2}(\bftheta) = -2 \left(1 + \frac{2\sigma_{\text{total}}^{2}}{l^2}\right)^{-d/2} \exp(\dots) \left( \frac{2(\bftheta - \bftheta^{\star})}{2(l^2 + 2\sigma_{\text{total}}^{2})} \right).
\]
The Hessian at $\bftheta = \bftheta^{\star}$ is found by differentiating again. 
The term multiplying $(\bftheta - \bftheta^{\star})$ goes to zero, leaving
\[
\nabla_{\bftheta}^{2} L_{2}(\bftheta)\Big|_{\bftheta=\bftheta^{\star}} = -2 \left(1 + \frac{2\sigma_{\text{total}}^{2}}{l^2}\right)^{-d/2} \left( \frac{-1}{l^2 + 2\sigma_{\text{total}}^{2}} \right) \bfI_{d}.
\]
Thus, $\wt{\bfH}(\bftheta^{\star}) = g \cdot \bfI_{d}$, where:
\[
g = \frac{2}{l^2 + 2\sigma_{\text{total}}^{2}} \left(1 + \frac{2\sigma_{\text{total}}^{2}}{l^2}\right)^{-d/2} = \frac{2}{l^2 + 2\sigma_{\text{total}}^{2}} \left(\frac{l^2 + 2\sigma_{\text{total}}^{2}}{l^2}\right)^{-d/2}.
\]
The squared scalar component is:
\[
g^2 = \frac{4}{(l^2 + 2\sigma_{\text{total}}^{2})^2} \left(\frac{l^2 + 2\sigma_{\text{total}}^{2}}{l^2}\right)^{-d}.
\]

We need $\wt{\bfSigma}_{\rm score} = \E_{\wt{\bfW} \sim p*m} [ s(\wt{\bfW}, \bftheta^{\star}) s(\wt{\bfW}, \bftheta^{\star})\trans ]$, where the score is
\[
s(\wt{\bfw}, \bftheta) = -2 \nabla_{\bftheta} \E_{\wt{\bfZ} \sim q_{\bftheta}*m}[k(\wt{\bfw}, \wt{\bfZ})].
\]
The expectation is between a point $\wt{\bfw}$ and $\wt{\bfZ} \sim \mathcal{N}(\bftheta, \sigma_{\text{total}}^{2} \bfI_{d})$.
\[
\E_{\wt{\bfZ}}[k(\wt{\bfw}, \wt{\bfZ})] = \left(1 + \frac{\sigma_{\text{total}}^{2}}{l^2}\right)^{-d/2} \exp\left(-\frac{\|\wt{\bfw} - \bftheta\|^2}{2(l^2 + \sigma_{\text{total}}^{2})}\right).
\]
Taking the gradient w.r.t. $\bftheta$ and evaluating at $\bftheta=\bftheta^{\star}$, we have: 
\[
s(\wt{\bfw}, \bftheta^{\star}) = -2 \left(1 + \frac{\sigma_{\text{total}}^{2}}{l^2}\right)^{-d/2} \exp\left(-\frac{\|\wt{\bfw} - \bftheta^{\star}\|^2}{2(l^2 + \sigma_{\text{total}}^{2})}\right) \frac{(\wt{\bfw} - \bftheta^{\star})}{(l^2 + \sigma_{\text{total}}^{2})}.
\]
Let $\bfV = (\wt{\bfW} - \bftheta^{\star}) \sim \mathcal{N}(\bfzero, \sigma_{\text{total}}^{2} \bfI_{d})$. 
Then $\wt{\bfSigma}_{\rm score} = \E[s(\bftheta^{\star}+\bfV, \bftheta^{\star}) s(\bftheta^{\star}+\bfV, \bftheta^{\star})\trans]$.
\[
s(\bftheta^{\star}+\bfV, \bftheta^{\star}) = \underbrace{-\frac{2}{(l^2 + \sigma_{\text{total}}^{2})} \left(\frac{l^2 + \sigma_{\text{total}}^{2}}{l^2}\right)^{-d/2}}_{C_3} \exp\left(-\frac{\|\bfV\|^{2}}{2(l^2 + \sigma_{\text{total}}^{2})}\right) \bfV.
\]
So, $\wt{\bfSigma}_{\rm score} = C_{3}^{2} \E_{\bfV \sim \mathcal{N}(\bfzero, \sigma_{\text{total}}^{2} \bfI_{d})} \left[ \exp\left(-\|\bfV\|^{2}/(l^2 + \sigma_{\text{total}}^{2})\right) \bfV\bfV\trans  \right]$. 
By symmetry, the result is $\sigma_{s}^{2} \bfI_{d}$. 
We compute the scalar $\sigma_{s}^{2}$ by evaluating the $(1,1)$ entry:
\[
\sigma_{s}^{2} = C_{3}^{2} \cdot E\left[ \exp\left(-\frac{\sum_{i} V_{i}^2}{l^2 + \sigma_{\text{total}}^{2}}\right) V_{1}^{2} \right].
\]
The expectation is the integral:
\[
\int_{\R^d} \exp\left(-\frac{\sum_{i} v_{i}^2}{l^2 + \sigma_{\text{total}}^{2}}\right) v_{1}^{2} \frac{1}{(2\pi\sigma_{\text{total}}^{2})^{d/2}} \exp\left(-\frac{\sum_{i} v_{i}^2}{2\sigma_{\text{total}}^{2}}\right) d\bfv.
\]
The combined exponent is $-(1/2)\sum_{i} v_{i}^2 \left(2/(l^2 + \sigma_{\text{total}}^{2}) + \sigma_{\text{total}}^{-2}\right) = -\sum_{i} v_{i}^2/(2\sigma_B^2)$, where $\sigma_B^{-2} = (l^2 + 3\sigma_{\text{total}}^{2})/\{\sigma_{\text{total}}^{2}(l^2 + \sigma_{\text{total}}^{2})\}$.
The integral evaluates to $\sigma_B^2 \left(\sigma_B^2/\sigma_{\text{total}}^{2}\right)^{d/2}$. Substituting $\sigma_B^2$, we get:
\[
\text{Integral} = \frac{\sigma_{\text{total}}^{2}(l^2+\sigma_{\text{total}}^{2})}{l^2+3\sigma_{\text{total}}^{2}} \left( \frac{l^2+\sigma_{\text{total}}^{2}}{l^2+3\sigma_{\text{total}}^{2}} \right)^{d/2}.
\]
Therefore, 
\[
\sigma_{s}^{2} = C_{3}^{2} \cdot (\text{Integral}) = \frac{4}{(l^2 + \sigma_{\text{total}}^{2})^2} \left(\frac{l^2 + \sigma_{\text{total}}^{2}}{l^2}\right)^{-d} \left[ \frac{\sigma_{\text{total}}^{2}(l^2+\sigma_{\text{total}}^{2})}{l^2+3\sigma_{\text{total}}^{2}} \left( \frac{l^2+\sigma_{\text{total}}^{2}}{l^2+3\sigma_{\text{total}}^{2}} \right)^{d/2} \right].
\]

The final covariance is $\wt{\bfC}_{\rm score} = \sigma_{s}^{2}g^{-2} \bfI_{d}$.
\begin{align*}
\frac{\sigma_{s}^{2}}{g^2} &= \frac{\frac{4}{(l^2 + \sigma_{\text{total}}^{2})^2} \left(\frac{l^2 + \sigma_{\text{total}}^{2}}{l^2}\right)^{-d} \left[ \frac{\sigma_{\text{total}}^{2}(l^2+\sigma_{\text{total}}^{2})}{l^2+3\sigma_{\text{total}}^{2}} \left( \frac{l^2+\sigma_{\text{total}}^{2}}{l^2+3\sigma_{\text{total}}^{2}} \right)^{d/2} \right]}{\frac{4}{(l^2 + 2\sigma_{\text{total}}^{2})^2} \left(\frac{l^2 + 2\sigma_{\text{total}}^{2}}{l^2}\right)^{-d}} \\
&= \frac{\sigma_{\text{total}}^{2} (l^2 + 2\sigma_{\text{total}}^{2})^2}{(l^2 + \sigma_{\text{total}}^{2})(l^2 + 3\sigma_{\text{total}}^{2})} \cdot \frac{\left(\frac{l^2 + \sigma_{\text{total}}^{2}}{l^2}\right)^{-d}}{\left(\frac{l^2 + 2\sigma_{\text{total}}^{2}}{l^2}\right)^{-d}} \cdot \left( \frac{l^2 + \sigma_{\text{total}}^{2}}{l^2 + 3\sigma_{\text{total}}^{2}} \right)^{d/2} \\
&= \frac{\sigma_{\text{total}}^{2} (l^2 + 2\sigma_{\text{total}}^{2})^2}{(l^2 + \sigma_{\text{total}}^{2})(l^2 + 3\sigma_{\text{total}}^{2})} \cdot \left(\frac{l^2 + 2\sigma_{\text{total}}^{2}}{l^2 + \sigma_{\text{total}}^{2}}\right)^{d} \cdot \left( \frac{l^2 + \sigma_{\text{total}}^{2}}{l^2 + 3\sigma_{\text{total}}^{2}} \right)^{d/2} \\
&= \frac{\sigma_{\text{total}}^{2} (l^2 + 2\sigma_{\text{total}}^{2})^{d+2}}{(l^2 + \sigma_{\text{total}}^{2})^{1+d-d/2} (l^2 + 3\sigma_{\text{total}}^{2})^{1+d/2}} \\
&= \sigma_{\text{total}}^{2} (l^2 + 2\sigma_{\text{total}}^{2})^{d+2} \left\{ (l^2 + \sigma_{\text{total}}^{2})(l^2 + 3\sigma_{\text{total}}^{2}) \right\}^{-(1+d/2)}.
\end{align*}
Substituting $\sigma_{\text{total}}^{2} = \sigma^{2} + \tau^{2}$ gives the desired result.
\end{proof}

\section{Additional illustrations of main theorems} \label{sec: sm illust}

\subsection{Theorem \ref{th: mmd-equivalence} and \ref{th:hypothesis_testing}}

Here, we provide some additional numerical illustrations of the main theoretical results discussed in the main paper. 

In Example ~\ref{ex: mmd-equivalence}, we show that if the kernel is $k(w,z)\!=\!\exp\left\{-(w-z)^2/(2l^2)\right\}$ and the marginal noise distribution is $m(\cdot)\!=\!\mathcal{N}(0, \tau^{2})$, then the modified kernel $\wt{k}$ is a Gaussian kernel with an effective bandwidth of $(l^2 + 2\tau^{2})$. To illustrate this, we simulate an experiment where we assume that we have $N = 5000$ samples from $p \sim N(0,2^2)$ and $q \sim N(\theta,2^2)$, where $\theta$ takes 50 equally spaced values in $[-2,2]$. We set $l = 1$ and $\tau = 1$. We then compute $\wh{\convMMD_{u}^{2}}(p,q,m)$ and $\wh{\MMD_u^2}(p,q)$ using kernels $k$ and $\wt{k}$ respectively and repeat the experiment 100 times. We plot the average values across 100 runs vs. $\theta$ in  Fig.~\ref{fig:mmd-equivalence}(a). It can be seen from the figure that the average values of $\wh{\convMMD_{u}^{2}}(p,q,m)$ and $\wh{\MMD_u^2}(p,q)$ align perfectly with each other across values of $\theta$, which shows the validity of Theorem \ref{th: mmd-equivalence}. 

We also illustrate Theorem \ref{th:hypothesis_testing} in Fig.~\ref{fig:mmd-equivalence}(b). Here, we assume that we have samples from $p \sim N(0,2^2)$ and $q \sim N(0,2^2)$ for N = 100, 1000, 5000, 20000. We set $l = 1$ and $\tau = 1$. Since we have the case where p = q, the mean absolute error is given by the mean absolute value of $\wh{\convMMD_{u}^{2}}(p,q,m)$ and $\wh{\MMD_u^2}(p,q)$. It can be seen from the graph that both ${\MMD}$ and $\convMMD$ have the same rate of convergence. 

\subsection{Theorem \ref{th: mmd-var}}

We illustrate Theorem \ref{th: mmd-var} in Fig.~\ref{fig: mmd-var}. We simulate samples from $p \sim N(0, 2^2)$ and $q \sim N(0, 2^2)$ for N = 50, 100, 500, 1000, 2000, 20000. We vary the value of $\tau^{2}$, where $\tau^{2}$ represents the variance of the noise distribution, $m \sim N(0,\tau^{2})$. Finally, we estimate the variance of $\wh{\convMMD^2}$ by computing $\wh{\convMMD_{u}^{2}}$ for each N across a batch of 500 simulated samples. 
We see that the theoretical bound aligns with empirical observations. The variance of the \convMMD statistic increases with the noise level, but in a controlled manner that can be bounded. This result provides a theoretical guarantee that even in the presence of substantial noise, the decrease in performance of the estimator can be bounded by the properties of the noise process itself. 
\begin{figure}
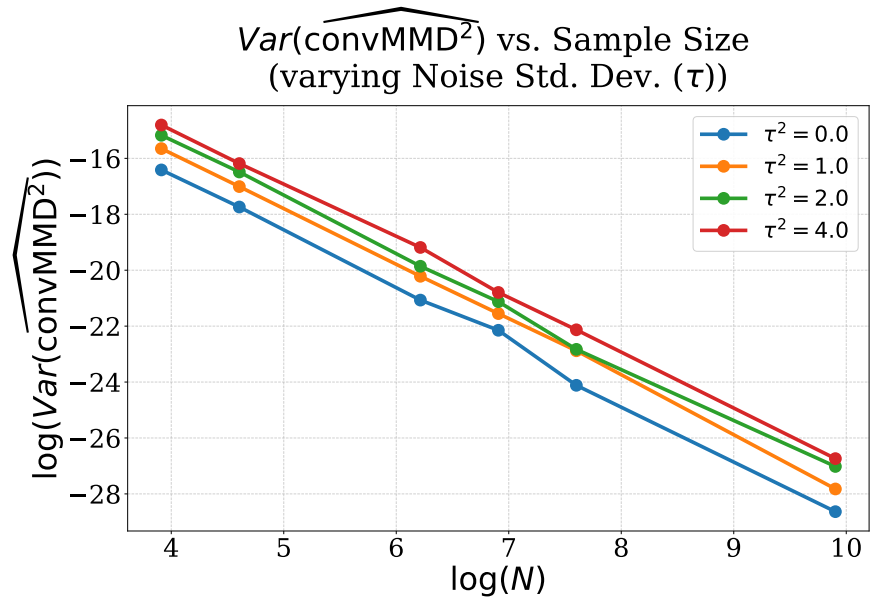

\figuresize{0.3078}
\figurebox{20pc}{25pc}{}[convmmd_var]
\caption{Illustration of Theorem \ref{th: mmd-var} for the special case of a Gaussian kernel with mean-zero Gaussian noise.}
\label{fig: mmd-var}
\end{figure}

\subsection{Theorem \ref{th:clt}}
We generate a sample of sizes $N = 100, 500, 1000, 2000, 5000$ and estimate the parameter $\wh{\sigma}$ using the procedure detailed in Algorithm \ref{alg:optimization} for each of $R = 1000$ Monte Carlo repetitions. We fix the bandwidth to $1.0$ and consider a learning rate of $0.01$ for optimization. The left part of the Fig.~\ref{fig:clt-experiment} shows the distribution of parameter estimates obtained using our algorithm for $N = 5000$. We overlay the plotted samples with the density of the normal distribution with parameters obtained using the central limit theorem discussed in Theorem \ref{th:clt}. We estimate the standard deviation $\wt{\bfC}_{\rm score}$ using Monte Carlo simulations (see Code in Supplementary Material for more details). The right part of the Fig.~\ref{fig:clt-experiment} in the main paper shows the variance of the estimates across different values of $N$ on a log scale. 

\section{Details of main implementation algorithm} \label{sec: Alg}

Our main \textsc{sgd} implementation algorithm proceeds as follows. 

\begin{algorithm}
\LinesNumbered
\DontPrintSemicolon
\caption{\MMD for parametric estimation from noisy data}\label{alg:optimization}
\KwIn{$\{\wt{w}_{i}\}_{i=1}^{N}$ (observed noisy samples from $p*m$), $q_{\bftheta}$ (candidate model), $\bftheta^{(0)}$ (initial parameters), $k(\cdot,\cdot)$ (kernel function), $N_{\text{iter}}$ (number of iterations), $\{\eta_t\}_{t=1}^{N_{\text{iter}}}$ (learning rates), $M$ (Monte Carlo batch size)}
\KwOut{$\bftheta^{(N_{\text{iter}})}$ (optimized parameters)}
Set hyperparameters for $k(\cdot,\cdot)$.\;
\For{$t = 1 \to N_{\text{iter}}$}{
    Sample a batch of $M$ \textit{clean} samples: $\mathbf{z} = \{z_j\}_{j=1}^{M} \sim q_{\bftheta^{(t-1)}}$.\;
    For each clean sample $z_j$, compute its score: $\mathbf{S}_j \gets \nabla_{\bftheta} \log q_{\bftheta}(z_j)|_{\theta=\theta^{(t-1)}}$.\;
    Generate a batch of $M$ noise samples: $\mathbf{u} = \{u_j\}_{j=1}^{M} \sim r(\cdot|{\bfphi}), {\bfphi} \sim g(\cdot|\bfpsi)$.\;
    Create a batch of $M$ \textit{noisy} model samples: $\wt{\mathbf{z}} \gets \mathbf{z} + \mathbf{u}$.\;
    For each noisy sample $\wt{z}_{j}$ in the batch, compute its per-sample \MMD contribution:\;
    \qquad $f_{j} \gets (M-1)^{-1}\sum_{l \neq j} k(\wt{z}_{j}, \wt{z}_l) - N^{-1}\sum_{i=1}^{N} k(\wt{z}_{j}, \wt{w}_{i})$.\;
    Compute the unbiased gradient estimate (the average of products):\;
    \qquad $\wh{\bfJ}_{\bftheta^{(t-1)}} \gets 2M^{-1} \sum_{j=1}^{M} f_{j} \cdot \mathbf{S}_j$.\;
    $\bftheta^{(t)} \gets \bftheta^{(t-1)} - \eta_t \wh{\bfJ}_{\bftheta^{(t-1)}}$.\;
}
\Return $\bftheta^{(N_{\text{iter}})}$\;
\end{algorithm}
The hyperparameters in Algorithm~\ref{alg:optimization} are the kernel function $k(.,.)$, the bandwidth associated with the kernel, and the batch size $M$. 
For all the simulation experiments, we use a Gaussian kernel with a multi-scale median heuristic for the bandwidth $\sigma$ and a batch size of M = N (sample size). 
A single bandwidth may fail to capture both coarse features (means) and fine features (variances/tails). A more robust practice is to use a mixture of kernels using different bandwidths \citep{schrab2025practical, biggs2023mmd}. 
Instead of one $\sigma$, we use a set of bandwidths:
\[
\bfsigma = \{0.5*\sigma_{\rm med},\sigma_{\rm med}, 1.5*\sigma_{\rm med} \},
\]
where $\sigma_{\rm med}$ represents the bandwidth selected using the median rule used in \citet{gretton12a}. According to the median rule, the bandwidth of the kernel is:
\[
\begin{aligned}
    \sigma_{\rm med}= {\rm median}(\{d_{ij}| i \neq j\}),
 \end{aligned}
\]
where $d_{ij} =\|V_i-V_j\|_2$ , $V = W \cup Z$ given two sets  $W=\{w_1, w_2, \dots w_N\}$ and $Z=\{z_1,z_2, \dots, z_N\}$.

\section{Additional details on simulations} \label{sec: sm sim}

\subsection{Gaussian mixture model estimation}

Here, we provide some additional details for the simulation experiments reported in the main paper, covering the data generation process, model evaluation, and results summarization.

The true  data $\left\{X_{i}\right\}_{i=1}^{N}$ with $N=1000$ was generated from a three-component \textsc{gmm} with the
following parameters:
\begin{enumerate}
    \item Weights: $\pi^{\star}=\{0.23,0.33,0.44\}$.
    \item Means: $\mu^{\star}=\{-3.72,0.11,4.52\}$.
    \item  Standard deviations: $\sigma^{\star}=\{1.0,1.0,1.0\}$.
\end{enumerate}
We initialize the parameters by using estimates from a naive \textsc{gmm} model that doesn't account for the error in the data. The noise distributions are generated to have roughly similar effective variances across all experiments. The specific generation process is as follows:
\begin{enumerate}
    \item Gaussian (Homoscedastic): $U_{X_{i}} \sim r(\cdot | \phi) = N(0, \phi^2)$, where $\phi=1.258$.
    \item Uniform (Homoscedastic): $U_{X_{i}} \sim r(\cdot | \phi) = U(-\phi,\phi)$, where $\phi=2$.
    \item Gaussian (Heteroscedastic): $U_{X_{i}} \sim r(\cdot | \phi) = N(0, \phi^2)$, where $\phi \sim g(\cdot|\psi) = U(1,1.5)$.
    \item Laplace (Heteroscedastic): $U_{X_{i}} \sim r(\cdot | \phi) = \mathrm{Laplace}(0, \phi)$, where $\phi \sim g(\cdot|\psi) = 2^{1/2}U(1,1.5)$.
    \item Student's t (Heteroscedastic): $U_{X_{i}} \sim  r(\cdot | \phi) = \phi \cdot t(3)$, where $\phi  \sim g(\cdot|\psi) =3^{1/2}U(1,1.5)$.
\end{enumerate}
For \textsc{xdgmm}, we provide the true error standard deviation in case of homoscedastic noise processes as a parameter to the method. In the case of heteroscedastic noise processes, we provide point-wise error standard deviations that are derived from the conditional distribution of $\sigma$. In contrast, \convMMD, doesn't assume knowledge of true error standard deviation or point-wise error standard deviation. It only takes the parameters of noise distribution as an input and simulates noise from the distribution at each iteration.

We deal with the issue of label switching by indexing the estimated parameters in ascending order of estimated weights. This rule is applied to all three methods considered in the experiment to make a fair comparison.

\subsection{Regression with errors}
In this experiment, we consider a regression problem where both the covariates and the response are subject to measurement error. In the context of our theoretical framework (Section \ref{sec: estimation}), our goal is to estimate the parameter $\theta$ of a model distribution $q_{\bftheta}$ that best approximates the joint distribution of observed data $(\wt{X}, \wt{Y})$. We define our candidate model family $q_{\bftheta}(X,Y)$ as follows: 
\begin{equation*}
q_{\bftheta}(X, Y)=\underbrace{q\left(Y \mid X, \bftheta_{\text {reg }}\right)}_{\text {regression model }} \cdot \underbrace{q\left(X \mid \bftheta_{X}\right)}_{\text {latent covariate model }},
\end{equation*}
where the full parameter vector is $\bftheta = (\bftheta_{\text{reg}}, \bftheta_X)$. We model the unknown distribution of X using a \textsc{gmm} with $J = 2$ components. Thus, $\bftheta_X = \{\pi_{k},\mu_{k},\sigma_{k}\}_{k=1}^2$. 

The true data $\left\{X_{i}, Y_{i}\right\}_{i=1}^{N}$ with $N=1000$ was generated as follows: 
\begin{enumerate}
\item True covariate $(X)$: The true covariate distribution is a two-component Gaussian mixture:
\[
X \sim 0.3 \, \mathcal{N}\left(2.5,1^2\right)+0.7 \, \mathcal{N}\left(3,1^2\right).
\]
\item True response $(Y):$ The true response is generated via a linear relationship with additive Gaussian noise:
\[
Y=1.5+ X+\varepsilon, \quad \varepsilon \sim \mathcal{N}\left(0,1^2\right).
\]
\end{enumerate}

Thus, given the above data-generating process, we have $\bftheta_{\text{reg}}=(\alpha, \beta, \sigma_{\text{reg}})$. 
The observed noisy samples $\wt{X}$ and $\wt{Y}$ are generated by adding measurement noise to clean data. 
We used several noise configurations to demonstrate the flexibility of the framework. 
All the noise distributions are generated to have roughly similar effective variances across all experiments. 
The specific generation process is as follows:
\begin{enumerate}
    \item Gaussian (Homoscedastic): $U_{X_{i}}, U_{Y_{i}} \sim r(\cdot | \phi) = N(0, \phi^2)$, where $\phi=1.258$. 
    \item Uniform (Homoscedastic): $U_{X_{i}}, U_{Y_{i}} \sim r(\cdot | \phi) = U(-\phi,\phi)$, where $\phi=2$. 
    \item Gaussian (Heteroscedastic): $U_{X_{i}}, U_{Y_{i}} \sim r(\cdot | \phi) = N(0, \phi^2)$, where $\phi \sim g(\cdot|\psi) = U(1,1.5)$. 
    \item Laplace (Heteroscedastic): $U_{X_{i}}, U_{Y_{i}} \sim r(\cdot | \phi) = \mathrm{Laplace}(0, \phi)$, where $\phi \sim g(\cdot|\bfpsi) = 2^{1/2}U(1,1.5)$.
    \item Student's t (Heteroscedastic): $U_{X_{i}}, U_{Y_{i}} \sim  r(\cdot | \phi) = \phi \cdot t(3)$, where $\phi  \sim g(\cdot|\psi) =3^{1/2}U(1,1.5)$.
\end{enumerate}

We compare our results with \textsc{simex} \citep{cook1994simulation} and \texttt{linmix} \citep{kelly2007some}. For \texttt{linmix}, we use the Gibbs sampling implementation provided by the lrgs package \citep{mantz2016lrgs}. Note that LRGS assumes conjugate priors (Inverse-Wishart) for the covariate mixture components, differing slightly from the truncated uniform priors in the original IDL implementation. 
For \convMMD, we used a warm-start strategy to improve optimization. The parameters of the latent covariate model $q(X|\bftheta_X)$ were initialized by fitting a standard \textsc{gmm} to noisy $\wt{X}$, and the regression parameters ($\bftheta_{\text{reg}}$) were initialized using the ordinary least squares (\textsc{ols}) model on the noisy data ($\wt{X},\wt{Y}$). 
We adopted a mixture of Gaussian kernels with the bandwidths selected using the median rule discussed in Section~\ref{sec: Alg} and a fixed batch size of $M = 1000$ samples. 

In addition to looking at the mean absolute error of the regression estimates, we also looked at the distribution of parameter estimates for the case of the Laplace distribution. The simulation experiment was repeated 500 times to obtain the distribution of the regression coefficients, which is shown in Fig.~\ref{fig: mmd-reg_asymp}(b) and Fig.~\ref{fig: mmd-reg_asymp2}. 

\begin{figure}
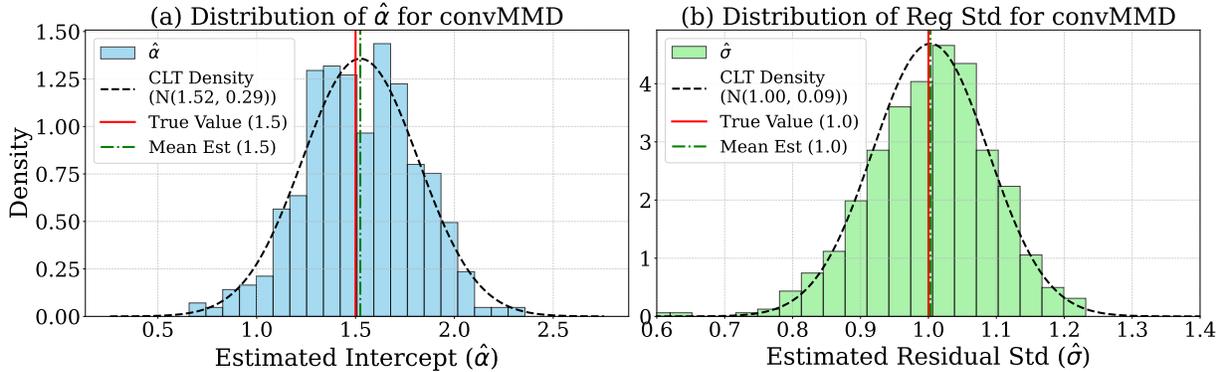

\figuresize{0.2236}
\figurebox{20pc}{25pc}{}[regression_sim_result_2]
\caption{Simulation experiments: Distribution of the \texttt{convMMD} estimates of $\hat{\alpha}$ (left) and Distribution of the \texttt{convMMD} estimates of $\hat{\sigma}$ (right).}
\label{fig: mmd-reg_asymp2}
\end{figure}

To quantify uncertainty, we estimated the asymptotic covariance matrix  \[\wt{\bfC}_{\rm score} = \wt{\bfH}(\bftheta^{\star})^{-1} \wt{\bfSigma}_{\rm score} \wt{\bfH}(\bftheta^{\star})^{-1},\] (see Theorem \ref{th:clt} for expression) for each run. 
Due to the high variance associated with score-function gradient estimators in the \convMMD objective, we employed a median-aggregated Monte Carlo approach: for each simulation run, we computed $B=50$ independent batch estimates of the Hessian and score covariance matrices and took the component-wise median to construct a stable variance estimate. 
In Fig.~\ref{fig: mmd-reg_asymp}, the empirical distribution of the coefficients is overlaid with a Gaussian density parametrized by the mean point estimate and the average estimated standard error across runs. 
The close alignment between the empirical histograms and the estimated theoretical CLT density (discussed in Theorem ~\ref{th:clt}) confirms that our variance estimator correctly captures the sampling variability of the parameters.

\section{Additional details on real data applications} \label{sec: sm real}

\subsection{Astronomy application: galaxy cluster data}

Here, we provide additional details of the real data applications discussed in the main paper. 

We use the dataset publicly available in the Supplementary Materials of \citet{farahi2019mass}. 
The input variables from the \textsc{des}-Y1 and XMM-Newton catalogues were transformed to log-space to linearize the power-law scaling relation. 
Let $\lambda_{\mathrm{RM}, i}$ be the optical richness and $T_{X, i}$ be the X-ray temperature for the $i$th cluster at redshift $z_i$. 
The observed covariates $\wt{X}$ and response $\wt{Y}$ were constructed as:
\begin{align*}
\wt{X}_{i} &= \log\left(\frac{\lambda_{\text{RM}, i}}{70}\right), \\
\wt{Y}_{i} &= \log(T_{X, i}) - \frac{2}{3}\log(E(z_i)),
\end{align*}
where $E(z)=\{\Omega_m(1+z)^3+\Omega_{\Lambda}\}^{1/2}$ is the dimensionless Hubble parameter describing cosmic expansion. 
In this application, a key advantage is that the data processing pipelines for both variables provide direct, quantitative estimates of the statistical uncertainty for each individual measurement, which serves as our noise parameter $\phi_{i}$. 
These estimates are treated as samples from the parameter distribution $g(\phi \mid \psi)$ and are provided as input to the algorithm. 
The conditional noise distribution $r(\cdot | \phi_{i})$ is assumed to be a mean-zero Gaussian $N(0,\phi_{i}^2)$.  

We used a warm-start strategy to improve optimization. The parameters of the latent covariate model $q(X|\bftheta_X)$ were initialized by fitting a standard \textsc{gmm} to noisy $\wt{X}$, and the regression parameters ($\bftheta_{\text{reg}}$) were initialized using the ordinary least squares (\textsc{ols}) model on the noisy data ($\wt{X},\wt{Y}$). We adopted a mixture of Gaussian kernels with the bandwidths selected using the median rule discussed in Section~\ref{sec: Alg} and a fixed batch size of $M = 1000$ samples.

\subsection{Anthropometry application: Davis data}
The Davis dataset \citep{davis1990body} contains data on the height and weight of $N = 183$ individuals who exercise regularly. 
The dataset contains both self-reported ($\wt{X}$) and measured observations ($X$) of the participants and is publicly available in the car package in R \citep{fox2018cardata}. 
The difference in measured weight and self-reported weight is treated as a sample from the marginal noise distribution ${m}(\cdot)$. 
For \convMMD, we do not make a parametric assumption on $m(\cdot)$; instead, during the convolution step ($q_{\bftheta} * m$), we implicitly sample noise values from this empirical residual distribution and add them to the latent data ($X_{\text{sim}}$) generated by the model. 
For \textsc{simex} and \texttt{linmix}, we use the empirical noise distribution to calculate the noise standard deviation that is passed as input to \textsc{simex} and \texttt{linmix}. 
The data also contains an outlier where measured height and measured weight for one observation are switched. 
In the case where the outlier is included in the data, the estimated standard deviation of the noise distribution increases significantly, which adversely affects both \textsc{simex} and \texttt{linmix}. In contrast, \convMMD remains unaffected by the presence of an outlier. 

We modelled the latent distribution of measured height using a \textsc{gmm} with $J = $ components for both \convMMD and \texttt{linmix}.
To ensure stable convergence, we use a warm-start initialization strategy.
The parameters of the latent covariate model $q(X|\bftheta_X)$ were initialized by fitting a standard \textsc{gmm} to noisy $\wt{X}$, and the regression parameters ($\bftheta_{\text{reg}}$) were initialized using the ordinary least squares (\textsc{ols}) model on the noisy data ($\wt{X}, Y$).
For optimization in \convMMD, we use the product of two Gaussian kernels $K_X$ and $K_Y$ with different bandwidths to deal with the difference in scale in $X$ and $Y$. 
The bandwidths for both kernels are selected using the median rule discussed in Section~\ref{sec: Alg}. 

\subsection{Homeownership application: housing survey data}
We use the American Housing Survey Study dataset (publicly available at \texttt{https://github.com/afarahi/American-Housing-Survey-Study-}). The response variable, $Y$, is binary, indicating whether a household owns their home ($Y=1$) or not ($Y=0$). We select two continuous covariates:
\begin{enumerate}
\item \textit{Household income ($X_{\rm inc}$):} The total pre-tax annual income of the household. We simulate measurement error by adding Gaussian noise, reflecting the tendency of survey respondents to estimate or round income figures.
\item \textit{Householder age ($X_{\rm age}$):} The age of the reference person for the household. We simulate measurement error using a centred Poisson process, reflecting discrete asymmetric variations often found in demographic reporting.
\end{enumerate}
Before analysis, both the variables (Household Age and Income) were standardized to a zero mean and unit variance. In this case, $q_{\bftheta}(Y|\bfX)$ is a logistic regression model: 
\begin{equation*}
\pr\left(Y=1 \mid \bfX, \bftheta_{\mathrm{reg}}\right)=\frac{1}{1+\exp \left\{-\left(\alpha+\bfbeta\trans \bfX\right)\right\}}.
\end{equation*}
The ``True'' model is trained on the standardized, clean data. 
To simulate the \textsc{eivr} setting, noise was added to the standardized features as follows:
\begin{enumerate}
    \item Variable income: $U_{\rm Income} \sim r(\cdot \mid \phi ) = \calN(0,\phi^2)$, where $\phi = 0.8$. Effective variance $= 0.64$.
    \item Variable age:  $U_{\rm Age} \sim r(\cdot \mid \bfphi ) =  0.5 \cdot ({\rm Poisson}(\phi) - \phi)$, where $\phi = 3$. This formulation ensures zero-mean noise with an effective variance of $0.75$.
\end{enumerate}
We modelled the latent distribution of the true covariates $X_{\text{Age}}$ and $X_{\text{Income}}$ using independent \textsc{gmm} models with $J = 5$ components. To ensure stable convergence, we employed a warm-start initialization strategy. The parameters of the latent covariate model $q(X_{\text{Age}}|\bftheta_{X_{\text{Age}}})$ and $q(X_{\text{Income}}|\bftheta_{X_{\text{Income}}})$ were initialized by fitting a standard \textsc{gmm} to noisy $\wt{X}_{\text{Age}}$ and $\wt{X}_{\text{Income}}$ respectively. The regression parameters ($\bftheta_{\text{reg}}$) were initialized using the logistic regression model on the noisy data ($\wt{X}_{\text{Age}}$, $\wt{X}_{\text{Income}}$, $Y$).
For optimization, we adopted a mixture of Gaussian kernels with the bandwidths selected using the median rule discussed in Section~\ref{sec: Alg} and a fixed batch size of $M = 10000$ samples.

%% file: main_arxiv.bbl
\begin{thebibliography}{}

\bibitem[Biggs {\em et~al.}(2023)Biggs, Schrab, and Gretton]{biggs2023mmd}
Biggs, F., Schrab, A., and Gretton, A. (2023).
\newblock Mmd-fuse: Learning and combining kernels for two-sample testing
  without data splitting.
\newblock {\em Advances in Neural Information Processing Systems\/}, {\bf 36},
  75151--75188.

\bibitem[Briol {\em et~al.}(2019)Briol, Barp, Duncan, and
  Girolami]{briol2019statistical}
Briol, F.-X., Barp, A., Duncan, A.~B., and Girolami, M. (2019).
\newblock Statistical inference for generative models with maximum mean
  discrepancy.
\newblock {\em arXiv preprint arXiv:1906.05944\/}.

\bibitem[Cook and Stefanski(1994)Cook and Stefanski]{cook1994simulation}
Cook, J. and Stefanski, L. (1994).
\newblock Simulation-extrapolation estimation in parametric measurement error
  models.
\newblock {\em Journal of the American Statistical Association\/}, {\bf 89},
  1314--1328.

\bibitem[Davis(1990)Davis]{davis1990body}
Davis, C. (1990).
\newblock Body image and weight preoccupation: A comparison between exercising
  and non-exercising women.
\newblock {\em Appetite\/}, {\bf 15}, 13--21.

\bibitem[Farahi {\em et~al.}(2019)Farahi {\em et~al.}]{farahi2019mass}
Farahi, A. {\em et~al.} (2019).
\newblock Mass variance from archival x-ray properties of dark energy survey
  year-1 galaxy clusters.
\newblock {\em Monthly Notices of the Royal Astronomical Society\/}, {\bf 490},
  3341--3354.

\bibitem[Ferguson(2017)Ferguson]{ferguson2017course}
Ferguson, T.~S. (2017).
\newblock {\em A course in large sample theory\/}.
\newblock Routledge.

\bibitem[Fox {\em et~al.}(2018)Fox, Weisberg, and Price]{fox2018cardata}
Fox, J., Weisberg, S., and Price, B. (2018).
\newblock cardata: Companion to applied regression data sets. {R} package
  version 3.0-2.

\bibitem[Gretton {\em et~al.}(2012)Gretton, Borgwardt, Rasch, Sch{{\"o}}lkopf,
  and Smola]{gretton12a}
Gretton, A., Borgwardt, K.~M., Rasch, M.~J., Sch{{\"o}}lkopf, B., and Smola, A.
  (2012).
\newblock A kernel two-sample test.
\newblock {\em Journal of Machine Learning Research\/}, {\bf 13}, 723--773.

\bibitem[Kelly(2007)Kelly]{kelly2007some}
Kelly, B.~C. (2007).
\newblock Some aspects of measurement error in linear regression of
  astronomical data.
\newblock {\em The Astrophysical Journal\/}, {\bf 665}, 1489.

\bibitem[Mantz(2016)Mantz]{mantz2016lrgs}
Mantz, A.~B. (2016).
\newblock {LRGS}b: Linear regression by {Gibbs} sampling.
\newblock {\em Astrophysics Source Code Library\/}, pages ascl--1602.

\bibitem[Schrab(2025)Schrab]{schrab2025practical}
Schrab, A. (2025).
\newblock A practical introduction to kernel discrepancies: Mmd, hsic \& ksd.
\newblock {\em arXiv preprint arXiv:2503.04820\/}.

\bibitem[Sriperumbudur {\em et~al.}(2010)Sriperumbudur, Gretton, Fukumizu,
  Sch{\"o}lkopf, and Lanckriet]{sriperumbudur2010hilbert}
Sriperumbudur, B.~K., Gretton, A., Fukumizu, K., Sch{\"o}lkopf, B., and
  Lanckriet, G.~R. (2010).
\newblock Hilbert space embeddings and metrics on probability measures.
\newblock {\em Journal of Machine Learning Research\/}, {\bf 11}, 1517--1561.

\end{thebibliography}
